\documentclass[11pt]{article}
\usepackage[verbose=true,letterpaper]{geometry}
\AtBeginDocument{
  \newgeometry{
    textheight=9in,
    textwidth=6.5in,
    top=1in,
    headheight=14pt,
    headsep=25pt,
    footskip=30pt
  }
}
\usepackage{booktabs} 
\usepackage[ruled]{algorithm2e} 

\usepackage[english]{babel}
\usepackage[utf8x]{inputenc}
\usepackage{amsmath}
\usepackage{amsthm}
\usepackage{amsfonts}
\usepackage{algorithmic}
\usepackage{graphicx}
\usepackage{tabularx,ragged2e, array, stackengine}
\newcolumntype{P}[1]{>{\Centering}m{#1}}

\usepackage{hyperref}
\usepackage[nameinlink]{cleveref}

\newcommand{\append}{\textsc{Append}}
\newcommand{\writeOp}{\textsc{Write}}

\newcommand{\ttlAppend}{Q_A}

\newcommand{\rew}{\textsc{Reward}}

\newcommand{\down}{\textsc{DP}}
\newcommand{\eu}{\textsc{EU}}

\usepackage[dvipsnames]{xcolor}
\definecolor{DustyRose}{HTML}{DC6B82}

\usepackage{color-edits}

\addauthor{cz}{DustyRose}

\newtheorem{theorem}{Theorem}[section]
\newtheorem{corollary}[theorem]{Corollary}
\newtheorem{lemma}[theorem]{Lemma}
\newtheorem{proposition}[theorem]{Proposition}

\newtheorem{definition}[theorem]{Definition}

\newtheorem*{conjecture*}{Conjecture}
\newtheorem*{direction*}{Research Direction}

\newtheoremstyle{nonindented}{1ex}{1ex}{}{}{\bfseries}{.}{.5em}{}
\newtheoremstyle{indented}{1ex}{1ex}{\itshape\addtolength{\leftskip}{0.6cm}\addtolength{\rightskip}{0.6cm}}{}{\bfseries}{.}{.5em}{}
\theoremstyle{nonindented}
\theoremstyle{indented}
\theoremstyle{plain}

\newenvironment{alignedequation*}{\begin{equation*} \begin{aligned}}{\end{aligned} \end{equation*}}

\def\min{\qopname\relax n{min}}
\def\max{\qopname\relax n{max}}

\newcommand{\lt}{\left}
\newcommand{\rt}{\right}

\title{Analyzing the Economic Impact of Decentralization on Users}

\author{Amit Levy 
\thanks{Better Bytes \& Princeton University, \tt{amit@amitlevy.com}.} \and S. Matthew Weinberg 
\thanks{Princeton University, \tt{smweinberg@princeton.edu}. Supported by NSF CAREER Award CCF-1942497 and an Ethereum Foundation Academic Grant.} \and Chenghan Zhou \thanks{Stanford University, \tt{chzhou@stanford.edu}.}} 
\date{}
\begin{document}
\maketitle
\begin{abstract}
   We model the ultimate price paid by users of a decentralized ledger as resulting from a two-stage game where Miners (/Proposers/etc.) first purchase blockspace via a Tullock contest, and then price that space to users. When analyzing our distributed ledger model, we find:
\begin{itemize}
    \item A characterization of all possible pure equilibria (although pure equilibria are not guaranteed to exist). 
    \item A natural sufficient condition, implied by Regularity (\`{a} la~\cite{Myerson81}), for existence of a ``market-clearing'' pure equilibrium where Miners choose to sell all space allocated by the Distributed Ledger Protocol, and that this equilibrium is unique.
    \item The \emph{market share of the largest miner} is the relevant ``measure of decentralization'' to determine whether a market-clearing pure equilibrium exists.
    \item Block rewards do not impact users' prices at equilibrium, when pure equilibria exist. But, higher block rewards can cause pure equilibria to exist.
\end{itemize}

We also discuss aspects of our model and how they relate to blockchains deployed in practice. For example, only ``patient'' users (who are happy for their transactions to enter the blockchain under any miner) would enjoy the conclusions highlighted by our model, whereas ``impatient'' users (who are interested only for their transaction to be included in the very next block) still face monopoly pricing.
\end{abstract}


\maketitle

\begingroup
\renewcommand\thefootnote{}
\footnotetext{\textbf{Acknowledgments.}~The authors are extremely grateful to Scott Kominers and Jacob Leshno for multiple exceptionally useful conversations that helped us improve this work.}
\addtocounter{footnote}{-1}
\endgroup

\addtocounter{page}{-1}

\thispagestyle{empty}

\setcounter{tocdepth}{2} 
\newpage
\section{Introduction}

Following Nakamoto's creation of Bitcoin in 2008~\cite{Nakamoto08}, adoption of blockchain technology for various purposes has steadily grown.\footnote{For example, Forbes reports a cryptocurrency market cap of \$3.27T USD at time of writing. Source: https://www.forbes.com/digital-assets/crypto-prices/} More relevant to this paper is ongoing interest in so-called ``Web3'' or ``Decentralized Apps'', for which an estimated \$5.4B USD in VC funding was raised in 2024.\footnote{Source: https://cointelegraph.com/news/vc-roundup-web3-funding-5-4-billion-2024 . Crunchbase further estimates a cumulative \$111 B USD in VC funding raised for Web3: https://news.crunchbase.com/web3-startups-investors/.} This paper seeks contributions to a theoretical foundation for why users might (or might not) ultimately find value in decentralized services in comparison to centralized alternatives.

\vspace{0.1cm}
\noindent\textbf{Classic vs.~Modern Pitches for Decentralized Services.} Aside from financial speculation, perhaps the dominant `real' use case of blockchain technology is as a currency for users with no viable alternative. While compelling applications, the economic case for such users is relatively straight-forward because the competing product\footnote{For example: a currency likely to be frozen by an authoritarian government, a hyperinflating currency, or a currency that can be tracked by law enforcement critical of your illicit activity} is so dysfunctional that concerns about (say) Bitcoin's transaction fees, volatility, and UI become very second-order. A classical pitch for decentralization therefore emphasizes simply that decentralized services make it more challenging for authoritarian leaders (and law enforcement) to deny access, and this pitch is plenty convincing in comparison to the (functionally non-existent) alternatives.

A modern discussion on blockchain technology, however, includes applications targeting users in developed economies with highly developed alternatives. For example, the pitch for stablecoins to users with a hyper-inflating local currency looks very different than to users with access to Venmo, Paypal, and credit cards. Consider also decentralized services such as file storage (for which centralized services such as Dropbox are a reasonable substitute), social networks (for which centralized services such as Facebook or Twitter are a reasonable substitute), or gaming (for which centralized gaming services produced by Riot or Blizzard are a reasonable substitute) -- what would cause users to prefer decentralized services over highly-developed centralized alternatives? 

A natural answer is that perhaps the decentralized service might somehow be `better' than the centralized competitor.\footnote{See here for an example of this pitch: https://a16zcrypto.com/posts/article/how-stablecoins-will-eat-payments/ .} But it is initially confusing how that might possibly arise as centralized services can optimally coordinate to lower internal costs, whereas decentralized services must additionally manage incentives/trust across distributed entities.

\vspace{0.1cm}
\noindent\textbf{Decentralizing Natural Monopolies.} One well-understood source of inefficiency in centralized services is deadweight loss caused by a monopolist.\footnote{The holdup problem is another -- see Section~\ref{sec:related} for a brief discussion.} That is, a decentralized service might plausibly be desirable to a highly-developed centralized alternative simply because the decentralized service results in different prices, and this can still be the case even if the centralized infrastructure is more efficient. Therefore it is natural to target domains with a ``natural monopoly'' aspect (such as social networks, payment systems, marketplaces, etc.).

Indeed, independently of any blockchain discussions,~\cite{Tirole23} highlights natural monopolies for digital services as a growing challenge, and further poses several possible approaches (each with drawbacks). One approach is described as follows: ``An alternative approach to full-scale regulation consists in insulating a natural monopoly (or bottleneck or essential facility) segment, as became popular in the late twentieth century. This segment remains regulated and is constrained to provide a fair and nondiscriminatory access to competitors in segments that do not exhibit natural monopoly characteristics and therefore can sustain competition.''\footnote{\cite{Tirole23} cites several examples: electricity markets might insulate the natural monopoly (the grid) and enable open competition on generation, or rail travel might insulate the natural monopoly (the tracks/stations) and enable open competition on train operation. A prevalent digital example is Local Loop Unbundling, where many countries insulate the natural monopoly (the ``local loop'' -- physical copper wires servicing telecommunications) by requiring its owners to lease access at nondiscriminatory prices to service providers.} One of two key drawbacks of this approach is as follows: ``\ldots one wants to break up the incumbent without destroying the benefits of network externalities. For example, breaking a social network into two or three social networks might not raise welfare.''\footnote{The second key challenge highlighted is identifying a core bottleneck to insulate. We argue in Section~\ref{sec:conclusion} that for many domains of interest, such a bottleneck can be identified and (in theory, at least) insulated. A final challenge highlighted is the actual process of unbundling an existing product, which is unrelated to our work.}

 One interpretation of Bitcoin is exactly through this lens, and~\cite{HubermanLM21} are the first to make this point.\footnote{``We model this novel economic structure and show that the BPS’s [Bitcoin Payment System's] decentralized design offers a prototype of a payment system in which users are protected from monopoly harm even if the payment system were a monopoly\ldots Standard economic arguments suggest that weak competition among monopolistic firms calls for regulation to mitigate monopoly harm. Under the BPS, users are protected from abuses of monopoly power even without competition from other payment systems. Thus, the BPS addresses potential antitrust concerns in a novel, even revolutionary, way.''} Indeed, when viewed as a payment system, the natural monopoly segment is `ledger maintenance' (where a consistent record of transactions is maintained), while the user-facing `transaction processing' segment is not a natural monopoly.  In the language of Bitcoin, substantial network effects arise from having consensus on a single consistent ledger, but minimal network effects arise from users transacting within the same block (or with the same miner). In the language of payment systems, substantial network effects arise from users transacting in the same `currency', but minimal network effects arise from users using the same app to process those transactions.\footnote{The preceding sentence is necessarily clunky, as it is somewhat unnatural to imagine separating a centralized payment system into a back-end transaction processor (that stores data and moves money around, where network effects arise) and front-end transaction processor (that interfaces directly with users, and minimal network effects arise). One high-level contribution of the ``Decentralized view'' is as a lens to dis-integrate services without an obvious dis-integration.} Perhaps shockingly, this insulation is maintained without regulation,\footnote{Our analysis does rely on Miners treating core aspects of the consensus protocol as exogenous, which bears conceptual similarity to regulation.} and therefore provides a novel approach to insulating natural monopolies.

\vspace{0.1cm}
\noindent\textbf{But does Decentralization Actually Help?} The preceding paragraph highlights blockchain-style decentralization an innovative approach to insulate natural monopolies from derivative services, but should we expect users to ultimately be better off? How would the answer depend on market primitives? Moreover, how do we even draw conclusions on users' utility from decentralized services? Surprisingly few answers are known to questions like these, and surprisingly fewer frameworks are known to even approach them. The goal of this paper is therefore to provide a framework towards such questions in the core domain of distributed ledgers, with an emphasis on connecting users' ultimate utility to properties of the decentralized ledger.

\subsection{Overview of Results}
We consider the core setting of a ledger. Ultimately, users desire the service of writing their transaction to the ledger, and have some value for doing so. Inspired by the preceding discussion, we separate this service into an Upstream segment which is a natural monopoly, and a Downstream segment which is not.\footnote{See Section~\ref{sec:prelim} for further discussion. Intuitively, the Upstream segment directly edits the ledger, which is a natural monopoly due to network effects of multiple users sharing access to the same ledger. The Downstream segment directly interfaces with users to solicit their transactions and pass to the Upstream segment, and exhibits minimal network effects.}

A centralized ledger would simply provide the entire service in a vertically integrated manner.\footnote{Venmo is a good example to have in mind for this model -- the Venmo backend is the Upstream segment, and the Venmo app is the Downstream segment. Users enjoy network effects due to the backend database, and minimal network effects from opening the same app on their phones.} A classic Industrial Organization exercise might consider dis-integrating the Upstream monopolist from separate Downstream firms that compete with one another.\footnote{The authors are not aware of a live example matching this model. A hypothetical example to have in mind would be if Venmo allowed third-party apps to access its ledger, and charged access fees to those apps.} A distributed ledger removes centralized control entirely -- the Upstream segment is provided by a hard-coded Protocol with exogenously set parameters. Competitive Downstream providers then `purchase' the Upstream resource according to the rules of the protocol and use it to process users' transactions.\footnote{In the language of Bitcoin, miners are Downstream providers. Miners solicit transactions from users, and `purchase' the right to include transactions in a Bitcoin block by `paying' in hashes.} We analyze this in Section~\ref{sec:bitcoin}, and in particular include a discussion of why the model captures key aspects of distributed ledgers.

We then draw the following conclusions in our model:
\begin{itemize}
    \item To the extent that a quantitative ``measure of decentralization'' impacts the price faced by users, it is the \emph{size of the largest miner}.\footnote{See Theorems~\ref{thm:bitcoinpure} and~\ref{thm:bitcoinsufficient} for precise statements.} 
    \item Block rewards have limited impact on users. Specifically, block rewards cannot impact the ultimate price users would face in equilibrium (provided equilibria exist),\footnote{See Theorem~\ref{thm:bitcoinpure} for a precise statement.} but can cause equilibria to exist.\footnote{See Propositions~\ref{prop:blockrewards} and~\ref{prop:blockrewards2} for precise statements.}
    \item We also characterize all possible equilibria (Theorem~\ref{thm:bitcoinpure}) and provide sufficient conditions for desirable equilibria to exist (Theorem~\ref{thm:bitcoinsufficient}).
    \item Our model applies only to \emph{patient} users (who are happy to have their transaction included in any block), whereas \emph{impatient} users (who want their transaction included in the next block or not at all) instead face miners with monopoly power over the contents of that block.
\end{itemize}

\subsection{Roadmap}
In Section~\ref{sec:related}, we discuss related work. Section~\ref{sec:prelim} overviews our model. Section~\ref{sec:enduser} provides technical preliminaries.\footnote{Due to space constraints, we prioritize presentation of our model, statements of results, and implications in the body. Section~\ref{sec:enduser} is useful primarily for technical intuition, and so is moved to the Appendix due to space constraints.} Section~\ref{sec:bitcoin} describes our Distributed Ledger Model, highlights key distinctions to a Centralized provider, highlights its connection to distributed ledgers in practice, and provides our main analysis. Section~\ref{sec:conclusion} concludes. 

\subsection{Related Work} \label{sec:related}

\noindent\textbf{Modeling Economic Impact of Decentralized Technologies.} The most closely related works in terms of motivation also seek to understand potential economic benefits of aspects of decentralized technologies (although there is no technical overlap between our work and any of these). By far the most related in terms of motivation is~\cite{HubermanLM21}, who also view distributed ledgers through the lens of insulating a natural monopoly.~\cite{HubermanLM21} considers users with a simple value for service (either High or Low), and who prefer not to wait for their transactions to be included. In their model, a monopolist excludes all Low users, but immediately processes all High transactions, causing deadweight loss. Bitcoin, on the other hand, processes all users, but with delay cost.\footnote{\cite{HubermanLM21} further analyze the delay as a function of Bitcoin protocol parameters.} So their work highlights a tradeoff between a monopolist (deadweight loss) and Bitcoin (delay cost). In comparison, our work (a) focuses exclusively on the monetary cost paid by users, (b) considers a richer model of user preferences (i.e.~an arbitrary demand curve, sometimes subject to a standard regularity condition), and (c) exclusively studies users of a decentralized ledger and does not explicitly compare to a monopolist.

Other works analyze the economic impact of aspects of decentralized technologies from an orthogonal viewpoint. For example,~\cite{SockinX23,Reuter24} view decentralization/tokenization as a commitment device by which a platform can cede control to users. In addition,~\cite{GoldsteinGS24} similarly view tokenization as a commitment device to future competitive pricing. These works address a similar high-level challenge (platforms with network effects), and also through novel approaches that arose recently alongside blockchain technology. However, these works still involve a rent-seeking platform (in comparison to our exogenous protocol), and cede control to users or external investors (in comparison to changing the market structure).

\vspace{0.1cm}
\noindent\textbf{Tullock Contests.} At a technical level, our work studies equilibria of a two-part game, one of which is a Tullock Contest and the second of which is an auction (see Section~\ref{sec:bitcoin} for a precise specification). As such, much of our technical analysis concerns Tullock Contests~\cite{Tullock80, BuchananTT80,HillmanR89, Gradstein95}, which are commonly used to capture the game played by Bitcoin miners to produce blocks (and also to capture related aspects of blockchain ecosystems)~\cite{ArnostiW22,AlsabahC20,Dimitri17, BahraniGT24}. The key technical distinction between our work and these works lies in our second-stage auction game, which will become clear in Section~\ref{sec:bitcoin}. 

\vspace{0.1cm}
\noindent\textbf{Industrial Organization Theory.} Our model is inspired by `textbook' Industrial Organization Theory models~\cite{tirole1988theory, SchmalenseeAW89}. Our model focuses on textbook settings (without demand uncertainty, and without costly marketing) to isolate the impact of the novel blockchain-inspired market structure -- additional aspects may be fruitful to consider as future work develops. In classical language, we model downstream producers that sell identical products (because the users are patient, and therefore indifferent to which block they get in). Impatient users (which are not the focus of our work, as they simply face a downstream monopolist) would instead be captured by perfectly differentiated downstream products (because impatient users want only to enter the next block).

\vspace{0.1cm}
\noindent\textbf{Other Economic Aspects of Blockchains.} Numerous other works consider economic aspects of blockchains. Several consider the economic incentives of protocol participants~\cite{EyalS14, Eyal15, CarlstenKWN16, KiayiasKKT16, SapirshteinSZ16, BrownCohenNPW19, FiatKKP19, GorenS19, NeuderMRP20, NeuderMRP21, Budish22, NeuTT22, FerreiraHWY22,YaishTZ22, Newman23, YaishSZ23, BarZurET20, AlpturerW24, BahraniW24, CaiLWZ24,FerreiraGHHWY24, BudishLR24, Nagy2025forking, BahraniNW25}. These works uncover reasons why participants may not be incentivized to follow the protocol specifications. In comparison to these works, we assume the underlying blockchain protocol functions as intended. Several consider ``transaction fee mechanism design'' -- the auction specified by the protocol for users to purchase transactions from miners~\cite{LaviSZ19, Roughgarden21, ShiCW23, WuSC24, Yao18, GafniY24, FerreiraMPS21, ChungS23, ChungRS24, ChenSZZ24, BasuEOS19, LeonardosMRSP21, LeonardosRMP23,GaneshTW24, GaneshTW26}. We model miners running a first-price auction with reserve, and discuss briefly in Section~\ref{sec:bitcoin} the connection between our modeling decision and blockchains with alternate TFMs (such as Ethereum's EIP-1559).\footnote{Briefly, what really matters for our model is the cost of including a transaction on-chain (which in EIP-1559 is the base fee, and in Bitcoin is zero), and how a profit-maximizing miner would choose to sell block space (given that cost) to users who can choose to instead purchase from other miners.} Finally,~\cite{Nisan23} considers the pricing dynamics of serial monopolists selling blockspace to patient buyers. In comparison to our work,~\cite{Nisan23} considers Miners who produce only a single block and aim to maximize their revenue from that block in isolation, whereas our work considers Miners who aim to optimize their joint revenue from multiple blocks (and also models the Tullock contest by which Miners earn the right to produce those blocks).

\section{Preliminaries}\label{sec:prelim}
\noindent\textbf{Running Story.} Our model is motivated by the concept of a ledger. Ultimately, the product consumed by an end-user is the ability to write information to the global ledger (which we call a \writeOp{}).\footnote{In order to focus on the relevant market primitives, we do not explicitly model ledger maintenance, consensus, cryptography, privacy, or reading. The service purchased by an end-user gets their message onto the ledger, and in a manner that can be read by the desired recipients.} That is, each end-user has a message they would like to write on the ledger, and purchases a \writeOp{} to do so. 

The entire value proposition of a global ledger is that there is ultimately a single consistent ledger. It is therefore crucial that \emph{some} aspects of ledger maintenance are performed via a single entity/protocol/etc.~(for example, centralized ledgers should maintain a single consistent back-end database. Decentralized ledgers should have a single protocol from which observers can conclude a single consistent ledger). Intuitively, these are operations that directly edit content in the ledger (and because there is a single consistent ledger, these operations must be carefully coordinated by a single entity/protocol/etc.). Other aspects of ledger maintenance can in principle be performed by competing entities (for example, end-users can in principle face different User Interfaces, pricing schemes, etc.). We abstract away precise details of the ledger maintenance process, and simply refer to operations that directly edit the single consistent ledger as \emph{Upstream} (and refer to one unit of these operations as an \append{}), and those that could in principle be performed by competing entities \emph{Downstream}. 

It may help to have a few examples in mind. Imagine breaking a centralized ledger (i.e.~Venmo) into its back-end database maintenance and front-end User Interface. The back-end database must ensure consistency on a single global ledger, and so is Upstream. Edits to the back-end database must be reliable and consistent (even if the database is replicated, distributed, etc.). One \append{} constitutes the resources necessary to add one entry to the back-end database (maintaing consistency, availability, etc.). The front-end User Interface is Downstream -- the front-end UI interacts directly with consumers, and turns communication with end-users into a query to the Upstream back-end database. The front-end UI consumes \append{}s in order to produce \writeOp{}s, and sells \writeOp{}s to users. Note that, in principle, the centralized ledger could offer different front-end UIs to different consumers (with different pricing schemes, different communication protocols, different app layout, etc.) -- doing so does not in principle interfere with the ability to maintain a single, consistent back-end database.

One could imagine instead a centralized back-end database that provides an API for third-party app access. The centralized back-end database again is a producer of \append{}s. Each third-party app is a consumer of \append{}s and a producer of \writeOp{}s. End-users purchase \writeOp{}s from a third-party app (who incurs costs both from interacting with the end-user, and from purchasing \append{}s). Again, each third-party app could in principle differ in pricing schemes, communication protocols, app layouts, etc., and purchase \append{}s from the same back-end database (that interacts with each third-party app in a manner that maintains a single consistent ledger).

One could also imagine a decentralized consensus protocol maintaining a decentralized ledger, allowing participation from ``miners'', ``stakers'', ``proposers'', etc.\footnote{Throughout this paper, we adopt the language of Bitcoin and refer to these participants as miners.} The decentralized protocol outlines a costly procedure by which miners receive \append{}s.\footnote{For example, Bitcoin miners receive \append{}s by repeated hashing (which costs electricity and hardware). Ethereum stakers receive \append{}s by locking up ETH in the Ethereum protocol (which costs capital).} Each miner is a consumer of \append{}s (that they ``purchase'' by completing the costly procedure specified in the decentralized protocol), and a producer of \writeOp{}s. End-users purchase \writeOp{}s from a miner (who incurs costs both from interacting with the end-user, and ``purchasing'' an \append{}). The consensus protocol structures its ``sale'' of \append{}s so that all ledger updates contribute to a single consistent ledger.\footnote{That is, this paper assumes that the consensus protocol functions as intended. See Section~\ref{sec:related} for a brief discussion on on related work surrounding this assumption.}

The subsequent paragraphs formalize our model in the abstract -- the running story provides intuition for each concept.

\vspace{0.1cm}
\noindent\textbf{Market Resources.} We consider two types of resources. The Downstream resource, \writeOp{}, is consumed by end-users. The Upstream resource, \append{}, is required to produce \writeOp{}s (in one-to-one ratio). Production of the Upstream resource is a natural monopoly, and therefore will be produced by a single entity/protocol. Production of the Downstream resource is not a natural monopoly, therefore we model a Downstream market with multiple competing participants. 

\vspace{0.1cm}
\noindent\textbf{Market Participants.} There are two types of market participants. End-users are the ultimate consumers, who desire \writeOp{}s. Each end-user wants a single \writeOp{}, and has some value $v$ should they receive one. Downstream producers produce \writeOp{}s, which necessitates consumption of \append{}s. The protocol (which is hard-coded and has no objective function or strategic decisions) produces \append{}s. 

\vspace{0.1cm}
\noindent\textbf{Market Primitives.} There is a continuum of end-users, with $D(p)$ denoting the mass of consumers with value at least $p$ for a \writeOp{}. 

We assume that $D(\cdot)$ provides finite revenue to a monopolist (that is, $\sup_p\{p \cdot D(p)\} < \infty$). Our main results require a standard regularity assumption on $D(\cdot)$.

\begin{definition}[Regular] A demand curve $D(\cdot)$ is \emph{Regular} if:
\begin{itemize}
    \item $D(\cdot)$ is differentiable and strictly decreasing. In this case, we use $-d(\cdot):=D'(\cdot)$. 
    \item The function $\varphi_D(x):=x - \frac{D(x)}{d(x)}$ is monotone non-decreasing in $x$.
\end{itemize}
\end{definition}

\noindent\textbf{Structure of the Game.} We model the interactions between the Upstream protocol, Downstream providers, and End-Users as a three-stage game. First, the Upstream protocol sets the dynamics for selling \append{}s to Downstream providers. Next, with this protocol fixed, Downstream providers set their strategies both for purchasing \append{}s and for selling \writeOp{}s to End-Users. Finally, with these strategies fixed, End-Users set their strategies for purchasing \writeOp{}s from Downstream providers. 

\vspace{0.1cm}
\noindent\textbf{Equilibrium Analysis.} Let $\down$ denote the set of Downstream providers, $\eu$ denote the set of End-Users, and $P_i(\vec{a})$ denote the payoff to Player $i$ when the action profile is $\vec{a}$. 

An End-User Equilibrium fixes some actions $\vec{a}_{\down}$ by the Downstream providers, and is a Nash Equilibrium of the End-User game induced by $\vec{a}_{\down}$ (with payoff $P_i(\vec{a}_{\down};\vec{a}_{\eu})$ to Player $i \in \eu$ on action profile $\vec{a}_{\eu}$).

A Downstream Equilibrium\footnote{We will sometimes simply call this an Equilibrium.} specifies, for each possible action profile $\vec{a}_{\down}$ of the Downstream providers, an End-User Equilibrium $E(\vec{a}_{\down})$ for the end-user game induced by $\vec{a}_{\down}$, and then (together with $E(\cdot)$) is a Nash Equilibrium among Downstream providers for the Downstream game induced by $E(\cdot)$ (which awards payoff $P_i(\vec{a}_{\down}, E(\vec{a}_{\down}))$ to Player $i \in \down$ on action profile $\vec{a}_{\down}$). When $E(\cdot)$ is unique (or otherwise clear from context), we will abuse notation and simply refer to $\vec{\alpha}_{\down}$ as a Downstream Equilibrium. Moreover, we will also abuse notation and say that a Downstream Strategy $\alpha_i$ dominates $\alpha'_i$ if $\alpha_i$ dominates $\alpha'_i$ in the game among Downstream providers induced by $E(\cdot)$.

\vspace{0.1cm}
\noindent\textbf{Notation.} For a (not necessarily continuous) monotone non-increasing function $F(\cdot)$, we let $F^{-1}(y):=\{x\ |\ \lim_{z \rightarrow x^+} F(z) \leq y \leq \lim_{z \rightarrow x^-} F(z)\}$, $F_{\inf}^{-1}(y)$ denote the infimum of $F^{-1}(y)$, and $F^{-1}_{\sup}(y)$ denote the supremum of $F^{-1}(y)$. Observe that if $F(\cdot)$ is left-continuous, then $F(x) \geq y$ for any $x \in F^{-1}(y)$.\footnote{Because for any $x \in F^{-1}(y)$, $y \leq \lim_{z \rightarrow x^-} F(z) = F(x)$.} If $F(\cdot)$ is continuous and strictly decreasing, we simplify notation and define $F^{-1}(y):=F_{\inf}^{-1}(y) = F_{\sup}^{-1}(y)$.

\vspace{0.1cm}
\noindent\textbf{First-Price Auction with Reserve.} First-Price Auctions with Reserves are a common subgame in our market structures. With a continuum of bidders and a total supply of $Q$, a first-price auction with reserve $r$ concludes as follows. First, let $B(q)$ denote the mass of bidders who submit a bid at least as large as $q$.\footnote{Observe that $B(\cdot)$ is left-continuous. To see this, observe that all bidders who bid at least $b-\varepsilon$ contribute to $B(b-\varepsilon)$. So the bidders that contribute to $B(b-\varepsilon)$ for all $\varepsilon > 0$ are exactly those who bid at least $b$ -- the same bidders that contribute to $B(b)$.} Next, if $B(r) < Q$, then every bidder who submits a bid at least as large as $r$ wins and pays their bid. If $B(r) \geq Q$, then every bidder who submits a bid strictly exceeding $B^{-1}_{\sup}(Q)$ wins and pays their bid, every bidder who submits a bid strictly below $B^{-1}_{\sup}(Q)$ loses, a mass of $B(B^{-1}_{\sup}(Q)) - Q$ bidders who submit a bid of exactly $B^{-1}_{\sup}(Q)$ lose and the remainder win and pay their bid (and in this case a total mass of $Q$ bidders win).\footnote{Recall that $B(B^{-1}_{\sup}(Q)) - Q \geq 0$ as $B$ is left-continuous.} \footnote{All of our analysis holds no matter how ties are broken to select the winning bidders among those who bid $B(B^{-1}_{\sup}(Q))$.} Observe that every bid profile induces an effective price of $b:=\max\{B^{-1}_{\sup}(Q), r\}$ -- every bidder who submits a bid exceeding $b$ certainly wins (and pays their bid) and every bidder who submits a bid below $b$ certainly loses. 

\vspace{0.1cm}
\noindent\textbf{Equilibria of Simultaneous First-Price Auctions.} Simultaneous First-Price Auctions are another common subgame in our market structures. Below we overview Simultaneous First-Price Auctions and technical lemmas helpful to understand our results -- full analyses and proofs are in~\Cref{sec:enduser}.

\begin{definition}[Simultaneous First-Price Auctions] In Simultaneous First-Price Auctions, there are $n$ sellers. Each seller $i$ has a $Q_i$ mass of items for sale, and sets reserve $r_i$. We define $Q^{\leq}(r):=\sum_{i,\ r_i \leq r} Q_i$ to be the total mass of items for sale at reserve at most $r$,\footnote{Observe that $Q^{\leq r}$ is monotone non-decreasing, and right-continuous everywhere. To see this, observe that seller $i$ contributes $Q_i$ to $Q^{\leq}(r)$ if and only if $r_i \leq r$, and to $Q^{\leq}(r+\varepsilon)$ for all $\varepsilon > 0$ if and only if $r_i \leq r$. Therefore all sellers contribute the same to both $Q^{\leq}(r)$ and $\lim_{\varepsilon \rightarrow 0} Q^{\leq}(r+\varepsilon)$.} $Q^{<}(r):=\sum_{i,\ r_i < r} Q_i$ to be the total mass of items for sale at reserve strictly less than $r$,\footnote{Observe that $Q^{<}(r)$ is monotone non-decreasing, and left-continuous everywhere. To see this, observe that seller $i$ contributes to $Q^{<}(r)$ if and only if $r_i < r$, and to $Q^{<}(r-\varepsilon)$ for some $\varepsilon > 0$ if and only if $r_i < r$. Therefore, all sellers contribute the same to both $Q^{<}(r)$ and $\lim_{\varepsilon \rightarrow 0}Q^{<}(r-\varepsilon)$.} and $Q^{=}(r):=Q^{\leq}(r) - Q^{<}(r)$ to be the total mass of items for sale at reserve exactly $r$. 

A continuum of unit-demand buyers each submit a (possibly $0$) bid to each first-price auction. Each first-price auction executes exactly as defined in Section~\ref{sec:prelim}. An equilibrium of Simultaneous First-Price Auctions is simply a strategy profile where each bidder best responds. 
\end{definition}

For equilibria among bidders, fixing all $Q_i, r_i$,~\Cref{lem:firstpriceeq} establishes that all winning bidders pay the same price in equilibrium, and define the value of this as \emph{clearing price}. 

\begin{definition}[Clearing Price and Canonical Equilibrium] \label{def:fpa-equilibrium} For an equilibrium $E$ of Simultaneous First-Price Auctions, we refer to its \emph{clearing price} $c(E)$ as the bid promised by Lemma~\ref{lem:firstpriceeq} such that every winning bidder wins exactly one item at bid $c(E)$. We say an equilibrium is \emph{canonical} if (i) a total supply of $\min\{D(c(E)),Q^{\leq}(c(E))\}$ items are sold,\footnote{That is, there does not exist a non-zero mass of buyers with value $D(c(E))$ and an unsaturated auction with reserve $c(E)$.} and (ii) the clearing price is minimal across all equilibria.\footnote{Intuitively, what happens is the following. There is a demand curve $D(\cdot)$ defined by the users' demand. The sellers $\{(Q_i, r_i)\}$ jointly define a supply curve, with $S(q)$ denoting the quantity of \writeOp{}s sold in some auction with reserve at most $q$. The supply curve has jump discontinuities, and so the demand may ``meet'' supply in a discontinuity. Still, any point where supply meets demand can be the effective price, and if $D(\cdot)$ is continuous and strictly decreasing there is a unique such point.} Note that if $D(\cdot)$ is continuous and strictly decreasing, all equilibria are canonical.
\end{definition}

Downstream Equilibria in our market structures concern the behavior of sellers in Simultaneous First-Price Auctions (i.e. choosing a quantity $Q_i$ according by participating in the Upstream protocol, and setting a reserve $r_i$). We refer to the reserve-setting aspect as a \emph{price-setting equilibrium} (noting that a Downstream Equilibrium must both induce a price-setting equilibrium for fixed $\vec{Q}$, \emph{and} be a joint equilibrium when considering both investment and reserves). 

\begin{definition}[Price-Setting Game]\label{def:price-setting-game} A \emph{Price-Setting Game} has the following structure:
\begin{itemize}
    \item[] \textbf{Players.} There are $n > 0$ sellers. Seller $i$ has quantity $Q_i$ of items.
    \item[] \textbf{Action Space.} Each seller picks a reserve $r_i$ to set in a first-price auction.
    \item[] \textbf{Costs.} Each seller pays a cost of $c_i$ per item sold.
    \item[] \textbf{Payoffs.} On strategy profile $\vec{r}$, a continuum of buyers with values according to $D(\cdot)$ bids in equilibrium of the simultaneous first-price auctions with quantities $\vec{Q}$ and reserves $\vec{r}$, which induces a clearing price of $p(\vec{Q},\vec{r})$. If Seller $i$ sells a mass of $Q'_i$ items in this equilibrium, their payoff is $Q'_i \cdot (p(\vec{Q},\vec{r})-c_i)$. 
\end{itemize}

Note that if $D(\cdot)$ is continuous and strictly decreasing, the clearing price $p(\vec{Q},\vec{r})$ is unique. We refer to a Price-Setting Game as canonical if the equilibrium selected by buyers is canonical.
\end{definition}

A concept throughout our analyses is whether a seller clears their entire inventory, and whether they determine the price at which the bidding equilibrium clears. 

\begin{definition}\label{def:saturate}
We say that Seller $i$ is \emph{saturated} in a strategy profile $\vec{r}$ if either (i) Seller $i$ sells a mass of $Q_i$ items or (ii) the clearing price $p(\vec{Q},\vec{r}) \leq c_i$. We further say that an equilibrium is saturated if all sellers are saturated. Finally, we refer to Seller $i$ as a \emph{price-setter} in the strategy profile $\vec{r}$ if $Q_i > 0$ and the clearing price $p(\vec{Q},\vec{r}) = r_i > c_i$.
\end{definition}

Finally, Proposition~\ref{prop:maineq} characterizes that all potential price-setting equilibria are either a ``market-clearing equilibrium'' where no seller is sufficiently large to profit from price-setting, or have a unique price-setter.

\begin{proposition}\label{prop:maineq}
Let $Q:=\sum_i Q_i$. Then every price-setting equilibrium takes one of two forms:
\begin{itemize}
    \item Every seller is saturated and the clearing price is $D^{-1}_{\sup}(Q)$. An equilibrium of this form exists if and only if $Q_i \leq \frac{(x - c_i) (Q-D(x))}{x - D^{-1}_{\sup}(Q)}$ for all $i$ and all $x > D^{-1}_{\sup}(Q)$. 
    \item There is a single price-setter $i$, who sets price $r^*_i:= \arg\max_{x \geq D^{-1}_{\sup}(Q)}\{(x - c_i) (D(x)+Q_i-Q)\}$. If an equilibrium of this form exists, it certainly exists with $i^* := \arg \max_i\{r^*_i\}$ as the price-setter (but equilibria with other price-setters are possible). Moreover, if $c_i = c_j$ for all $i,j$, then $\arg\max_i\{Q_i\} = i^*$.
\end{itemize}
\end{proposition}
\section{Distributed Ledger Model}\label{sec:bitcoin}

We now formally present our Distributed Ledger model. After formally specifying the model, we overview key differences to classical market structures, and its connection to distributed ledgers in practice.

\begin{definition}[Distributed Ledger Model] The \emph{Distributed Ledger Model} is defined below. We refer to the Upstream provider as Protocol and to the Downstream providers as \emph{Miners}.

\textsc{Upstream}
\vspace{-0.15cm}
\begin{itemize}
    \item[] \textbf{Protocol.} The Upstream protocol produces a fixed amount of $\ttlAppend$ \append{}s and maintains consensus. The protocol runs a \emph{Tullock Contest}~\cite{Tullock80} in some Resource to distribute the supply of \append{}s, with a block reward $B\geq 0$. Specifically, Miner $i$ who invests $q_i$ of Resource receives a fraction of the total \append{} supply proportional to their investment (i.e., $\ttlAppend \cdot \frac{q_i}{\sum_j q_j}$). Miner $i$ also receives a block reward (i.e~payment) of $B \cdot \ttlAppend \cdot \frac{q_i}{\sum_j q_j}$.
    \item[] \textbf{Payoffs.} Upstream protocol has no payoffs -- it simply maintains ledger consensus.\footnote{In the case of Bitcoin (and Ethereum, and most other Decentralized Ledgers), Miners are also participants in a consensus protocol. It may be helpful to think of Upstream providers as nodes that pass messages, verify authenticity, etc. in roles that would not also result in the ability to sometimes dictate contents of a block.} \footnote{We have intentionally modeled the decisions of the consensus protocol as exogenous to the game we study. Of course, \emph{someone} decides on $Q_A$, $B$, and to use a Tullock Contest in the first place. In practice, these decisions happen on a \emph{much} slower time scale than the game we model. For example, both Bitcoin and Ethereum (and all permissionless blockchains the authors are aware of) have used Tullock Contests since their creation. Bitcoin has never changed the formula for its block reward, and Bitcoin has technically not changed its blocksize either (although `soft forks' have occasionally increased Bitcoin's functional blocksize). Still, it is also worthwhile for future work to study the processes by which protocol parameters are set.}
\end{itemize}

\textsc{Downstream}
\vspace{-0.15cm}
\begin{itemize}
    \item[] \textbf{Players. } There is a set of $n$ Miners. 
    \item[] \textbf{Action Space.} Each Miner $i$ chooses a quantity of investment $q_i$ in the Upstream Tullock Contest, and a reserve price $r_i$ for a first-price auction they will run among end-users.
    \item[] \textbf{Costs.} Miner $i$ pays cost $c_i^R$ per unit of Resource, and $c^W$ per \writeOp{}. That is, if Miner $i$ wishes to invest $q_i$ in the Upstream game, they pay cost $q_i \cdot c_i^R$. If Miner $i$ eventually sells $Q'_i \leq \frac{q_i}{\sum_j q_j} \cdot Q_A$ \writeOp{}s, Miner $i$ pays cost $Q'_i \cdot c^W$. W.l.o.g.~we let $c_i^R \leq c_{i+1}^R$ for all $i \in [n-1]$.
    \item[] \textbf{Payoffs.} For a Miner $i$ who invests $q_i$ in the Upstream game and eventually sells $Q'_i$ \writeOp{}s, their total cost is $c_i^R \cdot q_i + Q'_i \cdot c^W$. They receive a block reward of $B \cdot\ttlAppend\cdot  q_i/\sum_{j} q_j$, plus any additional revenue earned in their first-price auction. Therefore, if Miner $i$ earns revenue $R_i$ from their first-price auction, invests $q_i$ in the upstream game, and sells $Q'_i$ \writeOp{}s, their payoff is $R_i - Q'_i \cdot c^W- q_i \cdot c_i^R + B \cdot \ttlAppend \cdot \frac{q_i}{\sum_j q_j}$.
\end{itemize}

\textsc{End-User}
\vspace{-0.15cm}
\begin{itemize}
        \item[] \textbf{Players.} There is a continuum of End-Users who follow a demand curve $D(\cdot)$ for \writeOp{}s.
    \item[] \textbf{Action Space.} Each end-user submits a bid to each First-Price Auction.
    \item[] \textbf{Payoffs.} An end-user with value $v$ has payoff $v-q$ for receiving at least one \writeOp{} and paying total price $q$,\footnote{That is, End-Users are unit-demand and only want a single \writeOp{} -- if they win multiple auctions they do not get additional utility. Still, they make a payment in any auction they win.} and payoff $0$ if they do not get a \writeOp{}.
\end{itemize}
\end{definition}

Before proceeding to analysis, some discussion is warranted on why the above model captures popular distributed ledgers, and what differentiates it from classic market structures.

\vspace{0.1cm}
\noindent\textbf{Key Differences.} One key difference between the Distributed Ledger Model and traditional market structures is the presence of a non-strategic protocol. Specifically, the Upstream game is hard-coded in the Distributed Ledger Model rather than endogenously optimized by a profit-maximizing Monopolist. This distinction is key.\footnote{It is \emph{certainly} relevant to \emph{also} study the meta-game by which protocol rules are formed, but the game induced by a fixed protocol would still be relevant for the entirety of Bitcoin's existence.} Additionally, the decision to fix the quantity $Q_A$ of \append{}s and run a Tullock contest is material, and meaningfully affects the analysis.\footnote{For example, conclusions would change if instead the protocol set a price $p_A$ for \append{}s and sold whatever is demanded, even if $p_A$ were determined exogenously. As previously noted, all blockchain protocols the authors are aware of run an Upstream Tullock contest, and it is not clear how to implement alternate Upstream market structures with a secure protocol.}  Finally, the decision to have a block reward (which directly rewards miners for purchasing \append{}s, even if they do not ultimately sell \writeOp{}s) is material, although our analysis shows limited impact on end-users (see Section~\ref{sec:bitcoinblock}).

\vspace{0.1cm}
\noindent\textbf{Connecting the Distributed Ledger Model to Distributed Ledgers in Practice.} First, we map aspects of Bitcoin onto the Distributed Ledger Model. Assuming that the Bitcoin protocol functions as intended,\footnote{See Section~\ref{sec:related} for a subset of works describing manners by which the protocol may not function as intended.} let us first describe the interaction between Miners and Protocol. In order to produce a valid block, Miners must solve a ``proof-of-work cryptopuzzle.'' Specifically, Miners trade one hash computation for one independent Bernoulli trial to create a valid block.\footnote{With extremely low probability of success -- roughly $2^{-78}$ at the time of writing.} Moreover, the success rate of each Bernoulli trial is dynamically adjusted by the Bitcoin protocol so that one block is created amongst the entire network every ten minutes. Each block provides 1 MB of space for the Miner to include transactions, and awards the miner a block reward (currently 3.125 BTC). So in our model, the interaction between Miners and Protocol captures the following aspect: Resource is hash computations. Each Miner $i$ has some cost $c_i^R$ to perform one hash computation.\footnote{This includes electricity, operational costs, amortized hardware costs, etc.} \append{}s are units of space in a valid block, and Protocol has hard-coded that $Q_A$ = 1 MB per ten minutes are awarded in total and that each Miner receives a fraction of 1 MB blocks proportional to their hash computations (because the success probability of each hash dynamically adjusts to enforce a total quantity of 1 MB per ten minutes). Finally, the protocol hardcodes $B = 3.125$ BTC per 10 minutes as the total block reward,\footnote{Note that this quantity halves every four years, as pre-specified by the Bitcoin protocol.} which is also distributed proportionally to miners according to their hash computations.

Aside from its interaction with Miners, Protocol simply maintains consensus on the contents of the ledger. For example, Protocol verifies validity of contents of the ledger, resolves any conflicts using ``Nakamoto consensus''~\cite{Nakamoto08}, and widely disseminates the ledger itself. In particular, the protocol rules suffice to identify a unique consistent ledger to disseminate. 

End-users get their transactions in a Bitcoin block by broadcasting to Miners. Each transaction includes a transaction fee, which is paid to whichever Miner includes that transaction in their block. Processing a transaction induces costs such as checking validity, and maintaining network connectivity (to hear about transactions in the first place), which are captured by $c^W$ in our model. Miners are typically revenue-maximizing, and typically fill their blocks with transactions paying the highest fees (and, to the best of the authors' knowledge, typically without reserves). For a \emph{patient} end-user, who wants their transaction in the Bitcoin ledger eventually but not necessarily immediately, each Miner is a potential seller running their own First-Price Auction, and the service offered by distinct miners is indistinguishable. 

It is also worth highlighting which of these aspects are key to fit our model, and which are not. All permissionless distributed ledgers the authors are aware of run a Tullock Contest in \emph{some} Resource. Some (including Bitcoin Cash, Litecoin, Ethereum Classic)
also use hash computations as Resource (``proof-of-work''). Others (including Ethereum, Solana, Cardano)
use locked capital as Resource (``proof-of-stake'').\footnote{That is, Miners are now called Stakers, who trade one unit of locked capital per Bernoulli trial.} It is not material to our analysis which Resource is used, only that the Protocol ultimately awards block space proportional to that resource. 

We note that~\cite{ArnostiW22} also model blockchain investment games as a Tullock contest, while other works~\cite{HubermanLM21} instead model it as perfect competition with free entry. We briefly note that free entry can also be captured arbitrarily well in our model by taking $n \rightarrow \infty$ Miners with identical $c_i^R$ (see Theorem~\ref{thm:bitcoinsufficient}, for example).\footnote{Specifically, the final bullet concludes that the natural ``market clearing equilibrium'' becomes an equilibrium with sufficiently-many identical Miners.}

On the other hand, it is crucial for our model to accurately capture end-users that they are \emph{patient} (such as those making non-urgent payments, or Layer-2 protocols posting data on-chain) and therefore equally happy to be included in any valid block. \emph{Impatient} users (such as those interacting directly with the ledger for DeFi applications) are not captured by our model.Instead, impatient users view the particular block offered by the next Miner as the only resource of interest, and therefore that Miner faces no competition. Therefore, a Miner selling blockspace primarily to impatient users is instead a monopolist. 

Additionally, observe that Miners(/Stakers) are free to use whatever ``off-chain'' auction they like in order to sell space in their created blocks, independent of whatever ``on-chain'' mechanism is hardcoded. For example, it is immaterial to our model that Bitcoin's on-chain mechanism is pay-your-bid, whereas Ethereum's is a posted-price mechanism, because Miners(/Stakers) in both protocols can run a first-price auction with reserve off-chain to determine which transactions are included in the first place. On the other hand, the fact that Ethereum's EIP-1559 \emph{burns}\footnote{That is, transactions included in an Ethereum pay a posted-price (set by the Ethereum protocol) that is destroyed, and not awarded to the miner.} revenue from the posted-price mechanism is material, and can be captured in our model via $c^W$. That is, our model adopts the perspective of~\cite{GaneshTW24, GaneshTW26} that EIP-1559 is really specifying a \writeOp{} cost per included transaction (the burned base fee) on each \emph{Proposer}, and the Proposer is then free to run whatever off-chain auction they like to build their block (rather than that EIP-1559 specifies \emph{the} auction that Proposers \emph{must} run when facing users).

Finally, while all mainstream protocols the authors are aware of currently have a single Proposer at each time slot, some protocols are now experimenting with ``Multiple Concurrent Proposer (MCP)'' protocols. In these protocols, there is no longer a monopolist for each block slot. Instead, multiple Proposers have the opportunity to insert transactions. Interestingly, our model also captures MCP protocols with either patient or impatient users.\footnote{To be extra clear, our model verbatim captures an MCP protocol where the block space in each block is partitioned according to stake (i.e.~a 10\% staker gets 10\% of every block). Most MCP proposals instead sample a discrete number of Proposers proportional to stake, and allow each such Proposer an equal fraction of the block. Our model does not capture this verbatim, as the sampling process would meaningfully complicate analysis, but it would be a natural direction for follow-up work to modify our model to capture these MCP protocols verbatim.}

In summary, the Distributed Ledger Model captures key aspects of many mainstream blockchain protocols (they are Tullock contests, and the protocol can impact $c^W$), while not capturing others (such as impatient end-users, or blockchains with a heavy MEV ecosystem\footnote{In blockchains with a heavy MEV ecosystem, the process of turning \append{}s into \writeOp{}s itself is cost-intensive and meaningfully asymmetric.}). Indeed, this is in line with the motivation for our paper: Decentralization does not uniformly impact all users identically in all domains, and so our model necessarily picks one canonical domain of focus.

\vspace{0.1cm}
Throughout this section, we abuse notation and use $(\vec{q},\vec{r})$ to refer to the Downstream Equilibrium $(\vec{q},\vec{r})$ together with the mapping $E(\cdot)$ that takes $(\vec{q}',\vec{r}')$ to the canonical End-User Equilibrium. We repeat the key takeaways of this section below:

\begin{itemize}
\item {While Equilibria may not always exist, Proposition~\ref{prop:bitcoinundominated} establishes an upper bound on the price any Miner will set: any strategy that sets a price exceeding the monopoly reserve for $D(\cdot)$ is a dominated strategy.}
    \item Theorem~\ref{thm:bitcoinpure} characterizes all potential Equilibria. In particular, there exists a single $\vec{Q}$ such that Miner $i$ wins $Q_i$ \append{}s in all pure equilibria.
    From here, there is a unique ``market-clearing'' potential equilibrium (where each miner sets reserve at most $D^{-1}(Q_A)$), and for each $i$ a unique potential equilibrium where Miner $i$ is a price-setter. 
    \item Theorem~\ref{thm:bitcoinsufficient} provides necessary and sufficient conditions for the ``market-clearing'' potential equilibrium to be an equilibrium.\footnote{In addition, we extend the necessary and sufficient conditions for the "market-clearing" equilibrium (when they exists) to the scenario where different Miners having different cost per \writeOp{} in \Cref{app:bitcoinasym}.} Importantly, the condition \emph{depends only on $D(\cdot), Q_A$, and $\max_i \{Q_i\}$}. That is, to the extent that a ``measure of decentralization'' impacts the ultimate price paid by end-users, the correct ``measure of decentralization'' is the market share of the largest miner.
    \item The magnitude of the block reward (or whether there is a block reward at all) \emph{does not impact $\vec{Q}$, nor any of the potential Equilibria identified by Theorem~\ref{thm:bitcoinpure}}. But, a larger block reward makes Equilibria more likely to exist (see \Cref{prop:blockrewards}). 
\end{itemize}

\subsection{Higher-than-Monopolist Reserves are Dominated Strategies}

Before reasoning about equilibria, we first reason about what reserves a Miner might set in an undominated strategy. Formally, we say that Downstream Strategy $(q_i, r_i)$ dominates $(q'_i, r'_i)$ if \emph{in the game among Downstream Providers}, $P_i((q_i, \vec{q}_{-i}), (r_i, \vec{r}_{-i})) \geq P_i((q'_i, \vec{q}_{-i}), (r'_i, \vec{r}_{-i}))$ for all $\vec{q}_{-i}, \vec{r}_{-i}$.

\begin{proposition}
\label{prop:bitcoinundominated}
    {Let $r^*(D, Q_A):= \arg\max_{r \geq D^{-1}_{\inf}(Q_A)}\{(r-c^W)\cdot D(r)\}$. Then for all Miners $i$, all $q_i > 0$, and all $r_i > r^*(D, Q_A)$, $(q_i, r^*(D, Q_A))$ dominates $(q_i, r_i)$. }
\end{proposition}

Proposition~\ref{prop:bitcoinundominated} establishes an upper bound on what price could possibly arise, even out of equilibrium -- it can be no worse than the price that would be set by a single Miner who produces all blocks.

\subsection{Characterizing Equilibria (when they exist)}\label{sec:bitcoinpure}

End-User Equilibria are analyzed in~\Cref{sec:enduser}, which describes how to determine the clearing price as a function of the Miners' strategies. This section focuses on analyzing Downstream Equilibria. While Pure Equilibria do not always exist (even when $D(\cdot)$ is Regular -- see Section~\ref{sec:bitcoinexample}), we are able to cleanly characterize all \emph{potential} Pure Equilibria. Below, we outline the characterization.

\begin{itemize}
    \item Taking investments $\vec{q}$ as fixed, Miners have quantities $\vec{Q}$, which necessarily satisfy $\sum_i Q_i = Q_A$ (as $Q_A$ is exogenously set). In any Pure Equilibrium, it therefore must be that $\vec{r}$ is a price-setting equilibrium for $\vec{Q}$.
    \item Taking reserves $\vec{r}$ as fixed and the total quantity $Q_A$ as fixed, $\vec{q}$ must be an investment equilibrium. Knowing from Proposition~\ref{prop:firstpriceeq} that any price-setting equilibrium has at most one price-setter, we can instead study: taking the clearing price $r$, the price-setter $i$, and the total quantity $Q_A$ as fixed, $\vec{q}$ must be an investment equilibrium. This is \emph{almost} like asking for an equilibrium in the Tullock contest defined by costs $c_i^R$ and total reward $(r-c^W) \cdot D(r)+B \cdot Q_A$.
    \begin{itemize}
        \item But, the game is not \emph{exactly} a Tullock contest. Indeed, no one gets reward $((r-c^W) \cdot D(r)+B\cdot Q_A)\cdot q_i/\sum_j q_j$ -- all $j \neq i$ receive reward $(r-c^W+B)\cdot Q_A \cdot q_i/\sum_j q_j$, and the price-setter $i$ receives reward $(r-c^W+B)\cdot Q_A \cdot q_i/\sum_j q_j - (r-c^W) \cdot (Q_A - D(r))$. However, after inspecting both reward formulas, \emph{the marginal change in each Miners' payoff is identical to the marginal change in $(r-c^W+B) \cdot Q_A\cdot q_i/\sum_j q_j$}. Therefore, the same local optimality conditions that must be satisfied by an equilibrium of a Tullock contest with total reward $(r-c^W+B)\cdot Q_A$ must be satisfied by any equilibrium with clearing price $r$.
        \item Importantly, however, while the equilibrium \emph{investments} in a Tullock contest certainly depends on the total reward split, the equilibrium \emph{resulting market shares} do not. 
        \end{itemize}
    \item The previous two bullets suggest the following as necessary conditions for a pure equilibrium:
    \begin{itemize}
        \item The resulting market share of \append{}s won must match those in equilibrium of a Tullock contest where Miner $i$ incurs cost $c_i^R$ per Resource. Importantly, this equilibrium is unique and well-defined. Call this vector of quantities $\vec{Q}^*(\vec{c})$. 
        \item Let $Q^{\text{OPT}}_i:=\arg\max_{Q \in [\sum_{j \neq i }Q^*_j, Q_A]}\{(Q -\sum_{j \neq i} Q_j)\cdot D^{-1}(Q)\}$. Then some Miner $i$ is a price-setter at $D^{-1}(Q^{\text{OPT}}_i)$ (this includes the possibility that $Q^{\text{OPT}}_i = Q_A$ for all $i$, and the unique possible Pure Equilibrium saturates all Miners). 
        \item Indeed, Theorem~\ref{thm:bitcoinpure} confirms these characterize all potential Pure Equilibria.
    \end{itemize}
    \item Ultimately, in order to be a Pure Equilibrium, the question is whether whether the pair $(q_i, r_i)$ \emph{which can be changed in tandem} is a best-response to $\vec{q}_{-i},\vec{r}_{-i}$. The previous bullets expound upon necessary conditions for this to plausibly occur -- if $(q_i,r_i)$ is to be a best response to $(\vec{q}_{-i},\vec{r}_{-i})$, $q_i$ must be a best response to $(\vec{q}_{-i},\vec{r}_{-i};r_i)$, and $r_i$ must be a best response to $(\vec{q}_{-i},\vec{r}_{-i},q_i)$. Section~\ref{sec:bitcoinexample} contains an example demonstrating the possibility of no Pure Equilibria.
\end{itemize}

We execute this outline in Appendix~\ref{app:bitcoin} -- Theorem~\ref{thm:bitcoinpure} characterizes all possible Pure Equilibria.

\begin{definition}
    Define $c^*(\vec{c}^R)$ to be the unique solution to $\sum_{i=1}^n \max\{0,1-\frac{c_i^R}{c^*(\vec{c}^R)}\}=1$.\footnote{\cite{ArnostiW22} establish that $c^*(\vec{c}^R)$ is well-defined -- the proof is straight-forward.} Further define $x_i^*(\vec{c}^R):=\max\{0,1-\frac{c_i^R}{c^*(\vec{c}^R)}\}$. 
\end{definition}

\begin{theorem}\label{thm:bitcoinpure}
Let $(\vec{q},\vec{r})$ be an Equilibrium in the Distributed Ledger Model, and let the clearing price for End-Users be $r$. Then:
\begin{itemize}
    \item $\sum_j q_j = Q_A \cdot (r+B-c^W)/c^*(\vec{c}^R)$.
    \item For all Miners $i$, $\frac{q_i}{\sum_j q_j} = x_i^*(\vec{c}^R)$.
\end{itemize}
Moreover, $r \geq D^{-1}_{\sup}(Q_A)$, and:
\begin{itemize}
    \item If $r = D^{-1}_{\sup}(Q_A)$, then $Q_A \cdot x^*_i(\vec{c}^R) \leq \frac{(x - c^W) (Q_A-D(x))}{x - r}$ for all $i$ and all $x > r$. 
    \item If $r > D^{-1}_{\sup}(Q_A)$, then there is a single price-setter $i^*$, who sets a price equal to $r_{i^*}:= \arg\max_{x > D^{-1}_{\sup}(Q_A)}\{(x - c^W) \cdot (D(x)+Q_A\cdot x^*_{i^*}(\vec{c}^R)-Q_A)\}$. If an equilibrium of this form exists, one certainly exists with $i^* =1$ as the price-setter (equilibria with other price-setters are possible).
\end{itemize}
\end{theorem}

Importantly, observe in Theorem~\ref{thm:bitcoinpure} that:
\begin{itemize}
    \item Block rewards play no role in the ultimate clearing price, nor the resulting market share of each Miner.\footnote{As noted previously, Block rewards do play a role in determining whether Pure Equilibria exist -- see Section~\ref{sec:bitcoinblock}.}
    \item A necessary condition for an equilibrium to exist with clearing price $D^{-1}_{\sup}(Q_A)$ (the smallest possible clearing price) is only a function of $x_1^*$, $Q_A$, and $D(\cdot)$. That is, to the extent that a quantitative measure of decentralization plays a role in the ultimate price paid by end-users, it is the size of the largest Miner. Moreover, $x_1^*$ can be determined \emph{only as a function of $\vec{c}^R$}. 
    \item That is, there is one term that depends only on $\vec{c}^R$ ($x_1^*(\vec{c}^R)$), and another that depends only on $c^W$ and $D(\cdot)$ ($\inf_{x>r} \{\frac{(x-c^W)\cdot (Q_A-D(x))}{x-r}\}$), and an Equilibrium that clears all $Q_A$ \writeOp{}s made available by the protocol can plausibly exist if and only if $x_1^*(\vec{c}^R) \leq \inf_{x>r} \{\frac{(x-c^W)\cdot (Q_A-D(x))}{x-r}\}$.
\end{itemize}

\subsection{A Sufficient Condition for Pure Equilibria}
Theorem~\ref{thm:bitcoinpure} characterizes all possible pure Equilibria, and so it is tempting to proceed with equilibrium analysis under these conditions. Unfortunately, pure Equilibria are not guaranteed to exist in the Distributed Ledger Model. In Section~\ref{sec:bitcoinexample} we provide an example demonstrating this, and identify the barrier. This motivates a natural sufficient condition that we analyze in Section~\ref{sec:bitcoinsufficient}. Along the way, we also discuss the impact of block rewards on equilibria in Section~\ref{sec:bitcoinblock}.

\subsubsection{An Example with  with Non-Existence}\label{sec:bitcoinexample}

Consider a demand curve $D(\cdot)$ with $D(x) = 1-x$ for all $x \in [0,1]$, $Q_A = 1$, $B = 0$, and $c^W = 0$. Consider also $n=3$ miners each with $c^R_i = 1$. Then Theorem~\ref{thm:bitcoinpure} concludes:
\begin{itemize}
    \item In any potential equilibrium, it must hold that $Q_i(\vec{q},\vec{r}) = 1/3$ for all $i$.
    \item Fixing $Q_1 = 1/3$, $\arg\max_{x \geq 0}\{x \cdot (1/3-x)\} = 1/6$. Therefore, it is not possible to have an equilibrium with clearing price $r = D^{-1}(1) = 0$.\footnote{Because in such an equilibrium, all three Miners earn profit zero, whereas any Miner could deviate to set a price of $1/6$ and instead earn profit $1/36$ (letting the other Miners earn profit $1/18$).}
    \item $c^*(1,1,1) = 3/2$. Therefore, the only potential equilibria have one Miner as a price-setter at price $1/6$, with $\sum_j q_j = \frac{1/6}{3/2} = 1/9$ (and therefore each Miner has $q_i = 1/27$). 
    \item However, this is not an equilibrium. In this strategy profile, the price-setter earns revenue $1/36$ from the simultaneous first-price auctions, but pays $1/27$ in Resource cost, yielding negative payoff. The price-setter would be better off not investing at all. 
\end{itemize}

In particular, what stands out about this example is that a fully-saturated equilibrium has absolutely no shot (because such an equilibrium would generate zero revenue in total). This suggests a natural sufficient condition: that the fully-saturated equilibrium generate \emph{some} fraction of the optimal revenue that a single entity controlling all \append{}s could earn. 

\subsubsection{A Sufficient Condition}\label{sec:bitcoinsufficient}

Theorem~\ref{thm:bitcoinsufficient} below states an interpretable sufficient condition for an equilibrium with clearing price $r=D^{-1}(Q_A)$, and a technical condition that is necessary and sufficient when $B=0$. 

\begin{theorem}\label{thm:bitcoinsufficient} Consider a potential equilibrium $(\vec{q}^*,\vec{r}^*)$ such that: 
\begin{itemize}
    \item The clearing price is $r = D^{-1}_{\sup}(Q_A)$.
    \item $\sum_j q^*_j = Q_A \cdot (r+B-c^W)/c^*(\vec{c}^R)$.
    \item For all Miners $i$, $\frac{q_i}{\sum_j q_j} = x_i^*(\vec{c}^R)$.
\end{itemize}

Then:
\begin{itemize}
    \item If $D(\cdot)$ is Regular and $x_1^*(\vec{c}^R) \leq 1-\frac{1}{D(0)/Q_A -1}$, then $(\vec{q}^*,\vec{r}^*)$ is an Equilibrium.
    \item If $B = 0$, define $k(z):=\frac{D^{-1}_{\sup}(z \cdot Q_A)\cdot z\cdot Q_A}{D^{-1}_{\sup}(Q_A)\cdot Q_A}$. Then $(\vec{q}^*,\vec{r}^*)$ is an Equilibrium if and only if $x_1^*(\vec{c}^R) \leq 1-\sup_{z \in [0,1]}\left\{\frac{k(z)-1}{2\cdot\left(\sqrt{k(z)/z}-1\right)}\right\}$
\end{itemize}
    
\end{theorem}

As previously noted, Theorem~\ref{thm:bitcoinsufficient} provides interpretable sufficient conditions for $Q_A$ \writeOp{}s to clear in Equilibrium (in Bullet One), and necessary and sufficient conditions (when $B = 0$) in Bullet Two. In both cases, the condition depends only on $D(\cdot), Q_A,$ and $x_1^*(\vec{c}^R)$, highlighting the ``size of the largest miner'' (which can be computed as a function only of $\vec{c}^R$) as the relevant ``measure of decentralization'' for determining the economic impact on end-users (in our model).

We can also use Theorem~\ref{thm:bitcoinsufficient} to reason about modifications to the example in Section~\ref{sec:bitcoinexample}. Consider the case of $n$ Miners each with $c_i^R = 1$, $D(x) = 1-x$ for all $x \in [0,1]$, $B = 0$ and $c^W = 0$, but we will vary $Q_A$.
\begin{itemize}
    \item Then Bullet One of Theorem~\ref{thm:bitcoinsufficient} confirms it is an Equilibrium for $Q_A$ \writeOp{}s to clear (at price $1-Q_A$) as long as $1/n \leq 1-\frac{1}{1/Q_A-1}$, which can be rewritten as $Q_A\leq \frac{n-1}{2n}$. For this particular example in Section~\ref{sec:bitcoinexample}, this is not particularly impressive, as even a Miner controlling the entire blockchain would sell all $Q_A$ \append{}s as long as $Q_A \leq 1/2$.
    \item The more precise Bullet Two of Theorem~\ref{thm:bitcoinsufficient} could be applied. In this case, $k(z)=\frac{z \cdot (1-z\cdot Q_A)}{1-Q_A}$, and for example when $Q_A = 3/4$, $\sup_{z \in [0,1]}\left\{\frac{k(z)-1}{2\cdot \left(\sqrt{k(z)/z}-1\right)}\right\} = 2/3$. Therefore, $n=3$ identical Miners are sufficient in order for a Protocol with $Q_A = 3/4$ to have an Equilibrium where all $Q_A$ \writeOp{}s are sold (whereas $n=1$ Miner is insufficient, as such a Miner would choose to monopoly-price and sell only $1/2$).
\end{itemize}

\subsubsection{Block Rewards Support Existence of Pure Equilibria}\label{sec:bitcoinblock}
In this section, we reason about the role of block rewards on equilibria. This section has two main results, both below.

\begin{proposition}\label{prop:blockrewards}
    Let $(\vec{q}^*,\vec{r}^*)$ induce a clearing price of $r$, and be an equilibrium of the Distributed Ledger Model with Resource costs $\vec{c}^R$, \writeOp{} cost $c^W$, demand curve $D(\cdot)$, and block reward $B$. Then $( (1+\frac{B'-B}{Q_A \cdot (r-c^W+B)})\cdot \vec{q}^*, \vec{r}^*)$ is an equilibrium of the Distributed ledger Model with Resource costs $\vec{c}^R$, \writeOp{} cost $c^W$, demand curve $D(\cdot)$, and block reward $B'>B$. 
\end{proposition}

Proposition~\ref{prop:blockrewards} establishes that higher block rewards support the existence of pure equilibria (although Theorem~\ref{thm:bitcoinpure} establishes that block rewards do not impact potential pure equilibria themselves). This is of standalone interest for understanding the impact of modeling parameters on end-users, and also a technical ingredient in the proof of Theorem~\ref{thm:bitcoinsufficient}.

Additionally, we show that if $\vec{c}^R, Q_A, D(\cdot), c^W$ does not immediately rule out a market-clearing equilibrium, there is a sufficiently large block reward so that a market-clearing equilibrium exists. 

\begin{proposition}\label{prop:blockrewards2}
Let $\vec{c}^R, D(\cdot), Q_A, c^W$ be such that $Q_A \cdot x_i^*(\vec{c}^R) < \inf_{x > D^{-1}_{\sup}(Q_A)}\left\{\frac{(x-c^W) \cdot (Q_A - D(x))}{x-D^{-1}_{\sup}(Q_A)}\right\}$. Then, there exists a sufficiently large $B < \infty$ such that a market-clearing equilibrium exists in the market defined by $\vec{c}^R, D(\cdot), B, Q_A, c^W$.
\end{proposition}

Recall from Theorem~\ref{thm:bitcoinpure} that if $Q_A \cdot x_i^*(\vec{c}^R) > \inf_{x > D^{-1}_{\sup}(Q_A)}\left\{\frac{(x-c^W) \cdot (Q_A - D(x))}{x-D^{-1}_{\sup}(Q_A)}\right\}$, no market-clearing equilibrium can possibly exist -- Proposition~\ref{prop:blockrewards2} shows that with sufficiently large block rewards, this is the only barrier (up to cases where $Q_A \cdot x_i^*(\vec{c}^R) = \inf_{x > D^{-1}_{\sup}(Q_A)}\left\{\frac{(x-c^W) \cdot (Q_A - D(x))}{x-D^{-1}_{\sup}(Q_A)}\right\}$). We additionally show via example in Section~\ref{sec:bitcointight} that Proposition~\ref{prop:blockrewards2} is ``tight'' in the sense that the strict inequality cannot be replaced with a weak inequality -- examples of $\vec{c}^R, D(\cdot), Q_A, c^W$ satisfying $Q_A \cdot x_i^*(\vec{c}^R) = \inf_{x > D^{-1}_{\sup}(Q_A)}\left\{\frac{(x-c^W) \cdot (Q_A - D(x))}{x-D^{-1}_{\sup}(Q_A)}\right\}$ exist so that no matter how large a block reward \emph{all} $B < \infty$ is provided, no market-clearing equilibrium exists in the market defined by $\vec{c}^R, D(\cdot), Q_A, c^W, B$. 
\section{Conclusion} \label{sec:conclusion}
We investigate decentralization as a means to insulate a natural monopoly from its derivative services. We then draw conclusions regarding the ultimate utility of end-users as a function of protocol parameters. We further highlight the impact of various aspects on our conclusions: (a) our analysis applies to patient users (who are content with purchase from any miner) and not impatient users (who view each miner as a monopolist anyway), (b) the relevant ``measure of decentralization'' for impact on users' price is the size of the largest miner (which can be determined exclusively as a function of the profile of Resource investment costs, as in a pure Tullock contest), (c) block rewards don't impact users' price at equilibrium, but can influence whether equilibria exist.

Within distributed ledgers, our model considers the basic setup where users directly interact with miners to include a transaction on the blockchain. Of course, many blockchain ecosystems have evolved and now include additional parties (Builders and Layer-2s are two notable examples). Our work provides a framework through which to ask: how does the presence of these parties ultimately impact the service users receive?

Beyond our model, it is also important to endogenize aspects that our model treats as exogenous. For example, our model treats as exogenous the fact that Bitcoin/Ethereum run a Tullock contest in Computation/Stake. While this is an accurate representation of all major blockchains since their inception, and Upstream protocol mechanics change \emph{much} more slowly than strategic Downstream decisions, these protocols are not truly exogenous -- protocol rules are set by some governance process (perhaps formally specified, perhaps not). It is therefore an important direction for future work to additionally model the dynamics by which protocol mechanics are determined. Additionally, our model treats the ``lines'' between distinct miners/stakers are exogenous. While it is again the case that miners/stakers merge at a \emph{much} slower pace than adapting strategic Downstream decisions, miner/staker identities are not exogenously fixed -- parties might certainly merge and jointly strategize if they find it beneficial. In particular, if all parties in our model were to merge/collude or otherwise jointly strategize, they could profit by setting monopoly prices, so it is important to additionally model the process by which miners/stakers might merge, collude, or otherwise jointly strategize.

Beyond distributed ledgers, there are many natural monopolies for digital services that pose challenges for traditional regulatory approaches~\cite{Tirole23}. Our results provide theoretical foundations for exploring decentralized protocols as a tool that \emph{might} prove useful to insulate such natural monopolies. Future work could, for example, consider decentralized protocols that manage marketplaces (allowing competition among matchmakers, search, etc.) or store social network data (allowing competition among UIs, content moderation, etc.). Our work motivates further investigation of decentralized protocols as a means to insulate natural monopolies in other domains in a detailed manner that can both (a) identify which natural monopolies might be amenable to insulation by a decentralized protocol, and (b) provide a complete analysis linking design choices of the decentralized protocol to impact on users. Importantly, our work models only the decentralized protocol and analyzes the resulting users' utilities, but does not compare to a counterfactual. That is, our model is apt for comparing users' utilities in a decentralized protocol with a largest miner of size $x_1$ versus users' utilities in a decentralized protocol with a largest miner of size $x'_1$, but not for comparing users' utilities in a decentralized protocol with a largest miner of size $x_1$ to users' utilities with an alternative market structure (i.e.~a monopolist, a regulated market, or anything else). Therefore, an important direction for future work is to rigorously model alternative market structures so that comparisons to decentralized protocols are possible (without this, one cannot argue that decentralized protocols are ``better'' or ``worse'' than existing alternatives).

Beyond insulating natural monopolies, our work contributes to an emerging line of works seeking theoretical foundations for the impact of decentralized systems on users~\cite{HubermanLM21, SockinX23, GoldsteinGS24, Reuter24}. Significant further work along these lines is necessary in order to understand domains where decentralized systems have a shot at providing lasting value, even in the presence of highly-developed incumbents.

\bibliographystyle{alpha}
\bibliography{references}

\appendix
\section{Technical Preliminaries: Equilibria of Simultaneous First-Price Auctions}\label{sec:enduser}

Simultaneous First-Price Auctions are a recurring subgame in our market structures. End-User Equilibria will correspond to equilibria for bidders, and (aspects of the) Downstream Equilibria will correspond to equilibria for sellers.

The results in this section will be used as building blocks for our main theorems. Most of our main theorems will eventually assume $D(\cdot)$ is Regular, but our results here do not require any further assumptions on $D(\cdot)$. 

In particular, observe that $D(\cdot)$ is necessarily non-increasing, but for this section we do not assume $D(\cdot)$ is strictly increasing. In addition, $D(\cdot)$ is necessarily left-continuous everywhere,\footnote{To see this, observe that all end-users with value at least $v-\varepsilon$ contribute to $D(v-\varepsilon)$. So the end-users that contribute to $D(v-\varepsilon)$ for all $\varepsilon > 0$ are exactly those with value at least $v$ -- the same consumers that contribute to $D(v)$. } and right-continuous except at countably-many points,\footnote{To see this, recall that any monotone function can only have countably-many jump discontinuities (and can only have jump discontinuities).} but for this section we do not assume $D(\cdot)$ is continuous everywhere. 

We'll further use the notation $D^{>}(p):=\lim_{q\rightarrow p^+} D(q)$ to denote the mass of users with value strictly exceeding $p$. $D^{>}(\cdot)$ is necessarily right-continuous everywhere,\footnote{To see this, observe that the end-users contributing to $D^{>}(p)$ are exactly those with value strictly exceeding $p$, and the users contributing to $D(p+\varepsilon)$ for some $\varepsilon > 0$ are also exactly those with value strictly exceeding $p$.} and left-continuous except at countably-many points, and that $D^{>}(\cdot)$ and $D(\cdot)$ share the same discontinuities.

\subsection{Bidding in Simultaneous First-Price Auctions}\label{sec:bidding}

First, we analyze equilibrium behavior of bidders in Simultaneous First-Price Auctions, treating $\vec{Q},\vec{r}$ as fixed. The main result of this section is Proposition~\ref{prop:firstpriceeq}. We first show that at equilibrium, every bidder who wins an item pays the same amount (\Cref{lem:firstpriceeq}), and define this payment as the \emph{clearing price}.~\Cref{prop:firstpriceeq} then characterizes possible clearing prices at equilibrium. 

\begin{lemma}\label{lem:firstpriceeq} In any equilibrium of Simultaneous First-Price Auctions, there exists a single bid $b$ such that every bidder either wins exactly one item and pays $b$, or loses. 
\end{lemma}

\begin{proof}[Proof of Lemma~\ref{lem:firstpriceeq}]
    Assume for contradiction that two bidders $i, j$ both receive at least one $\writeOp{}$, and pay total prices $b_i > b_j$ in an equilibrium. Because bidder $j$ pays total price $b_j$ to receive at least one \writeOp{}, there must exist at least one auction with an effective price of $p \leq b_j$. Bidder $i$ could therefore submit a bid of $p < b_i$ into this auction and $0$ into all other auctions, winning a \writeOp{} and making total payment $p < b_i$, contradicting that the original bids form an equilibrium.\footnote{Importantly, observe that \emph{because Bidders are a continuum}, a single bidder $i$ changing their bid does not impact the clearing price in any auction.} 
\end{proof}

\begin{proposition}\label{prop:firstpriceeq}
Consider a continuum of unit-demand buyers participating in simultaneous first-price auctions with asymmetric quantities $\vec{Q}$ and reserves $\vec{r}$. Then:
\begin{itemize}
    \item Let $p_{\min}:=\inf\{b\ |\ D^{>}(b) \leq Q^{\leq}(b)\}$ and $p_{\max}:=\sup\{b\ |\ D(b) \geq Q^{<}(b)\}$. Let also $B_{\min}(p):=\max\{D^{>}(p), Q^{<}(p)\}$ and $B_{\max}(p):=\min\{D(p),Q^{\leq}(p)\}$. Then an equilibrium exists with clearing price $p$ and total supply cleared $B$ if and only if $p \in [p_{\min},p_{\max}]$ and $B \in [B_{\min}(p),B_{\max}(p)]$.
    \item All canonical equilibria have clearing price $p_{\min}$, and clear a total supply cleared $B_{\max}(p_{\min})$.
    \item If $D(\cdot)$ is continuous and strictly decreasing, let $p$ be the unique $p$ such that $Q^{<}(p) \leq D(p) \leq Q^{\leq}(p)$. Then an equilibrium exists with clearing price $p$ and total supply cleared $D(p)$, and all equilibria have clearing price $p$ and total supply cleared $D(p)$.
\end{itemize}
\end{proposition}

\begin{proof}
We begin by asking whether it is possible to have an equilibrium where a total mass of $B>0$ buyers each bid $b$ into exactly one auction and win, and all remaining buyers lose. In order for such an equilibrium to exist, it must be the case that:
\begin{itemize}
    \item First, the mass of buyers with value strictly exceeding $b$ is at most $B$ ($D^{>}(b) \leq B$). Otherwise, a user with value strictly exceeding $b$ is losing, and could strictly profit by submitting a bid of $b$ to any auction with reserve $\leq b$ (which must exist, if only a mass of $B$ buyers win an item).
    \item Second, the mass of users with value at least $b$ is at least $B$ ($D(b) \geq B$). Otherwise, some user with value strictly below $b$ is winning and paying $b$, and they could strictly profit by bidding $0$ everywhere instead.
    \item Third, it must be that the total mass of items sold at reserve at most $b$ is at least $B$ ($Q^{\leq}(b) \geq B$). Otherwise, it is not possible to have a mass of $B$ items sold at price less than $b$. 
    \item Fourth, it must be that the total mass of items sold at reserve strictly less than $b$ is at most $B$ ($Q^{<}(b) \leq B$). Otherwise, there must exist an auction with reserve $r < b$ whose supply is not fully exhausted, and any user currently paying $b$ could instead submit a bid of $r$ to that auction and win. This further establishes that every auction with reserve $r_i < b$ must have its supply exhausted (i.e.~receive a mass of $Q_i$ bids of $b$).
\end{itemize}

The first two conditions establish restrictions on $B$ as it relates to $D(\cdot)$, and the second two conditions establish restrictions on $B$ as it relates to $\{(Q_i,r_i)\}_{i}$. It is feasible for a mass of $B$ items to be sold at price $b$ only if: (a) $D^{>}(b) \leq B \leq D(b)$, and (b) $Q^{<}(b) \leq B \leq Q^{\leq}(b)$, which can be rephrased as:
$$B \in [\max\{D^{>}(b), Q^{<}(b)\}, \min\{D(b), Q^{\leq}(b)\}].$$

Moreover, an equilibrium exists where a mass of $B$ items are sold at price $b$ for any $B \in [\max\{D^{>}(b), Q^{<}(b)\}, \min\{D(b), Q^{\leq}(b)\}]$. Indeed, simply sort the bidders in decreasing order of value, and sellers in increasing order of reserve. Process the sellers in order, and when processing seller $j$, match a $Q_j$ mass of unmatched buyers to bid exactly $b$ in $j$'s auction. Let the process terminate once a total mass of $B$ users have been matched. Because $B \geq D^{>}(b)$, all unmatched users have value at most $B$. Because $B \geq Q^{<}(b)$, all auctions are either fully saturated with bids of $b$ or have a reserve exceeding $b$. Therefore, all unmatched users best respond by losing. Moreover, because $B \leq D(b)$, all matched users have value at least $b$. Therefore (and because all auctions are fully saturated with bids of $b$ or have a reserve exceeding $b$), all matched users best respond by winning and paying $b$. Because $B \leq Q^{\leq}(b)$, $b$ is indeed a winning bid at all auctions receiving bids. Therefore, this is an equilibrium. 

Therefore, to have an equilibrium at price $b$, the range $[\max\{D^{>}(b), Q^{<}(b)\}, \min\{D(b), Q^{\leq}(b)\}]$ must be non-empty. Moreover, when the range is non-empty, the maximum mass of items that can be cleared in equilibrium is $\min \{D(b), Q^{\leq}(b)\}$.

It remains to to establish for which $b$, the range $[\max\{D^{>}(b), Q^{<}(b)\}, \min\{D(b), Q^{\leq}(b)\}]$ is non-empty. Let $b_{\min}:=\inf\{b\ |\ D^{>}(b) \leq Q^{\leq}(b)\}$, and $b_{\max}:=\sup\{b\ |\ D(b) \geq Q^{<}(b)\}$. We claim that $[\max\{D^{>}(b), Q^{<}(b)\}, \min\{D(b), Q^{\leq}(b)]$ is non-empty if and only if $b \in [b_{\min},b_{\max}]$.

To see this, observe that $[\max\{D^{>}(b), Q^{<}(b)\}, \min\{D(b), Q^{\leq}(b)\}]$ is is non-empty if and only if both: (i) $D^{>}(b) \leq Q^{\leq}(b)$ and (ii) $D(b) \geq Q^{<}(b)$. Because $D^{>}(\cdot)$ is weakly decreasing and $Q^{\leq}(\cdot)$ is weakly increasing, (i) holds if and only if $b \geq b_{\min}$.\footnote{(i) clearly holds if $b > b_{\min}$ and clearly holds only if $b \geq b_{\min}$. (i) holds at $b=b_{\min}$ because both $D^{>}(\cdot)$ and $Q^{\leq}(\cdot)$ are right-continuous.} Because $D(\cdot)$ is weakly decreasing and $Q^{<}(\cdot)$ is weakly increasing, (ii) holds if and only if $b \leq b_{\max}$.\footnote{(ii) clearly holds if $b < b_{\max}$ and clearly holds only if $b \leq b_{\max}$. (ii) holds at $b=b_{\max}$ because both $D(\cdot)$ and $Q^{<}(\cdot)$ are left-continuous.} Therefore, both hold if and only if $b \in [b_{\min},b_{\max}]$. 

Finally, we claim that $p_{\min} \leq p_{\max}$. We claim that $D^{>}(p_{\max}) \leq Q^{\leq}(p_{\max})$, which guarantees $p_{\min} \leq p_{\max}$. Indeed, $D(p_{\max}+\varepsilon) < Q^{<}(p_{\max}+\varepsilon)$ for all $\varepsilon > 0$. Moreover, $D^{>}(p_{\max}) = \lim_{\varepsilon \rightarrow 0^+} D(p_{\max}+\varepsilon)$, and $Q^{\leq}(p_{\max}) = \lim_{\varepsilon \rightarrow 0^+} Q^{<}(p_{\max}+\varepsilon)$. Therefore, $D^{>}(p_{\max}) \leq Q^{\leq}(p_{\max})$, as desired. This completes the proof of the first two bullets, as it characterizes all possible equilibria (and therefore the canonical ones have minimal clearing price and clear maximal mass of items).

To see the third bullet, first observe that we have previously established that $Q^{<}(p_{\max}) \leq D(p_{\max}) = D^{>}(p_{\max}) \leq Q^{\leq}(p_{\max})$ (the middle inequality holds as $D(\cdot)$ is continuous). Therefore, we wish to show that the same inequalities do not hold for any $p \neq p_{\max}$. Indeed, by definition of $p_{\max}$, $D(b) < Q^{<}(b)$ for any $b > p_{\max}$. To reason about $b < p_{\max}$, observe that (i) $Q^{\leq}(b) \leq Q^{<}(p_{\max})$ for all $b < p_{\max}$, and also (ii) $D^{>}(b) = D(b) < D(p_{\max}) \leq Q^{\leq}(p_{\max})$. Therefore, $D(b) < Q^{<}(p_{\max})$ for all $b < p_{\max}$, and $p_{\max}$ is the unique solution to both inequalities (implying $p_{\min} = p_{\max}$. This concludes that $p_{\min} = p_{\max}$ is the unique clearing price in any equilibrium. 

To wrap up the third bullet, observe that $B_{\min}(p_{\max}) = D(p) = B_{\max}(p_{\max})$ (when $D(\cdot)$ is continuous) as $Q^{<}(p_{\max}) \leq D(p_{\max}) = D^{>}(p_{\max}) \leq Q^{\leq}(p_{\max})$.
\end{proof}

\subsection{Price-Setting in Simultaneous First-Price Auctions}\label{sec:firstpricesetting}

Proposition~\ref{prop:firstpriceeq} characterizes the equilibrium behavior of bidders in Simultaneous First-Price Auctions. Downstream Equilibria in our market structures concern the behavior of \emph{sellers} in Simultaneous First-Price Auctions (i.e.~choosing a quantity $Q_i$ according to whatever Upstream game occurs, and setting a reserve $r_i$). In this section, we discuss the process of reserve price-setting, assuming quantities are fixed.\footnote{Therefore, this section does not directly characterize Downstream Equilibria in any market structure we consider, but will be used as technical lemmas.}

Proposition~\ref{prop:maineq} is the main result of this section, and characterizes all potential price-setting equilibria. Intuitively, Bullet One corresponds to a ``market-clearing equilibrium'' where no seller is sufficiently large to profit from price-setting, and Bullet Two corresponds to the converse.

Before we proceed to prove \Cref{prop:maineq}, we will prove several useful Lemmas. We will use the notation $Q_{-i}^{\leq}(b):=\sum_{j\neq i,\ r_j \leq b} Q_j$ and $Q_{-i}^{<}(b):=\sum_{j\neq i,\ r_j < b} Q_j$. 

\begin{lemma}\label{lem:one}
    In any equilibrium of a Price-Setting Game with clearing price $p> 0$: (i) every seller is either saturated or a price-setter (or both), and (ii) every seller $i$ with $Q_i > 0$ and $p > c_i$ has strictly positive profit.
\end{lemma}

\begin{proof}[Proof of Lemma~\ref{lem:one}]
If $r_i < p$, then Seller $i$ must be saturated. If $r_i = p$, then Seller $i$ is a price-setter. 

If $r_i > p$, then Seller $i$ achieves zero profit. If $Q_i = 0$, then Seller $i$ is surely saturated. If $Q_i > 0$, then updating $r'_i\in (c_i,p)$ certainly generates non-zero profit. To see this, observe that: (a) this change cannot possibly lower the clearing price below $r'_i$,\footnote{Because $p > r'_i$, $D^{>}(b) \leq Q^{\leq}(b)$ for all $b < r'_i$, and updating $r_i$ to $r'_i$ does not change this.} and (b) this change must result in non-zero quantity of purchases from Seller $i$.\footnote{Because $p > r'_i$, $D^{>}(r'_i) > Q_{-i}^{\leq}(r_i)$. Therefore, if the new clearing price is $r'_i$, it is not possible for all demand to clear exclusively from sellers $\neq i$. If the new clearing price exceeds $r'_i$, then Seller $i$ must be saturated. (a) establishes that the new clearing price is at least $r'_i$.} Therefore, Seller $i$ can strictly profit, contradicting that this is a price-setting equilibrium.
\end{proof}

\begin{lemma}\label{lem:singleprice}
    In every price-setting equilibrium, at least one of the following must hold: (a) every seller is saturated, or (b) there is at most one price-setter.
\end{lemma}

\begin{proof}[Proof of Lemma~\ref{lem:singleprice}]
Recall by Lemma~\ref{lem:one} that all unsaturated sellers are price-setters, and that all price-setters earn strictly positive revenue. So assume for contradiction that there are multiple price-setters, each of whom earn strictly positive revenue, and at least one of whom is unsaturated. Refer to this seller as $i$, and let the clearing price be $p = r_i$. Observe that a mass of items strictly exceeding $Q^{<}_{-i}(p)$ is sold by sellers $\neq i$ -- every seller with $r_j < p$ is saturated, and the other price-setters sell non-zero supply. At most $D(p)$ total quantity of items are purchased. Therefore, the mass of items sold by Seller $i$ at most $D(p) - Q^{<}_{-i}(p) - \varepsilon$ for some $\varepsilon > 0$.

Consider if Seller $i$ were to update their reserve to $p - \delta$. This cannot possibly cause the clearing price to fall below $p-\delta$. This might cause the clearing price to remain $p$. In this case, Seller $i$ becomes saturated at the same clearing price, which is clearly a strict improvement (because $p > c_i$). It might also cause the clearing price to lower, but remain at least $p-\delta$. In this case, at least $D(p)$ total mass must be sold, at most $Q^{<}_{-i}(p)$ can come from other sellers. Therefore, the quantity sold by Seller $i$ must be at least $D(p)-Q^{<}_{-i}(p)$, a strict improvement in quantity sold by at least $\varepsilon$. Therefore, updating $r_i:=p-\delta$ for any $\delta > 0$ either results in increased sales at the same price of $p$, or an increase in sales of at least $\varepsilon$ with a decreased price of at most $\delta$. For $\delta < \frac{\varepsilon \cdot (p- c_i)}{Q_i}$, this is a strict improvement in profit.\footnote{To see this, let $X$ denote the supply originally sold. Then the original profit is $X(p-c_i)$ and the new profit is at least $X(p-c_i) +\varepsilon (p-c_i) - \delta (X+ \varepsilon)$. For $\delta < \frac{\varepsilon\cdot (p-c_i)}{Q_i}$, this is a strict improvement (as $X+\varepsilon \leq Q_i$).}
\end{proof}

\begin{lemma}\label{lem:saturatedeq}
    Let $Q:=\sum_i Q_i$. A saturated price-setting equilibrium exists if and only if $Q_i \leq \frac{(x - c_i) (Q - D(x))}{x - D^{-1}_{\sup}(Q)}$ for all $i$ and all $x> D^{-1}_{\sup}(Q)$. 
\end{lemma}

\begin{proof}[Proof of \Cref{lem:saturatedeq}]
Consider any candidate saturated equilibrium. Then a total of $Q$ items must be sold, and therefore the clearing price (and all reserves) is at most $D^{-1}_{\sup}(Q)$, and Seller $i$ receives payoff at most $Q_i \cdot (D^{-1}_{\sup}(Q) - c_i)$. Let us consider possible deviations of Seller $i$.

Seller $i$ could instead set a reserve $r_i > D^{-1}_{\sup}(Q)$. From here, there are again two possibilities:
\begin{itemize}
    \item Perhaps $D(r_i) < Q - Q_i$. In this case, the clearing price is below $r_i$, and Seller $i$ gets payoff $0$.
    \item Perhaps $D(r_i) \geq Q- Q_i$. In this case, the clearing price becomes $r_i$, and Seller $i$ clears $D(r_i) +Q_i - Q$ in any canonical equilibrium.
\end{itemize}

Therefore, Seller $i$ can deviate to earn revenue $(x - c_i) \cdot (D(x)+Q_i - Q)$ for any $x > D^{-1}_{\sup}(Q)$. If $(x - c_i) \cdot (D(x)+Q_i-Q) > Q_i \cdot (D^{-1}_{\sup}(Q) - c_i)$, this deviation is strictly profitable. This rearranges to: $Q_i > \frac{(x - c_i) \cdot (D(x)-Q)}{x - D^{-1}_{\sup}(Q)}$. This establishes the only-if portion of the lemma.

To see the if portion of the lemma, consider the strategy profile where each Seller sets $r_i = D^{-1}_{\sup}(Q)$. Then the clearing price is $D^{-1}_{\sup}(Q)$, and Seller $i$ earns revenue $Q_i \cdot (D^{-1}_{\sup}(Q) - c_i)$. By the previous work, Seller $i$ cannot profit by deviating to any $r_i > D^{-1}_{\sup}(Q)$. Seller $i$ also cannot profit by deviating to any $r_i < D^{-1}_{\sup}(Q)$ -- such a deviation cannot increase neither the clearing price nor the quantity of items sold by Seller $i$, and therefore cannot increase Seller $i$'s revenue.
\end{proof}

Now we are ready to prove \Cref{prop:maineq}.

\begin{proof}[Proof of \Cref{prop:maineq}]
    Bullet One is a restatement of Lemma~\ref{lem:saturatedeq}. Bullet Two requires analyzing what happens when saturated price-setting equilibria do not exist.

    By Lemma~\ref{lem:saturatedeq}, a saturated price-setting equilibrium does not exist if and only if $Q_i >\frac{(x - c_i) \cdot (Q-D(x))}{x - D^{-1}_{\sup}(Q)}$ for some $i$. Moreover, Lemma~\ref{lem:singleprice} establishes that every unsaturated equilibrium must have exactly one price-setter. 
    
    So consider a candidate price-setting equilibrium with $i$ as the price setter. Observe that for any price $x > D^{-1}_{\sup}(Q)$, Seller $i$ can earn revenue \emph{at least} $(x - c_i) \cdot (Q_i - Q+D(x))$ by setting price $x$. To see this, observe:
    \begin{itemize}
        \item If $D(x) \leq Q-Q_i$, then the claim holds vacuously.
        \item If $x > r_j$ for all $j \neq i$, then $i$ would be a price-setter at $r_i = x$, and earn exactly $(x-c_i) \cdot (Q_i-Q+D(x))$.
        \item If some $r_j > x$, then $i$ might be saturated (and earn at least $(x-c_i) \cdot Q_i$). Or, $i$ might remain a price-setter (and still sell at least $D(x)+Q_i - Q$, as only $Q-Q_i$ can be sold by other sellers). 
    \end{itemize}

    Therefore, if $i$ is best-responding, they must earn at least $\max_{x > D^{-1}_{\sup}(Q)}\{(x-c_i) \cdot (Q_i -Q+D(x))\}$. As a price-setter, this can only be achieved by setting price $\arg\max_{x > D^{-1}_{\sup}(Q)}\{(x-c_i) \cdot (Q_i -Q+D(x))\}$. This completes the characterization of all possible price-setting equilibria.

   To see the second half of Bullet Two, consider when $i^*$ sets price $\arg\max_{x \geq D^{-1}_{\sup}(Q)}\{(x-c_{i^*}) \cdot (Q_{i^*} - Q+D(x))\}$, and seller $j$ sets reserve $0$ for all $j \neq i^*$. Seller $i^*$ is indeed best responding, because they earn revenue at most $Q_{i^*} \cdot (D^{-1}_{\sup}(Q) - c_{i^*})$ by setting a reserve $\leq D^{-1}_{\sup}(Q)$, and earn $\max\{0,(x - c_{i^*}) \cdot (Q_i - Q+D(x))\}$ by being a price-setter at $x$. Because $Q_{i} > \frac{(x - c_i) \cdot (Q-D(x))}{x - D^{-1}_{\sup}(Q)}$ for some $i$, this means that $r^*_i > D^{-1}_{\sup}(Q)$ for some $i$, and therefore $r^*_{i^*} > D^{-1}_{\sup}(Q)$. Therefore, $i^*$ indeed optimizes their response by setting price $r^*_{i^*}:=\arg\max_{x \geq D^{-1}_{\sup}(Q)}\{(x-c_{i^*}) \cdot (Q_{i^*} - Q+D(x))\}$, and Seller $i^*$ is indeed best-responding. 

    To see that Seller $j$ is best-responding for all $j \neq i^*$, observe that they currently enjoy revenue $Q_j \cdot r^*_{i^*}$, and enjoy exactly this revenue with any reserve $x < r^*_{i^*}$. For any $x > r^*_{i^*}$, Seller $j$ would instead become the price-setter at reserve $x$ and earn revenue $\max\{0,(x-c_j) \cdot (Q_j - Q+D(x))\}$. But we know that:
    \begin{align*}
        &(r^*_{i^*} - c_j) \cdot Q_j - (x - c_j) \cdot (Q_j - Q+D(x)) \\
        &\qquad\qquad\geq (r^*_j-c_j) \cdot Q_j- (x - c_j) \cdot (Q_j - Q+D(x)) \\
        &\qquad \qquad\geq (r^*_{j} - c_j) \cdot (Q_j-Q+D(r^*_{j})) - (x - c_j) \cdot (Q_j - Q+D(x))\\
        &\qquad\qquad\geq 0
    \end{align*}

Above, the first inequality follows as $r^*_{i^*}\geq r^*_j$. The second follows as $r^*_j \geq c_j$ and $Q \geq D(r^*_j)$. The third follows as $r^*_j = \arg\max_{x \geq D^{-1}_{\sup}(Q)}\{(x-c_j)\cdot (Q_j - Q+D(r^*_j))\}$. 

To see that $\arg\max_i\{Q_i\} = i^*$ when $c_i = c_j = c$ for all $i,j$, consider any $i,j$ with $Q_i > Q_j$ and the corresponding $r^*_i, r^*_j$. Observe that:

\begin{align*}
    &(r^*_i-c)\cdot (Q_i - Q+D(r^*_i)) \geq (r^*_j-c)\cdot (Q_i - Q+D(r^*_j)),\text{ and}\\
    &(r^*_j-c)\cdot (Q_j - Q+D(r^*_j)) \geq (r^*_i-c)\cdot (Q_j - Q+D(r^*_i))\\
     \Rightarrow & (r_i^*-r_j^*)\cdot (Q_i-Q_j) \geq 0\\
     \Rightarrow & r_i^* \geq r_j^*
\end{align*}
\end{proof}

\section{Omitted Proofs from Section~\ref{sec:bitcoin}}\label{app:bitcoin}

\begin{proof}[Proof of Proposition~\ref{prop:bitcoinundominated}] 

Observe first that, for any $(\vec{q}_{-i}, \vec{r}_{-i})$, the only difference in Miner $i$'s payoff for using $(q_i, r_i)$ as opposed to $(q_i, r^*(D, Q_A))$ is their profit in the simultaneous first-price auctions (their total expenditure on resources is the same, and their block reward is the same). Moreover, Miner $i$ also receives the same quantity of \append{}s. 

Let us now analyze Miner $i$'s revenue under both strategies. Let $S_i(r_i)$ denote the quantity of \writeOp{}s sold when using $(q_i, r_i)$, and $S_i(r^*(D,Q_A))$ denote the quantity sold when using $(q_i, r^*(D, Q_A))$. Observe first that if $S_i(r_i) = 0$, then certainly $S_i(r^*(D, Q_A))\cdot (r^*(D, Q_A)-c^W)\geq 0 = S_i(r_i) \cdot (r_i-c^W)$. If $S_i(r_i) > 0$, then:
\begin{align*}
&\qquad\qquad\qquad S_i(r_i)\leq D(r_i) - Q^{\leq}_{-i}(r^*(D,Q_A)),\\
&\qquad\qquad\text{ and } S_i(r^*(D,Q_A))\geq D(r^*(D, Q_A)) - Q^{\leq}_{-i}(r^*(D, Q_A)).\\
&\Rightarrow (r^*(D, Q_A)-c^W)\cdot S_i(r^*(D, Q_A)) - (r_i - c^W)\cdot S_i(r_i) \\
&\qquad \geq (r^*(D, Q_A)-c^W)\cdot \left(D(r^*(D, Q_A)) - Q^{\leq}_{-i}(r^*(D, Q_A))\right) \\
&\qquad \qquad - (r_i - c^W) \cdot \left(D(r_i) - Q^{\leq}_{-i}(r^*(D,Q_A))\right)\\
&\qquad = (r^*(D, Q_A)-c^W)\cdot D(r^*(D, Q_A) - (r_i - c^W) \cdot D(r_i) \\
&\qquad\qquad + (r^*(D, Q_A)-r_i) \cdot Q^{\leq}_{-i}(r^*(D,Q_A))\\
&\qquad \geq (r^*(D, Q_A)-c^W)\cdot D(r^*(D, Q_A)) - (r_i - c^W) \cdot D(r_i) \\
&\qquad \geq 0.
\end{align*}
Above, the first line follows as, because $S_i(r_i) > 0$, the clearing price is at least $r_i$. Therefore, at most $D(r_i)$ \writeOp{}s are sold, and at least $Q_{-i}^{<}(r_i) \geq Q_i^{\leq}(r^*(D, Q_A))$ must be sold to Miners $\neq i$. The second line follows because any Miner $i$ setting a price of $r$ sells quantity at least $\max\{Q_A, D(r)\} - Q_{-i}^{\leq}(r) \geq D(r) - Q_{-i}^{\leq}(r)$. The third inequality and subsequent equality are basic algebra. The penultimate inequality follows as $r_i > r^*(D,Q_A)$. {The final inequality follows as $r^*(D,Q_A)$ optimizes $(r-c^W)\cdot D(r)$ over all $r \geq D^{-1}_{\inf}(Q_A)$.}

To see that $(q_i, r^*(D, Q_A)$ may sometimes give strictly larger payoff than $(q_i, r_i)$, consider the case that each other Miner invests $q_j = 0$. Then Miner $i$'s profit from end-users by setting reserve $r^*(D, Q_A)$ is $Q_A \cdot ( r^*(D, Q_A)-c^W) > D^{-1}(r_i)\cdot (r_i - c^W)$, which is the profit earned by setting reserve $r_i$. 

\end{proof}

\subsection{Omitted Proofs from Section~\ref{sec:bitcoinpure}}

\begin{proposition}\label{prop:bitcoinpricesetting} Let $(\vec{q},\vec{r})$ be an Equilibrium in the Distributed Ledger Model. Then one of the following holds:
\begin{itemize}
    \item Each Miner sells $Q_i := Q_A \cdot q_i/\sum_j q_j$ \writeOp{}s at a clearing price of $D^{-1}_{\sup}(Q_A)$, and $r_i \leq D^{-1}_{\sup}(Q_A)$ for all $i$. Further, $Q_i \leq \frac{(x - c^W) (Q_A-D(x))}{x - D^{-1}_{\sup}(Q_A)}$ for all $i$ and all $x > D^{-1}_{\sup}(Q_A)$. 
    \item There is a single price-setter $i$, who sets price $r_i:= \arg\max_{x > D^{-1}_{\sup}(Q)}\{(x-c^W) \cdot (D(x)+Q_i-Q)\}$.
\end{itemize}
\end{proposition}

\begin{proof}[Proof of \Cref{prop:bitcoinpricesetting}] The proof follows immediately from Proposition~\ref{prop:maineq}. Indeed, in order for $(q_i,r_i)$ to be a best response to $(\vec{q}_{-i},\vec{r}_{-i})$, it must be that $r_i$ optimizes Miner $i$'s payoff after fixing $q_i,\vec{q}_{-i},\vec{r}_{-i}$. Therefore, in any equilibrium it must hold simultaneously for all $i$ that $r_i$ optimizes Miner $i$'s payoff after fixing $q_i,\vec{q}_{-i},\vec{r}_{-i}$. 

Observe that this condition fixes $\vec{q}$ and therefore $\vec{Q}$, and asks that $r_i$ simultaneously optimize Miner $i$'s payoff in response to $\vec{r}_{-i}$. This is exactly asking for an equilibrium of the price-setting game parameterized by $D(\cdot)$ and $\vec{Q}$, and its equilibria are characterized in Proposition~\ref{prop:maineq}.
\end{proof}

Next, we establish that the market shares of each Miner must match those in a Tullock contest. We first recap known characterizations of Tullock contest equilibria for comparison.

\begin{proposition}\label{prop:bitcoinmustbeTullock}
Let $(\vec{q},\vec{r})$ be an Equilibrium in the Distributed Ledger Model, and let the clearing price for End-Users be $r$. Then:
\begin{itemize}
    \item $\sum_j q_j = Q_A \cdot (r+B-c^W)/c^*(\vec{c}^R)$.
    \item For all Miners $i$, $\frac{q_i}{\sum_j q_j} = x_i^*(\vec{c}^R)$.
\end{itemize}
\end{proposition}

\begin{proof}[Proof of Proposition~\ref{prop:bitcoinmustbeTullock}] 

Recall that the goal of Proposition~\ref{prop:bitcoinmustbeTullock} is to propose \emph{necessary} conditions on any equilibrium $(\vec{q},\vec{r})$ of the Distributed Ledger Model. Therefore, it suffices to consider, for example, local optimality conditions (which are necessary, but not sufficient). Let $r$ denote the clearing price at the candidate equilibrium $(\vec{q},\vec{r})$.

To this end, we focus the optimization problem facing a particular Miner $i$, and consider deviations of the following form:
\begin{itemize}
    \item Miner $i$ will only consider deviations that \emph{do not change the clearing price $r$} and \emph{do not change the unsaturated price-setter} (if there is one).
    \item Therefore, if $r = D^{-1}(Q_A)$, we will consider deviations for Miner $i$ from $(q_i,r_i)$ to $(q'_i,0)$ (and all deviations of this form will be considered).\footnote{Observe that if $r = D^{-1}(Q_A)$, it must be the case that all Miners $j$ with $q_j > 0$ have $r_j \leq D^{-1}(Q_A)$, and therefore the clearing price will remain $D^{-1}(Q_A)$ so long as Miner $i$ deviates to some $r_i \leq D^{-1}(Q_A)$ as well.}
    \item If $r > D^{-1}(Q_A)$, then Proposition~\ref{prop:bitcoinpricesetting} establishes that there is a single price setter $i^*$. In order for $i^*$ to be a price-setter, it must be that $Q_{i^*} > Q_A - D(r)$, and also that $r_i < r$ for all $i\neq i^*$ with $q_i > 0$. Therefore, for Miner $i^*$, we will consider deviations of the form $(q'_{i^*}, r)$ \emph{that still result in $Q'_{i^*} > Q_A$}. Observe that there exists a sufficiently small $\delta>0$ such that deviations of the form $(q_{i^*}+x,r)$ satisfies this property for all $x \in (-\delta,\delta)$. For Miner $i\neq i^*$, we will consider deviations of the form $(q'_i, 0)$ that also \emph{still result in $Q'_{i^*} > Q_A$}. Observe again that there exists a sufficiently small $\delta > 0$ such that deviations of the form $(q'_i,0)$ satisfy this property for all $q'_i \in [0,q_i+\delta)$. 
    \item In conclusion, we will only ever consider deviations that do not change the clearing price and do not change the unsaturated price-setter (if there is one). However, the above bullets note there is sufficient flexibility in choosing such deviations that local optimality conditions on the choice of $q_i$ must hold.
\end{itemize}

Now, we consider local optimality conditions for deviations of the prescribed type for a particular Miner $i$. The proof essentially breaks into two (interleaved) parts: (a) we repeat calculations identical to those in~\cite{ArnostiW22} for analyzing equilibria of Tullock Contests, and (b) we confirm that the same local optimality conditions must hold for any equilibrium in the Distributed Ledger Model.

So, consider the function $x_i(q_i;q_{-i}):=q_i/(q_i +\sum_{j \neq i} q_j)$, which determines the fraction of the $Q_A$ \append{}s won by Miner $i$ as a function of $q_i$ after fixing $q_{-i}$. We compute (identically to~\cite{ArnostiW22}):

$$\frac{\partial x_i(q_i;q_{-i})}{\partial q_i} =\frac{1}{q_i + \sum_{j \neq i} q_j} - \frac{q_i}{\left(q_i + \sum_{j\neq i} q_j\right)^2}= \frac{1-x_i(q_i;q_{-i})}{q_i + \sum_{j \neq i} q_j}.$$

Now, \emph{in the range where the clearing price remains $r$ and the unsaturated price-seller (if one exists) remains $i^*$}, Miner $i\neq i^*$'s payoff for investing $q_i$ is: $P_i(q_i;q_{-i}):=Q_A \cdot x_i(q_i;q_{-i})\cdot (r-c^W+B) - c_i^R \cdot q_i$. Therefore, its derivative \emph{in this range} is:

\begin{align*}
    \frac{\partial P_i(q_i;q_{-i})}{\partial q_i} &= Q_A \cdot (r-c^W+B)\cdot \frac{\partial x_i(q_i;q_{-i})}{\partial q_i} - c_i^R\\
    &= Q_A\cdot (r-c^W+B) \cdot \frac{1-x_i(q_i;q_{-i})}{q_i + \sum_{j\neq i} q_j} - c_i^R
\end{align*}

If there is a price-setter $i^*$, then Miner $i^*$'s payoff \emph{in the prescribed range} is: $P_{i^*}(q_{i^*};q_{-i^*}):=Q_A \cdot x_{i^*}(q_{i^*};q_{-i^*}) \cdot (r-c^W+B) - (Q_A-D(r))\cdot (r-c^W)- c_R^{i^*} \cdot q_{i^*}$. Therefore, its derivative \emph{in this range} is:

\begin{align*}
    \frac{\partial P_{i^*}(q_{i^*};q_{-i^*})}{\partial q_{i^*}} &= Q_A \cdot (r-c^W+B)\cdot \frac{\partial x_{i^*}(q_{i^*};q_{-i^*})}{\partial q_{i^*}} - c_{i^*}^R\\
    &= Q_A\cdot (r-c^W+B) \cdot \frac{1-x_{i^*}(q_{i^*};q_{-i^*})}{q_{i^*} + \sum_{j\neq i^*} q_j} - c_{i^*}^R
\end{align*}

In particular, \emph{in this range}, the partial derivatives are the same. Now, consider any candidate equilibrium $(\vec{q},\vec{r})$ with clearing price $r$. 

\begin{itemize}
    \item If $i^*$ is a price-setter, then we must have $Q_{i^*} > Q_A - D(r)$ and $r_{i^*} = r$. Moreover, as long as $Q'_{i^*} > Q_A - D(r)$ and $r'_{i^*} = r$, $i^*$ will remain a price-setter and have payoff $P_{i^*}(q_{i^*};q_{-i^*})$ as defined above. Therefore, unless $\frac{\partial P_{i^*}(q_{i^*};q_{-i^*})}{\partial q_{i^*}} = 0$, there exists a sufficiently small $\varepsilon$ such that $(q_i\pm \varepsilon,r)$ is a strictly better response than $(q_i, r)$. We conclude that for any price-setter, $\frac{\partial P_{i^*}(q_{i^*};q_{-i^*})}{\partial q_{i^*}} = 0$ is a necessary condition for $(\vec{q},\vec{r})$ to be an equilibrium.
    \item If $i$ is not a price-setter, then any deviation of the form $(q'_i,0)$ for $q'_i \leq q_i$ maintains both the clearing price and the identity of the price-setter (if one exists), along with any deviation of the form $(q_i+\varepsilon,0)$ for sufficiently small $\varepsilon$. Therefore, it must either hold that (a) $\frac{\partial P_{i}(q_{i};q_{-i})}{\partial q_{i}} = 0$ or (b) $q_i = 0$ and $\frac{\partial P_{i}(q_{i};q_{-i})}{\partial q_{i}} \leq 0$.
    \item Together, we conclude that for all $i$, a necessary condition for $(\vec{q},\vec{r})$ to be an equilibrium is that (a) $\frac{\partial P_{i}(q_{i};q_{-i})}{\partial q_{i}} = 0$ or (b) $q_i = 0$ and $\frac{\partial P_{i}(q_{i};q_{-i})}{\partial q_{i}} \leq 0$. Rewriting the partial derivatives computed above, (a) holds if and only if $x_i(q_i;\vec{q}_{-i}) = 1-\frac{c_i^R \cdot \sum_j q_j}{Q_A \cdot (r+B-c^W)}$. (b) holds if and only if $1-\frac{c_i^R \cdot \sum_j q_j}{Q_A \cdot (r+B-c^W)} \leq 0$. Therefore, the proposed necessary conditions are indeed necessary for $(\vec{q},\vec{r})$ to be an equilibrium.
\end{itemize}
\end{proof}

\begin{proof}[Proof of Theorem~\ref{thm:bitcoinpure}]
The proof simply combines Propositions~\ref{prop:bitcoinpricesetting} and~\ref{prop:bitcoinmustbeTullock}.
\end{proof}

\subsection{Omitted Proofs from Section~\ref{sec:bitcoinsufficient}}\label{app:bitcoinsufficient}

Throughout this section, we will find a sufficient condition for an equilibrium that clears $Q_A$ \writeOp{}s with a block reward of $0$. By Proposition~\ref{prop:blockrewards}, this condition suffices for the same equilibrium to hold with any block reward.

\begin{definition} For $k(\cdot)$, we say that a quantity $Q$ \emph{$k(\cdot)$-covers} a Demand Curve $D(\cdot)$ and \writeOp{} cost $c^W$ if $k(x) \cdot Q\cdot (D^{-1}_{\sup}(Q)-c^W) \geq (x\cdot Q) \cdot (D^{-1}_{\sup}(x \cdot Q)-c^W)$ for all $x \in (0,1)$, and $k(1) = 1$. 

We say that a quantity $Q$ \emph{exactly $k(\cdot)$-covers} $D(\cdot),c^W$ if $k(x) \cdot Q\cdot (D^{-1}_{\sup}(Q)-c^W) = (x\cdot Q) \cdot (D^{-1}_{\sup}(x \cdot Q)-c^W)$ for all $x \in (0,1)$, and $k(1) = 1$. 
\end{definition}

Intuitively, $Q$ $k(\cdot)$-covers $D(\cdot), c^W$ if the total revenue earned selling quantity $Q$ guarantees some fraction of the total revenue that could be earned selling quantity $x\cdot Q$ instead, with the precise coverage required parameterized by $x$. $Q$ exactly $k(\cdot)$-covers $D(\cdot),c^W$ if $k(\cdot)$ is the tightest possible coverage. Note that $k(x) \geq x$ for all $x\in (0,1)$, as $D^{-1}_{\sup}(x\cdot Q) \geq D^{-1}_{\sup}(Q)$ for all $x \in (0,1)$.

\begin{lemma}\label{lem:bitcoinregular}
    Let $D(\cdot)$ be Regular. Then for all $Q \leq D(0)$, $Q$ $\frac{D(0)/Q-x}{D(0)/Q-1}$-covers $D(\cdot), c^W$.
\end{lemma}
\begin{proof}
    Consider the function $R_{c^W}(Q):=Q \cdot (D^{-1}(Q)-c^W)$. Then $R'_{c^W}(Q) = \varphi_D(D^{-1}(Q)) - c^W$. Because $D(\cdot)$ is Regular, $R'_{c^W}(\cdot)$ is decreasing, and therefore $R_{c^W}(\cdot)$ is concave. 

    Observe that $Q = \frac{D(0)/Q-1}{D(0)/Q-x}\cdot x \cdot Q + \frac{1-x}{D(0)/Q-x}\cdot D(0)$. Because $R_{c^W}(\cdot)$ is concave:
    \begin{align*}
R_{c^W}(Q) &\geq \frac{D(0)/Q -1}{D(0)/Q-x} \cdot R_{c^W}(x\cdot Q) + \frac{1-x}{D(0)/Q-x}\cdot R_{c^W}(D(0))\\
&\geq \frac{D(0)/Q -1}{D(0)/Q-x} \cdot R_{c^W}(x\cdot Q),
    \end{align*}
    confirming that $Q$ $\frac{D(0)/Q-x}{D(0)/Q-1}$-covers $D(\cdot), c^W$.
\end{proof}

Lemma~\ref{lem:bitcoinregular} allows us to conclude $k(\cdot)$-coverage immediately from the fact that $D(\cdot)$ is Regular, although for most Regular $D(\cdot)$ a tighter bound is possible.

Now, we argue that when $Q_A$ sufficiently-covers $D(\cdot)$, it is an Equilibrium for each miner to set $r_i = D^{-1}_{\sup}(Q_A)$, $\sum_j q_j = Q_A \cdot (D^{-1}_{\sup}(Q_A)-c^W)/c^*(\vec{c}^R)$, and $q_i/\sum_j q_j = x_i^*(\vec{c}^R)$ for all $i$ (i.e.~the potential equilibrium described by Bullet One in the second half of Theorem~\ref{thm:bitcoinpure}). For simplicity of notation in the rest of this section, we refer to equilibrium investments as $\vec{q}^*$, the equilibrium quantity of \append{}s won by each miner as $\vec{Q}^*$, the equilibrium fraction of \append{}s won as $\vec{x}^*$ (where $x_i^* = Q_i^*/Q_A$), and $c^*:=c^*(\vec{c}^R)$. We further use the notation $\rew(Q_A):=Q_A\cdot (r-c^W)$.

First, we analyze the investment cost Miner $i$ must pay in order to win a $(1-y)$ fraction of \append{}s against $\vec{q}_{-i}^*$.

For the rest of this section, we will leverage the following observation. For all Miners $i$ with $x_i^* = 0$, we are hoping to show that their best response is to maintain $q_i = 0$. Certainly, if their best response is to maintain $q_i = 0$ \emph{even if their cost were lowered to $c^*$}.\footnote{Recall that $x_i^* = \max\{0,1-\frac{c_i^R}{c^*}\}$, therefore, the cost of all Miners with $x_i^* = 0$ is at least $c^*$.} Therefore, if we can show that $(\vec{q}^*,\vec{r}^*)$ is an Equilibrium \emph{even when all Miners with $x_i^*=0$ have $c_i^R = c^*$}, then we will have established that $(\vec{q}^*,\vec{r}^*)$ is an Equilibrium even when non-participating Miners have higher costs. 

\begin{lemma}\label{lem:bitcoincost} Let $x^*_i > 0$, or $x_i^* = 0$ and $c_i^R = c^*$. Then in order to win $(1-y) \cdot Q_A$ \append{}s, against strategy profile $\vec{q}^*_{-i}$, Miner $i$ must invest $(1/y-1)\cdot (1-x^*_i)^2\cdot \rew(Q_A)$.   
If $x_i^* = 0$, then in order to win $(1-y) \cdot Q_A$ \append{}s, against strategy profile $\vec{q}^*_{-i}$, Miner $i$ must invest at least $(1/y-1)\cdot (1-x^*_i)^2\cdot \rew(Q_A)$.
\end{lemma}
\begin{proof}[Proof of Lemma~\ref{lem:bitcoincost}]
    By definition, the total Resources purchased by Miners $\neq i$ is $\rew(Q_A)\cdot (1-x^*_i) /c^*$. Therefore, in order to win a $(1-y)$ fraction of the market against these Resources, Miner $i$ must invest such that Miners $\neq i$ Resources become a $y$ fraction of the total Resources. Therefore, Miner $i$ must purchase $(1/y-1)\cdot \rew(Q_A)\cdot (1-x^*_i)/c^*$ Resources.\footnote{To quickly see that the calculation is correct, observe that this results in a total Resources of $(1/y) \cdot \rew(Q_A)\cdot (1-x^*_i)/c^*$, of which $\rew(Q_A)\cdot (1-x^*_i)/c^*$ is a $y$ fraction.} 

    Moreover, a Miner $i$ pays a cost of $c_i^R$ per Resource, meaning that Miner $i$ must invest $(1/y-1)\cdot \rew(Q_A)\cdot (1-x^*_i)\cdot \frac{c^R_i}{c^*}$. 
    
    Finally, if $x_i^*>0$, $x_i^* = 1-c^R_i/c^*$ and therefore $c_i^R/c^* = 1-x_i^*$ (and if $x_i^* = 0$, this holds by hypothesis). So we conclude a total investment of $(1/y-1)\cdot \rew(Q_A)\cdot (1-x^*_i)^2$, as desired.
\end{proof}

Now, we observe that Miner $i$'s strategy space consists of the following two decisions (made jointly): (a) pick a price $D^{-1}_{\sup}(x \cdot Q_A)$ to set, (b) pick a fraction $y$ to win $(1-y)\cdot Q_A$ \append{}s. After both choices are made, Miner $i$ earns revenue $(x-y) \cdot Q_A \cdot D^{-1}_{\sup}(x \cdot Q_A)$. Therefore, we get the following lemma:

\begin{lemma}\label{lem:bitcoinpureprofit} Let $x_i^*>0$ or $x_i^* = 0$ and $c_i^R = c^*$. Then for every strategy $(q_i,r_i)$ that Miner $i$ can use against $(\vec{q}_{-i}^*,\vec{r}_{-i}^*)$, there exists a $z \leq 1$ and $y \leq z$ such that:
$$P_i((q_i,r_i);(\vec{q}_{-i}^*,\vec{r}_{-i}^*)) = \left(1-\frac{y}{z}\right) \cdot (z \cdot Q_A) \cdot (D^{-1}_{\sup}(z \cdot Q_A)-c^W) - (1/y-1)\cdot(1-x_i^*)^2 \cdot \rew(Q_A).$$

Moreover, for every $z \leq 1$ and $y \leq z$, there exists a strategy guaranteeing Miner $i$ payoff exactly $\left(1-\frac{y}{z}\right) \cdot (z \cdot Q_A) \cdot (D^{-1}_{\sup}(z \cdot Q_A)-c^W) - (1/y-1)\cdot(1-x_i^*)^2 \cdot \rew(Q_A).$

Therefore, $(\vec{q}^*,\vec{r}^*)$ is an Equilibrium if and only if the function $\left(1-\frac{y}{z}\right) \cdot (z \cdot Q_A) \cdot (D^{-1}_{\sup}(z \cdot Q_A)-c^W)  - (1/y-1)\cdot(1-x_i^*)^2 \cdot \rew(Q_A)$ is optimized at $z = 1$ and $y = 1-x_i^*$.
    
\end{lemma}

\begin{proof}
    Ultimately, Miner $i$ makes some investment that wins $(1-y)\cdot Q_A$ \append{}s for some $y \in [0,1]$. The total cost of doing so, by Lemma~\ref{lem:bitcoincost} is $(1/y-1)\cdot(1-x_i^*)^2 \cdot \rew(Q_A)$.

    Ultimately, Miner $i$ also sets some price of the form $D^{-1}_{\sup}(z\cdot Q_A)$, for $z \in [0,1]$. This implies that a total quantity of $z \cdot Q_A$ \writeOp{}s are sold, of which $\max\{0,z-y\}\cdot Q_A$ are sold by Miner $i$. Therefore, Miner $i$'s total payoff is $\max\{0,z-y\}\cdot Q_A \cdot (D^{-1}_{\sup}(z \cdot Q_A)-c^W) - (1/y-1)\cdot(1-x_i^*)^2 \cdot \rew(Q_A)$. 

    If Miner $i$ happens to choose $y \leq z$, this matches the desired form. If not, observe that $\max\{0,z-y\}\cdot Q_A \cdot (D^{-1}_{\sup}(z \cdot Q_A)-c^W) = 0 = (y-y) \cdot Q_A \cdot (D^{-1}_{\sup}(y \cdot Q_A) - c^W)$. Therefore Miner $i$'s payoff matches the desired form after updating $z:=y$. This completes the proof (after observing that $(z-y) = (1-y/z) \cdot z$).

    To see the ``Moreover' portion of the lemma, simply observe that Miner $i$ can indeed pick any $D^{-1}_{\sup}(z \cdot Q_A)$ as a price to set, and any $y \leq z$ as a fraction of \append{}s to leave for other Miners, inducing the prescribed payoff. 

    To see the `Therefore' portion of the lemma, simply observe that $z = 1$ and $y = 1-x_i^*$ corresponds to $(q_i^*,r_i^*)$.
\end{proof}

From now on, we will use the function $P_i(y,z)$ to denote the payoff $P_i( (q_i,r_i);(\vec{q}_{-i}^*,\vec{r}_{-i}^*))$ of the strategy $(q_i,r_i)$ that wins a $(1-y)$ fraction of \append{}s and sets price $D^{-1}_{\sup}(z\cdot Q_A)$.

\begin{corollary}\label{cor:optimize} Let $Q_A$ $k(\cdot)$-cover $D(\cdot), c^W$. Then:
\begin{align*}
    P_i(y,z) &\leq \left(1-\frac{y}{z}\right) \cdot k(z) \cdot Q_A \cdot (D^{-1}(Q_A)- c^W) - (1/y-1)\cdot(1-x_i^*)^2 \cdot \rew(Q_A)\\
    &= \rew(Q_A) \cdot\left( \left(1-\frac{y}{z}\right) \cdot k(z) - (1/y-1)\cdot (1-x_i^*)^2\right) 
\end{align*}

with equality at $z = 1, y = 1-x^*_i$. Therefore, $(\vec{q}^*,\vec{r}^*)$ is an Equilibrium if, for all $i$, the function $\rew(Q_A) \cdot\left( \left(1-\frac{y}{z}\right) \cdot k(z) - (1/y-1)\cdot (1-x_i^*)^2\right)$ is optimized at $z = 1$ and $y = 1-x_i^*$.

Moreover, if $Q_A$ exactly $k(\cdot)$-covers $D(\cdot),c^W$, then:
\begin{align*}
    P_i(y,z) &= \rew(Q_A) \cdot\left( \left(1-\frac{y}{z}\right) \cdot k(z) - (1/y-1)\cdot (1-x_i^*)^2\right) 
\end{align*}
Therefore, $(\vec{q}^*,\vec{r}^*)$ is an Equilibrium if and only if, for all $i$, the function $P_i(y,z)$ is optimized at $z = 1$ and $y = 1-x_i^*$.
\end{corollary}

\begin{proof}[Proof of Corollary~\ref{cor:optimize}]
The proof follows immediately from Lemma~\ref{lem:bitcoinpureprofit} after substituting the definition of $k(\cdot)$-cover and exactly $k(\cdot)$-cover, and that $k(z) = 1$.  
\end{proof}

From here, we simply optimize the function provided in Corollary~\ref{cor:optimize}. We begin by optimizing $y$ as a function of $z$.

\begin{lemma}\label{lem:bitcoinoptimizey} For any $z$ and $y \leq z$:
\begin{align*}
&\rew(Q_A) \cdot\left( \left(1-\frac{y}{z}\right) \cdot k(z) - (1/y-1)\cdot (1-x_i^*)^2\right) \\
& \leq \rew(Q_A) \cdot \left(k(z) - 2(1-x_i^*)\cdot \sqrt{k(z)/z}  +(1-x_i^*)^2\right)
\end{align*}

with equality at $y = (1-x_i^*)\cdot \sqrt{\frac{z}{k(z)}}$ (which in particular implies equality at $y=1$ and $z = 1-x_i^*$).
\end{lemma}

\begin{proof}[Proof of Lemma~\ref{lem:bitcoinoptimizey}]
Simply take the derivative with respect to $y$. We get:
\begin{align*}
    &\frac{\partial \left( \rew(Q_A) \cdot\left( \left(1-\frac{y}{z}\right) \cdot k(z) - (1/y-1)\cdot (1-x_i^*)^2\right)\right) }{\partial y}\\
    &= -\frac{k(z)}{z} \cdot \rew(Q_A) + (1-x^*_i)^2\cdot \rew(Q_A)/y^2
\end{align*}

Observe that the derivative is decreasing in $y$. Therefore, the maximum is achieved when the derivative is $0$. This occurs when:
\begin{align*}
    &-\frac{k(z)}{z} \cdot \rew(Q_A) + (1-x^*_i)^2\cdot \rew(Q_A)/y^2 = 0\\
    \Rightarrow &y = (1-x_i^*)\cdot \sqrt{\frac{z}{k(z)}}.
\end{align*}

This completes the proof of the core lemma. To see that equality holds at $z=1, y=1-x^*_i$, simply observe that $k(1) = 1$, and therefore the RHS above simplifies when substituting $z=1$ to $(1-x_i^*)$. Therefore, the LHS and RHS in the lemma statement are identical after substituting $z=1$ and $y=(1-x_i^*)$ to both sides.

To see the simplification in the statement, simply substitute $y = (1-x^*_i)\cdot \sqrt{\frac{z}{k(z)}}$ as below:
\begin{align*}
   & \rew(Q_A) \cdot\left( \left(1-\frac{y}{z}\right) \cdot k(z) - (1/y-1)\cdot (1-x_i^*)^2\right) \\
    &\qquad \leq \rew(Q_A) \cdot\left( \left(1-\frac{(1-x_i^*)\cdot \sqrt{z/k(z)}}{z}\right) \cdot k(z) - \left(\frac{1}{(1-x_i^*)\cdot \sqrt{z/k(z)}}-1\right)\cdot (1-x_i^*)^2\right)\\
    &\qquad = \rew(Q_A) \cdot \left(k(z) - (1-x_i^*)\cdot \sqrt{k(z)/z} -(1-x_i^*)\cdot \sqrt{k(z)/z} +(1-x_i^*)^2\right)\\
    &\qquad = \rew(Q_A) \cdot \left(k(z) - 2(1-x_i^*)\cdot \sqrt{k(z)/z}  +(1-x_i^*)^2\right)
\end{align*}

\end{proof}

\begin{definition}
    From now on, we define
    \begin{align*}
L_i(z)&= P_i\left((1-x_i^*)\cdot \sqrt{\frac{z}{k(z)}}, z\right) = \rew(Q_A) \cdot \left(k(z) - 2(1-x_i^*)\cdot \sqrt{k(z)/z}  +(1-x_i^*)^2\right)
\end{align*}
\end{definition}

\begin{corollary}
    Let $Q_A$ $k(\cdot)$ cover $D(\cdot), c^W$. Then $(\vec{q}^*,\vec{r}^*)$ is an equilibrium if $L_i(z)$ is optimized at $z=1$. If $Q_A$ exactly $k(\cdot)$ covers $D(\cdot), c^W$, then $(\vec{q}^*,\vec{r}^*)$ is an equilibrium if and only if $L_i(z)$ is optimized at $z=1$. 
\end{corollary}
\begin{proof}
    If $L_i(z)$ is optimized at $z=1$, then we can conclude the following:
\begin{align*}
    P_i(y,z) &\leq \left(1-\frac{y}{z}\right) \cdot k(z) \cdot Q_A \cdot (D^{-1}(Q_A)- c^W) - (1/y-1)\cdot(1-x_i^*)^2 \cdot \rew(Q_A)\\
    &\leq L_i(z)\\
    &\leq L_i(1)\\
    &= P_i(\vec{q}^*,\vec{r}^*).
\end{align*}
      Above, the first line follows from Corollary~\ref{cor:optimize}, for some $y \leq z \leq 1$. The second line follows from Lemma~\ref{lem:bitcoinoptimizey}. The third line follows by assumption. The fourth line follows by the `with equality' portions of Corollary~\ref{cor:optimize} and Lemma~\ref{lem:bitcoinoptimizey}.

      If we further have that $Q_A$ exactly $k(\cdot)$ covers $D(\cdot),c^W$, and $L_i(z)$ is not optimized at $z=1$, we conclude that for whatever $z$ $L_i(z) > L_i(1)$ it holds:
\begin{align*}
    P_i\left((1-x_i^*)\cdot \sqrt{\frac{z}{k(z)}},z\right) &=  L_i(z)\\
    &> L_i(1)\\
    &= P_i(\vec{q}^*,\vec{r}^*).
\end{align*}

Above, the first equality holds by Lemma~\ref{lem:bitcoinoptimizey} and Corollary~\ref{cor:optimize}. The second line follows by assumption that $L_i(z)$ is not optimized at $z=1$. The third line follows by the `with equality' portions of Corollary~\ref{cor:optimize} and Lemma~\ref{lem:bitcoinoptimizey}.
      
\end{proof}

\begin{lemma}\label{lem:bitcoinmain2}
    $L_i(z) \leq L_i(1)$ for all $z \in [0,1]$ if and only if $1-x_i^* \geq \max_{z \in [0,1]}\lt\{\frac{k(z)-1}{2\cdot \left(\sqrt{k(z)/z}-1\right)}\rt\}$. Therefore, if for all $i$, $1-x_i^* \geq \max_{z \in [0,1]}\left\{\frac{k(z)-1}{2\cdot \left(\sqrt{k(z)/z}-1\right)}\right\}$, and $Q_A$ $k(\cdot)$-covers $D(\cdot),c^W$, $(\vec{q}^*,\vec{r}^*)$ is an Equilibrium.

    If $Q_A$ exactly $k(\cdot)$-covers $D(\cdot),c^W$, then $(\vec{q}^*,\vec{r}^*)$ is an Equilibrium if and only if for all $i$, it holds that: $1-x_i^* \geq \max_{z \in [0,1]}\left\{\frac{k(z)-1}{2\cdot \left(\sqrt{k(z)/z}-1\right)}\right\}$.
\end{lemma}
\begin{proof}[Proof of Lemma~\ref{lem:bitcoinmain2}]
\begin{align*}
    L_i(1) - L_i(z) &=\rew(Q_A) \cdot \left(k(1) - 2(1-x_i^*)\cdot \sqrt{k(1)/1)}+(1-x_i^*)^2\right)\\
    &\qquad + \rew(Q_A)\cdot \left(- k(z) +2(1-x_i^*)\cdot \sqrt{k(z)/z} - (1-x_i^*)^2 \right)\\
    &= \rew(Q_A) \cdot \left(1-k(z) +2(1-x_i^*) \cdot \left(\sqrt{k(z)/z}-1\right)\right)
\end{align*}

In particular, $L_i(1) - L_i(z) \geq 0$ if and only if:
\begin{align*}
    &1-k(z) +2(1-x_i^*) \cdot \left(\sqrt{k(z)/z}-1\right) \geq 0\\
   \Leftrightarrow & 1-x_i^* \geq \frac{k(z)-1}{2\cdot \left(\sqrt{k(z)/z}-1\right)}
\end{align*}
\end{proof}

\begin{corollary} \label{cor:bitcoin-sufficient-largest-miner}
    Let $D(0) = k \cdot Q_A$. Then as long as $x^*_i \leq 1-\frac{1}{k-1}$, $(\vec{q}^*,\vec{r}^*)$ is an Equilibrium.
\end{corollary}
\begin{proof}
    We simply plug into Lemma~\ref{lem:bitcoinregular} and Lemma~\ref{lem:bitcoinmain2}. Lemma~\ref{lem:bitcoinregular} asserts that $Q_A$ $\frac{k-z}{k-1}$-covers $D(\cdot), c^W_i$. Therefore, Lemma~\ref{lem:bitcoinmain2} concludes the desired equilibrium so long as $1-x_i^* \geq \frac{\frac{k-z}{k-1}-1}{2\cdot \left(\sqrt{\frac{k-z}{z\cdot (k-1)}}-1\right)}$ for all $z$. We observe that:
    \begin{align*}
        \frac{\frac{k-z}{k-1}-1}{2\cdot \left(\sqrt{\frac{k-z}{z\cdot (k-1)}}-1\right)} & \leq \frac{\frac{k-z}{k-1}-1}{2\cdot \left(\sqrt{\frac{k-1}{z\cdot (k-1)}}-1\right)}\\
        &=\frac{1}{2(k-1)}\cdot \frac{1-z}{1/\sqrt{z}-1}\\
        &=\frac{1}{2(k-1)}\cdot \frac{(1-\sqrt{z})\cdot (1+\sqrt{z})}{(1-\sqrt{z})/\sqrt{z}}\\
        &=\frac{\sqrt{z}+z}{2(k-1)}\\
        &\leq \frac{1}{k-1}
    \end{align*}
    Above, the first inequality follows as $z \in [0,1]$. The three equalities are basic algebra. The final inequality follows again as $z \in [0,1]$. 
\end{proof}

\begin{proof}[Proof of Theorem~\ref{thm:bitcoinsufficient}]
The proof follows immediately from the definition of exactly $k(\cdot)$-covers and Lemma~\ref{lem:bitcoinmain2}, and Corollary~\ref{cor:bitcoin-sufficient-largest-miner}
\end{proof}

\subsection{Omitted Proofs from Section~\ref{sec:bitcoinblock}}\label{app:bitcoinblock}

\begin{proof}[Proof of Proposition~\ref{prop:blockrewards}]
    Assume for contradiction that $( (1+\frac{B'-B}{Q_A \cdot (r-c^W+B)})\cdot \vec{q}^*, \vec{r}^*)$ is not an equilibrium. Then there exists some player $i$ and a deviation $((1+\frac{B'-B}{r-c^W+B})\cdot q_i, r_i)$ that achieves strictly better payoff against $( (1+\frac{B'-B}{r-c^W+B})\cdot \vec{q}_{-i}^*, \vec{r}_{-i}^*)$ than $((1+\frac{B'-B}{r-c^W+B})\cdot q^*_i, r^*_i)$. 

    The payoff of Player $i$ on strategy profile $(\vec{q},\vec{r})$ has three components:
    \begin{itemize}
        \item The revenue of Player $i$ in the simultaneous first-price auction, $R_i(\vec{q},\vec{r})$.
        \item The block reward earned by Player $i$, $B \cdot Q_A \cdot \frac{q_i}{\sum_j q_j}$.
        \item The cost paid by Player $i$, $c_i^R \cdot q_i$.
    \end{itemize}

    Let us make the following observations about similarities of outcomes between various strategy profiles.
    \begin{itemize}
        \item $R_i(\vec{q},\vec{r}) = R_i(x\cdot \vec{q},\vec{r})$ for any $x > 0$. To see this, observe that each Miner wins the same quantity of \append{}s under $\vec{q}$ as $c\cdot \vec{q}$, and therefore setting the same reserves $\vec{r}$ results in the same outcome under simultaneous first-price auctions.
        \item $B\cdot Q_A \cdot \frac{q_i}{\sum_j q_j} = B \cdot Q_A \cdot \frac{x\cdot q_i}{\sum_j x\cdot q_j}$. That is, scaling up all investments by $x > 0$ does not impact the block rewards won.
    \end{itemize}

    We therefore get the following chain of inequalities -- we use a superscript of $B$ on the Payoff to denote that the block reward is $B$.
\begin{align*}
    &P^{B'}_i\left(\left(\left(1+\frac{B'-B}{r-c^W+B}\right)\cdot q_i, r_i\right); \left(1+\frac{B'-B}{r-c^W+B}\right)\cdot \vec{q}_{-i}^*, \vec{r}_{-i}^*\right) - P^{B'}_i\left(\left(1+\frac{B'-B}{r-c^W+B}\right)\cdot \vec{q}^*, \vec{r}^*\right) \\
    &\qquad = R_i\left(\left(\left(1+\frac{B'-B}{r-c^W+B}\right)\cdot q_i, r_i\right); \left(1+\frac{B'-B}{r-c^W+B}\right)\cdot \vec{q}_{-i}^*, \vec{r}_{-i}^*\right) - R_i\left(\left(1+\frac{B'-B}{r-c^W+B}\right)\cdot \vec{q}^*, \vec{r}^*\right)\\
    &\qquad \qquad+ B'\cdot Q_A \cdot \left(\frac{\left(1+\frac{B'-B}{r-c^W+B}\right)\cdot q_i}{\left(1+\frac{B'-B}{r-c^W+B}\right)\cdot \left(q_i + \sum_{j \neq i}q^*_j\right)} - \frac{\left(1+\frac{B'-B}{r-c^W+B}\right)\cdot q_i^*}{\left(1+\frac{B'-B}{r-c^W+B}\right)\cdot \sum_{j }q^*_j}\right)\\
    &\qquad\qquad \qquad - \left(1+\frac{B'-B}{r-c^W+B}\right) \cdot c_i^R \cdot \left(q_i - q^*_i\right)\\
       &\qquad = R_i\left(\left( q_i, r_i\right); \vec{q}_{-i}^*, \vec{r}_{-i}^*\right) - R_i\left( \vec{q}^*, \vec{r}^*\right)\\
    &\qquad \qquad+ B'\cdot Q_A \cdot \left(\frac{q_i}{\left(q_i + \sum_{j \neq i}q^*_j\right)} - \frac{q_i^*}{\sum_{j }q^*_j}\right)\\
    &\qquad\qquad \qquad - \left(1+\frac{B'-B}{r-c^W+B}\right) \cdot c_i^R \cdot \left(q_i - q^*_i\right)\\
    &\qquad = P_i^B\left(\left(q_i,r_i\right);\vec{q}_{-i}^*,\vec{r}_{-i}^*\right) - P_i^B\left(\vec{q}^*,\vec{r}^*\right)\\
    &\qquad\qquad+ (B'-B) \cdot Q_A \cdot  \left(\left(\frac{q_i}{\left(q_i + \sum_{j \neq i}q^*_j\right)} - \frac{q^*_i}{\sum_{j }q^*_j}\right) - \frac{c_i^R \cdot (q_i-q_i^*)}{Q_A \cdot (r-c^W+B)}\right)\\
    &\qquad = P_i^B\left(\left(q_i,r_i\right);\vec{q}_{-i}^*,\vec{r}_{-i}^*\right) - P_i^B\left(\vec{q}^*,\vec{r}^*\right)\\
    &\qquad \qquad + \frac{B'-B}{r-c^W+B} \cdot \left(Q_A \cdot (r-c^W+B) \cdot \left(\frac{q_i}{\left(q_i + \sum_{j \neq i}q^*_j\right)} - \frac{q_i^*}{\sum_{j }q^*_j}\right) - c_i^R \cdot (q_i-q_i^*)\right)\\ 
    &\qquad \leq 0 + \frac{B'-B}{r-c^W+B} \cdot \left(Q_A \cdot (r-c^W+B) \cdot \left(\frac{q_i}{\left(q_i + \sum_{j \neq i}q^*_j\right)} - \frac{q_i^*}{\sum_{j }q^*_j}\right) - c_i^R \cdot (q_i-q_i^*)\right)\\
    &\qquad \leq 0
\end{align*}
Above, the first equality simply expands the definition of $P_i^{B'}$. The second equality observes that $R_i(\vec{q},\vec{r}) = R_i(x\cdot \vec{q},\vec{r})$ for all $x > 0$ (as the quantities won by all miners are the same, and so are the reserves set), and makes algebraic simplifications. The third equality observes that $P_i^B\left(\left(q_i,r_i\right);\vec{q}_{-i}^*,\vec{r}_{-i}^*\right) - P_i^B\left(\vec{q}^*,\vec{r}^*\right) = R_i\left(\left( q_i, r_i\right); \vec{q}_{-i}^*, \vec{r}_{-i}^*\right) - R_i\left( \vec{q}^*, \vec{r}^*\right) + B \cdot \left(\frac{q_i}{\left(q_i + \sum_{j \neq i}q^*_j\right)} - \frac{q_i^*}{\sum_{j }q^*_j}\right)-  c_i^R \cdot \left(q_i - q^*_i\right)$. The fourth equality is basic algebra. The penultimate inequality invokes the assumption that $(\vec{q}^*,\vec{r}^*)$ is an equilibrium with no block reward.

The final inequality holds for the following reason. Observe that $Q_A \cdot (r-c^W+B) \cdot \frac{q_i}{q_i + \sum_{j \neq i} q_j^*} - c_i^R \cdot q_i$ is exactly the payoff of Miner $i$ in a Tullock Contest with total reward $Q_A \cdot (r-c^W+B)$ when using strategy $q_i$ against $\vec{q}_{-i}^*$. Similarly, $Q_A \cdot (r-c^W+B) \cdot \frac{q_i^*}{\sum_j q_j^*} - c_i^R \cdot q_i^*$ is exactly the payoff of Miner $i$ in a Tullock Contest with total reward $Q_A \cdot (r-c^W+B)$ when using strategy $q^*_i$ against $\vec{q}_{-i}^*$. Therefore, if $\vec{q}^*$ is an equilibrium for the Tullock Contest with total reward $Q_A \cdot (r-c^W+B)$ and Resource costs $\vec{c}^R$, it must hold that $Q_A \cdot (r-c^W+B)\cdot \left(\frac{q_i}{\left(q_i + \sum_{j \neq i}q^*_j\right)} - \frac{q_i}{\sum_{j }q^*_j}\right) - c_i^R \cdot (q_i-q_i^*) \leq 0$. Indeed, by Theorem~\ref{thm:bitcoinpure}, any equilibrium of the Distributed Ledger Model with clearing price $r$ must have $\vec{q}^*$ as an equilibrium of the Tullock Contest with total reward $Q_A \cdot (r-c^W+B)$.    
\end{proof}

Finally, we provide a proof of Proposition~\ref{prop:blockrewards2}. Intuitively, the proof of Proposition~\ref{prop:blockrewards2} follows the following outline (below for ease of notation, let $r:= D^{-1}_{\sup}(Q_A)$)
\begin{itemize}
    \item Starting from the proposed market-clearing equilibrium, the payments are according to a Tullock contest with total prize $B+Q_A\cdot (r-c^W)$.
    \item A deviation in investment $q_i$, while still maintaining a clearing price of $r$, induces some loss for the deviating player, equivalent to that of deviating in a Tullock contest with total prize $B+Q_A\cdot (r-c^W)$, which is lower-bounded in Lemma~\ref{lem:TullockDeviation}.
    \item From here, a deviation in reserve \emph{might} increase profit, because perhaps the deviating player increased their quantity of \append{}s from equilibrium. But, the necessary loss in the Tullock contest to enable such a deviation to \emph{possibly} be profitable outweighs the possible profits.
\end{itemize}

The proof first requires a technical lemma which lower bounds the loss in a Tullock contest by deviating from equilibrium.

\begin{lemma}\label{lem:TullockDeviation} Consider a Tullock contest with total reward $Y$ and cost profile $\vec{c}^R$. Then fixing equilibrium investments $\vec{q}_{-i}$, Player $i$ investing an additional $z \cdot Y/c^*(\vec{c}^R)$ from their equilibrium investment $q_i$ decreases payoff by at least $\frac{z^2}{1+z} \cdot Y \cdot \min\{1,c_i^R/c^*(\vec{c}^R)\}$.
\end{lemma}

\begin{proof}
    First, assume that $x_i^*(\vec{c}^R) > 0$. Then, by investing $q_i + z \cdot Y / c^*(\vec{c}^R)$ instead of $q_i = x_i^*(\vec{c}^R)$, Miner $i$'s change in payoff is:

    \begin{align*}
        Y \cdot \frac{x_i^*(\vec{c}^R)+z}{1+z} - Y \cdot x_i^*(\vec{c}^R) - z \cdot Y\cdot c_i^R /c^*(\vec{c}^R) &= Y \cdot \left(\frac{z-z\cdot x_i^*(\vec{c}^R)}{1+z} - z\cdot \frac{c_i^R}{c^*(\vec{c}^R)}\right)\\
        &= Y \cdot \left(z\cdot \left(\frac{(1-x_i^*(\vec{c}^R))}{1+z} - \frac{c_i^R}{c^*(\vec{c}^R)}\right)\right)\\
        &= Y \cdot z \cdot \frac{c_i^R}{c^*(\vec{c}^R)} \cdot \left(\frac{1}{1+z}-1\right)\\
        &= -Y \cdot z^2 \cdot \frac{c_i^R}{c^*(\vec{c}^R)} / (1+z).
    \end{align*}

     Above, the LHS of the first line simply states the increase in Miner $i$'s prize, and subtracts the additional cost. All equalities are basic algebra.

     If instead $x_i^*(\vec{c}^R) = 0$, then Miner $i$'s change in payoff is instead:
\begin{align*}
        Y \cdot \frac{z}{1+z}  - z \cdot Y\cdot c_i^R /c^*(\vec{c}^R)
        &= Y \cdot z \cdot \left(\frac{1}{1+z}-\frac{c_i^R}{c^*(\vec{c}^R)}\right)\\
        &\leq Y \cdot z \cdot \left(\frac{1}{1+z}-1\right)\\
        &= -Y \cdot z^2 /(1+z)\\
    \end{align*}
\end{proof}

\begin{corollary}\label{cor:TullockIncrease} Consider a Tullock contest with total reward $Y$ and cost profile $\vec{c}^R$. Then fixing equilibrium investments $\vec{q}_{-i}$, Player $i$ increasing their market share from the equilibrium market share $x_i^*(\vec{c}^R)$ to $x_i^*(\vec{c}^R)+w$ decreases payoff by at least $w^2\cdot Y/2$.
\end{corollary}
\begin{proof}
The total equilibrium investment is $Y/c^*(\vec{c}^R)$, which means that in order to increase $x_i^*(\vec{c}^R)$ by $w$, Player $i$ must invest at least an additional $a\cdot Y/c^*(\vec{c}^R)$, where $\frac{x_i^*(\vec{c}^R)+a}{1+a} \geq x_i^*(\vec{c}^R)+w$. Solving for $a$, we conclude $a \geq \frac{w}{1-x_i^*(\vec{c}^R)-w}$.

Recall that $x_i^*(\vec{c}^R) = \max\{0,1-c_i^R/c^*(\vec{c}^R)\}$. Next, observe that if $x_i^*(\vec{c}^R) = 0$, then $\frac{w}{1-x_i^*(\vec{c}^R)-w} = \frac{w}{1-w} \geq w/\min\{1,c_i^R/c^*(\vec{c}^R)\}$.

Moreover, observe that if $x_i^*(\vec{c}^R) > 0$, then $\frac{w}{1-x_i^*(\vec{c}^R)-w} = \frac{w}{c_i^R/c^*(\vec{c}^R)-w} \geq w/\min\{1,c_i^R/c^*(\vec{c}^R)\}$.

Now that we have that Player $i$ must invest at least $\frac{w \cdot Y}{c^*(\vec{c}^R)\cdot \min\{1,c_i^R/c^*(\vec{c}^R)\}}$ in order to increase to $x_i^*(\vec{c}^R)+w$, we can simply plug into Lemma~\ref{lem:TullockDeviation} with $z = \frac{w}{\min\{1,c^R_i/c^*(\vec{c}^R)\}}$, and observe that $z^2/(1+z) \geq z^2/2$ whenever $z \leq 1$, and $z^2/(1+z) \geq z/2$ whenever $z \geq 1$.

So, if $w \leq \min\{1,c^R_i/c^*(\vec{c}^R)\}$, we can now conclude a payoff loss of at least $\frac{w^2}{2(\min\{1,c_i^R/c^*(\vec{c}^R)\})^2} Y \cdot \min\{1,c_i^R/c^*(\vec{c}^R)\} \geq z^2 \cdot Y/2$ (the final inequality follows as $\min\{1,c^R_i/c^*(\vec{c}^R)\} \leq 1$. Similarly, if $w \geq \min\{1,c^R_i/c^*(\vec{c}^R)\}$, we can now conclude a payoff loss of at least $\frac{w}{\min\{1,c^R_i/c^*(\vec{c}^R)\}}\cdot Y \cdot \min\{1,c^R_i/c^*(\vec{c}^R)\}/2 = w\cdot Y /2 \geq w^2 \cdot Y/2$ (the final inequality follows as $w \leq 1$ is the amount by which Miner $i$ aims to increase their market share).
\end{proof}

We need one last technical lemma, and then can complete the proof.

\begin{lemma}\label{lem:norevenue}
Let $\vec{q}$ be an investment profile such that Miner $i$'s total \append{}s $Q_i$ satisfies $Q_i \leq \frac{(x - c^W) (Q_A-D(x))}{x - D^{-1}_{\sup}(Q_A)}$ for all $x > D^{-1}_{\sup}(Q_A)$. Then one best response of Miner $i$ in the price-setting game when all other miners set reserve $D^{-1}_{\sup}(Q_A)$ is to set reserve $D^{-1}_{\sup}(Q_A)$ as well.
\end{lemma}
\begin{proof}
If Miner $i$ sets price $D^{-1}_{\sup}(Q_A)$, then the clearing price will be $D^{-1}_{\sup}(Q_A)$, and Miner $i$'s total profit from the price-setting game will be $Q_i \cdot (D^{-1}_{\sup}-c^W)$.

If Miner $i$ instead sets some price $x< D^{-1}_{\sup}(Q_A)$, then the clearing price will be no more than $D^{-1}_{\sup}(Q_A)$, and their profit will therefore be no more than $Q_i \cdot (D^{-1}_{\sup}(Q_A)-c^W)$, and this can never be a strictly better response than $D^{-1}_{\sup}(Q_A)$.

Finally, if Miner $i$ instead sets some price $x > D^{-1}_{\sup}(Q_A)$, then Miner $i$ will definitely be a price-setter, so the clearing price will be $x$. Because Miner $i$ is the price-setter, all $Q_A - Q_i$ items from other miners will definitely clear, and Miner $i$ will sell exactly $D(x) - (Q_A-Q_i)$ of their items. Therefore, Miner $i$'s profit will be $(x-c^W)\cdot (D(x) - Q_A + Q_i)$.

In order to see whether $D^{-1}_{\sup}(Q_A)$ is at least as good a response as $x > D^{-1}_{\sup}(Q_A)$, we can write:
\begin{align*}
    &\qquad \qquad Q_i \cdot (D^{-1}_{\sup}(Q_A)-c^W) \geq (x-c^W)\cdot (D(x) - Q_A + Q_i)\\
    &\Leftrightarrow \qquad (x-c^W)\cdot (Q_A-D(x)) \geq Q_i \cdot (x-D^{-1}_{\sup}(Q_A))\\
    &\Leftrightarrow \qquad \frac{(x-c^W)\cdot (Q_A-D(x))}{x-D^{-1}_{\sup}(Q_A)} \geq Q_i.
\end{align*}

The final line holds by hypothesis, and therefore $D^{-1}_{\sup}(Q_A)$ is a best response.

\end{proof}

\begin{proof}[Proof of Proposition~\ref{prop:blockrewards2}]

Consider now the candidate equilibrium proposed by Theorem~\ref{thm:bitcoinpure}:
\begin{itemize}
    \item $\sum_j q_j = Q_A \cdot (D^{-1}_{\sup}(Q_A)+B-c^W)/c^*(\vec{c}^R)$.
    \item For all Miners $i$, $\frac{q_i}{\sum_j q_j} = x_i^*(\vec{c}^R)$.
    \item $r_i = D^{-1}_{\sup}(Q_A)$ for all $i$. 
\end{itemize}

Consider a Miner $i$ deviating from this strategy profile to some $q'_i, r'_i$., and let $\varepsilon:= Q_A \cdot x_i^*(\vec{c}^R) - \inf_{x > D^{-1}_{\sup}(Q_A)}\{\frac{(x-c^W) \cdot (Q_A - D(x))}{x-D^{-1}_{\sup}(Q_A)}\} > 0$. Recall that the change in reward comes in two parts:
\begin{itemize}
    \item First, Miner $i$ changes $q_i$ to $q'_i$, but keeps $r_i = D^{-1}_{\sup}(Q_A)$. This maintains a clearing price of $D^{-1}_{\sup}(Q_A)$, and therefore the rewards change \emph{exactly as a Tullock contest with total reward $D^{-1}_{\sup}(Q_A)+B-c^W$.} If this deviation causes Miner $i$'s fraction of \append{}s to increase to $x_i^*(\vec{c}^R)+w$, then Corollary~\ref{cor:TullockIncrease} lower bounds this loss as at least $(D^{-1}_{\sup}(Q_A)+B-c^W)\cdot w^2/2$. 
    \item Next, Miner $i$ changes $r_i$ to $r'_i$. If $w \leq \varepsilon$, then $x_i^*(\vec{c}^R) + w \leq \inf_{x > D^{-1}_{\sup}(Q_A)}\{\frac{(x-c^W) \cdot (Q_A - D(x))}{x-D^{-1}_{\sup}(Q_A)}\}$, and Lemma~\ref{lem:norevenue} guarantees that this change cannot possibly improve Miner $i$'s payoff. Therefore, if $w \leq \varepsilon$, Miner $i$ cannot be strictly better responding (because they lose payoff when considering the Tullock contest, and do not gain when considering the price-setting game).
    \item So, the only possible joint better responses are for Miner $i$ to increase their investment in the Upstream game to increase the resulting $Q_i$ by at least $\varepsilon$. Corollary~\ref{cor:TullockIncrease} guarantees that Miner $i$ loses at least $\varepsilon^2\cdot (D^{-1}_{\sup}(Q_A) + B - c^W)/2 \geq B \cdot \varepsilon^2 /2$ by doing so.
    \item At the same time, even when increasing $Q_i$ all the way to $Q_A$, the best that can possibly result from the price-setting game is that Miner $i$ earns the full revenue of a monopolist, which is some finite number $X:=\sup_{x \leq Q_A} \{x \cdot (D^{-1}_{\sup}(x)-c^W)\}$.\footnote{$X$ is finite by hypothesis, stated in Section~\ref{sec:prelim}.}
    \item Therefore, as long as $B\cdot \varepsilon^2/2 \geq X$, the payoff loss from increasing investment in the Tullock contest outweighs any possible gain in the price-setting game, and therefore a block reward of $B \geq 2X/\varepsilon^2$ suffices to guarantee that the candidate equilibrium is in fact an equilibrium.
\end{itemize}

\end{proof}

Finally, we show that Proposition~\ref{prop:blockrewards2} is tight in the sense that the strict inequality cannot be relaxed to a weak inequality.

\subsubsection{Example: Proposition~\ref{prop:blockrewards2} is Tight}\label{sec:bitcointight}

Consider a slight modification of the example from Section~\ref{sec:bitcoinexample}. Let $\delta \in (0, 1)$. $D(\cdot)$ has $D(x) = 1+\delta-x$ for all $x \in [\delta,1+\delta]$, $Q_A = 1$, $c^W = 0$, and $c_i^R = 1$ for exactly $n = 1/\delta$ miners (and no other miners exist).

Then, $\inf_{x > D^{-1}_{\sup}(Q_A)}\{\frac{(x - c^W) (Q_A-D(x))}{x - D^{-1}_{\sup}(Q_A)}\} = \inf_{x > \delta}\{\frac{x \cdot (x-\delta)}{x -\delta}\} = \delta$. Moreover, because $c_i^R = 1$ for all $i$, we have $x_i^*(\vec{c}^R) = \delta$ for all $i$. So we have $x_i^*(\vec{c}^R) = \inf_{x > D^{-1}_{\sup}(Q_A)}\{\frac{(x - c^W) (Q_A-D(x))}{x - D^{-1}_{\sup}(Q_A)}\}$. 

So for any fixed $B$, the only possible market-clearing equilibrium is:
\begin{itemize}
    \item $c^*(\vec{c}^R) = 1/(1-\delta)$.
    \item $\sum_{i=1}^{1/\delta} q_i^R = (B+\delta)\cdot (1-\delta)$
    \item $q^R_i = (B+\delta)\cdot (1-\delta)\cdot \delta$ for all $i$. 
    \item $r_i \leq \delta$ for all $i$. 
\end{itemize}

We show that this is not in fact an equilibrium, for any $B$, by considering tiny deviations that slightly increase investment and then price-set in the price-setting game.

Indeed, let us first compute, as a function of $\varepsilon > 0$, the optimal strategy in the price-setting game for a Miner who has $Q_i = \delta + \varepsilon$ (when all other miners set $r_i \leq \delta$). By setting price $x \in (\delta,1+\delta)$, the miner's payoff would be:
\begin{align*}
    x \cdot (1+\delta-x - (1-\delta-\varepsilon)) &= x \cdot (2\delta +\varepsilon -x)\\
    &=(2\delta+\varepsilon)x - x^2
\end{align*}

This is maximized at $x = \delta + \varepsilon/2$, for a total payoff of $(\delta+\varepsilon/2)^2 = \delta^2 + \varepsilon \delta + \varepsilon^2/4$. This means that if Miner $i$ increases their investment to result in $Q_i = \delta+\varepsilon$, they can improve their revenue in the price-setting game (above what they earn in the price-setting game in equilibrium, $\delta^2$) by $\varepsilon \delta + \varepsilon^2/4$. 

Now, let's see how much it would cost a miner to increase their quantity of \append{}s by $\varepsilon$. The total investment of all other miners is $(B+\delta)\cdot (1-\delta) \cdot \delta (1/\delta -1) = (B+\delta)\cdot (1-\delta)^2$. Therefore, to achieve $Q_i = \delta+\varepsilon$, Miner $i$ would need to invest $q$ such that $\frac{q}{q+(B + \delta)\cdot (1-\delta)^2} = \delta+\varepsilon$. This solves to:
\begin{align*}
    &\qquad \qquad \frac{q}{q+(B + \delta)\cdot (1-\delta)^2} = \delta+\varepsilon \\
    &\Leftrightarrow \qquad q = (\delta + \varepsilon) q + (\delta + \varepsilon)\cdot  (B + \delta) \cdot (1-\delta)^2\\
    &\Leftrightarrow \qquad (1 - \delta - \varepsilon) q = (\delta + \varepsilon) \cdot (B + \delta) \cdot (1-\delta)^2 \\
    &\Leftrightarrow \qquad q = \frac{(\delta + \varepsilon)\cdot (B + \delta) \cdot (1-\delta)^2}{1 - \delta - \varepsilon}
\end{align*}

Putting everything together, this concludes that for any $\varepsilon > 0$, Miner $i$ has a strategy that earns total payoff:

\begin{align*}
    &(\delta+\varepsilon)\cdot B + (\delta+\varepsilon/2)^2 - \frac{(\delta + \varepsilon) \cdot (B + \delta) \cdot (1-\delta)^2}{1 - \delta - \varepsilon} \\
    & = \delta \cdot B +\varepsilon \cdot B + \delta^2 + \delta \varepsilon + \varepsilon^2/4 - \frac{\delta \cdot (B+\delta) \cdot (1-\delta)^2 + \varepsilon \cdot (B+\delta) \cdot (1-\delta)^2}{1-\delta-\varepsilon}\\
&= \delta \cdot B +\varepsilon \cdot B + \delta^2 + \delta \varepsilon + \varepsilon^2/4 - \frac{\delta \cdot (B+\delta) \cdot (1-\delta)^2 - \varepsilon \cdot (B+\delta) \cdot (1-\delta) \cdot \delta + \varepsilon \cdot (B+\delta)\cdot (1-\delta)}{1-\delta-\varepsilon}\\
&= \delta \cdot B +\varepsilon \cdot B + \delta^2 + \delta \varepsilon + \varepsilon^2/4 - \frac{\delta \cdot (B+\delta) \cdot (1-\delta)\cdot (1-\delta - \varepsilon) + \varepsilon \cdot (B+\delta)\cdot (1-\delta)}{1-\delta-\varepsilon}\\
&= \delta \cdot B +\varepsilon \cdot B + \delta^2 + \delta \varepsilon + \varepsilon^2/4 - \delta \cdot (B+\delta) \cdot (1-\delta) - \frac{\varepsilon \cdot (B+\delta)\cdot (1-\delta)}{1-\delta-\varepsilon}\\
&= \left(\delta \cdot B +\delta^2 - (B+\delta) \cdot \delta\cdot (1-\delta) \right) + \varepsilon \cdot \left(\varepsilon /4 +(B+\delta)\cdot \left(1-\frac{1-\delta}{1-\delta-\varepsilon}\right) \right)\\
&= \left(\delta \cdot B +\delta^2 - (B+\delta) \cdot \delta\cdot (1-\delta) \right) + \varepsilon \cdot \left(\varepsilon /4 -(B+\delta)\cdot \frac{\varepsilon}{1-\delta-\varepsilon}\right)\\
&= \left(\delta \cdot B +\delta^2 - (B+\delta) \cdot \delta\cdot (1-\delta) \right) + \varepsilon^2 \cdot \left(1/4 -\frac{B+\delta}{1-\delta-\varepsilon}\right)\\
\end{align*}
In particular, the left term is exactly the payoff in equilibrium, so the right term is exactly the change in utility by deviating. Now consider as $\varepsilon \rightarrow 0$. When $\delta \in (\frac{1}{5}, 1)$, $\frac{1}{4} - \frac{B+\delta}{1 - \delta - \varepsilon}$ is strictly negative for all $B \geq 0$.
Therefore, there is always a sufficiently small $\varepsilon$ such that Miner $i$ would prefer to invest beyond equilibrium to achieve $Q_i=\delta + \varepsilon$ and become a price-setter, and there is no market-clearing equilibrium. This direct analysis works for any $\delta \in (1/5, 1)$ and $B \geq 0$, providing the necessary counterexample.

If desired, we can also extend the above analysis to any $\delta \in (0,1)$. Observe that for any $\delta \in (0,1)$ and $B \geq 1$, the term $\frac{1}{4} - \frac{B+\delta}{1-\delta-\varepsilon}$ is strictly negative. Therefore, this direct analysis shows that there is no market-clearing equilibrium for any $\delta \in (0,1)$ and $B \geq 1$. To extend the analysis to all $B \geq 0$, we can use the contrapositive of Proposition~\ref{prop:blockrewards} -- because there is no market-clearing equilibrium for $B = 1$, there is no market-clearing equilibrium for any $B < 1$.

\section{Asymmetric \writeOp{} Costs} \label{app:bitcoinasym}
Our main results focus on the setting where \writeOp{} costs are identical, as this best captures decentralized payment systems. In this section, we briefly explore the case of asymmetric costs (which would better capture a system with significant MEV that is computationally-demanding to extract) and: (a) extend our results that do not require significant new ideas (but still require updated statements/proofs), (b) highlight the one aspect of our results that would require new ideas to extend.

Specifically, we consider now that each Miner $i$ has a possibly different cost per unit of \writeOp{}, denoted as $c_i^W$. We remind the reader of our main results, and briefly comment on extensions to this asymmetric model.

\begin{itemize}
\item Proposition~\ref{prop:bitcoinundominated} states that it is a dominated strategy for Miner $i$ to set a reserve higher than the price they would set as a monopolist who controls all blockspace. This extends to the asymmetric setting with an identical proof, and simply notes that ``the price they would set as a monopolist who controls all blockspace'' now depends on $c_i^W$. We repeat the statement and proof in Proposition~\ref{prop:asymbitcoinundominated}.
\item Theorem~\ref{thm:bitcoinpure} provides necessary conditions for equilibria to exist, and in particular characterizes: (a) that the resulting market shares in any equilibrium are independent of the ultimate clearing price, (b) the resulting price-setting equilibria that can possibly arise at the market shares determined by (a). With asymmetric costs, (a) no longer holds, but it is still possible to characterize the market shares as a function of the clearing price. Similarly, it is still possible to characterize the possible price-setting equilibria as a function of the market shares. As a result, Theorem~\ref{thm:bitcoinasympure} no longer provides quite as clean a characterization with asymmetric costs, but still provides a useful tool to reason about equilibria.
\item Theorem~\ref{thm:bitcoinsufficient} provides sufficient conditions for a market-clearing equilibrium to exist. Similar conditions suffice for a market-clearing equilibrium to exist with asymmetric costs -- Theorem~\ref{thm:bitcoinasymsufficient} makes necessary changes to the statement and proof. The theorem statement restricts to $B=0$, due to the fact that Proposition~\ref{prop:blockrewards} does not extend to asymmetric costs.
\item Proposition~\ref{prop:blockrewards} states that increasing block rewards: (a) does not change the possible candidate equilibria, and (b) can only cause a candidate equilibrium to become an equilibrium (and cannot cause a candidate equilibrium to no longer be an equilibrium). In the asymmetric setting, (a) no longer holds -- when \writeOp{} costs are asymmetric, the block reward impacts bidders differently, and therefore changes candidate equilibria.
\item Proposition~\ref{prop:blockrewards2} identifies a sufficient condition in order for a sufficiently large block reward to imply a market clearing equilibrium. Proposition~\ref{prop:asymblockrewards2} proposes a similar conclusion with asymmetric \writeOp{} costs, although the proof requires one meaningful additional step (due to the fact that block rewards now impact the market shares at a market-clearing equilibrium). 
\end{itemize}

\subsection{Dominated Strategies with Asymmetric \writeOp{} Costs}

\begin{proposition}
\label{prop:asymbitcoinundominated}

    {Let $r^*_i(D, Q_A, c_i^W):= \arg\max_{r \geq D^{-1}_{\inf}(Q_A)}\{(r-c^W_i)\cdot D(r)\}$. Then for all Miners $i$, all $q_i > 0$, and all $r_i > r^*_i(D, Q_A,c_i^W)$, $(q_i, r^*_i(D, Q_A,c_i^W))$ dominates $(q_i, r_i)$. }
\end{proposition}

\begin{proof}[Proof of Proposition~\ref{prop:asymbitcoinundominated}] 

Observe first that, for any $(\vec{q}_{-i}, \vec{r}_{-i})$, the only difference in Miner $i$'s payoff for using $(q_i, r_i)$ as opposed to $(q_i, r^*_i(D, Q_A, c_i^W))$ is their profit in the simultaneous first-price auctions (their total expenditure on resources is the same, and their block reward is the same). Moreover, Miner $i$ also receives the same quantity of \append{}s. 

Let us now analyze Miner $i$'s revenue under both strategies. Let $S_i(r_i)$ denote the quantity of \writeOp{}s sold when using $(q_i, r_i)$, and $S_i(r_i^*(D,Q_A, c_i^W))$ denote the quantity sold when using $(q_i, r^*_i(D, Q_A, c_i^W))$. Observe first that if $S_i(r_i) = 0$, then certainly $S_i(r^*_i(D, Q_A, c_i^W))\cdot (r^*_i(D, Q_A, c_i^W)-c_i^W)\geq 0 = S_i(r_i) \cdot (r_i-c^W_i)$. If $S_i(r_i) > 0$, then:
\begin{align*}
&\qquad\qquad\qquad S_i(r_i)\leq D(r_i) - Q^{\leq}_{-i}(r^*_i(D,Q_A, c_i^W)),\\
&\qquad\qquad\text{ and } S_i(r^*_i(D,Q_A, c_i^W))\geq D(r^*_i(D, Q_A, c_i^W)) - Q^{\leq}_{-i}(r^*_i(D, Q_A, c_i^W)).\\
&\Rightarrow (r^*_i(D, Q_A, c_i^W)-c^W_i)\cdot S_i(r^*_i(D, Q_A, c_i^W)) - (r_i - c^W_i)\cdot S_i(r_i) \\
&\qquad \geq (r^*_i(D, Q_A, c_i^W)-c_i^W)\cdot \left(D(r^*_i(D, Q_A, c_i^W)) - Q^{\leq}_{-i}(r^*_i(D, Q_A, c_i^W))\right) \\
&\qquad \qquad - (r_i - c^W_i) \cdot \left(D(r_i) - Q^{\leq}_{-i}(r^*_i(D,Q_A, c_i^W))\right)\\
&\qquad = (r^*_i(D, Q_Ac_i^W)-c_i^W)\cdot D(r^*_i(D, Q_A, c_i^W) - (r_i - c_i^W) \cdot D(r_i) \\
&\qquad\qquad + (r^*_i(D, Q_A, c_i^W)-r_i) \cdot Q^{\leq}_{-i}(r^*_i(D,Q_A, c_i^W))\\
&\qquad \geq (r^*_i(D, Q_A, c_i^W)-c_i^W)\cdot D(r^*_i(D, Q_A, c_i^W)) - (r_i - c_i^W) \cdot D(r_i) \\
&\qquad \geq 0.
\end{align*}
Above, the first line follows as, because $S_i(r_i) > 0$, the clearing price is at least $r_i$. Therefore, at most $D(r_i)$ \writeOp{}s are sold, and at least $Q_{-i}^{<}(r_i) \geq Q_i^{\leq}(r^*_i(D, Q_A, c_i^W))$ must be sold to Miners $\neq i$. The second line follows because any Miner $i$ setting a price of $r$ sells quantity at least $\max\{Q_A, D(r)\} - Q_{-i}^{\leq}(r) \geq D(r) - Q_{-i}^{\leq}(r)$. The third inequality and subsequent equality are basic algebra. The penultimate inequality follows as $r_i > r^*_i(D,Q_A, c_i^W)$. {The final inequality follows as $r^*_i(D,Q_A, c_i^W)$ optimizes $(r-c^W_i)\cdot D(r)$ over all $r \geq D^{-1}_{\inf}(Q_A)$.}

To see that $(q_i, r^*_i(D, Q_A, c_i^W)$ may sometimes give strictly larger payoff than $(q_i, r_i)$, consider the case that each other Miner invests $q_j = 0$. Then Miner $i$'s profit from end-users by setting reserve $r^*_i(D, Q_A, c_i^W)$ is $Q_A \cdot ( r^*_i(D, Q_A, c_i^W)-c_i^W) > D^{-1}(r_i)\cdot (r_i - c_i^W)$, which is the profit earned by setting reserve $r_i$. 

    
\end{proof}

\subsection{Necessary Conditions for Equilibria with Asymmetric \writeOp{} Costs} \label{app:bitcoinasympure}
The key difference brought by asymmetric cost per unit of \writeOp{} is that, when a Miner decides how much to invest in the Upstream, they not only need to consider their $c_i^R$, but also their $c_i^W$ relative to the market clearing price $r$ and block reward $B$. We capture the effect of $\vec{c}^W$ in the following definition.
\begin{definition}
    Define $c^*(\vec{c}^R, \vec{c}^W, B, r)$ to be the unique solution to $\sum_{i=1}^n \max\left\{0,1-\frac{c_i^R}{(r+B-c_i^W) \cdot c^*(\vec{c}^R, \vec{c}^W, B, r)}\right\} = 1$.
    Further define $x_i^*(\vec{c}^R, \vec{c}^W, B, r):=\max\left\{0,1-\frac{c_i^R}{(r+B-c_i^W) \cdot c^*(\vec{c}^R, \vec{c}^W, B, r)}\right\}$. 
\end{definition}

\begin{lemma}
    For all $r, B, \vec{c}^R, \vec{c}^W$, there is a unique $x$ satisfying $\sum_{i=1}^n \max\left\{0,1-\frac{c_i^R}{(r+B-c_i^W) \cdot x}\right\} = 1$.
\end{lemma}

\begin{proof}
    For ease of notation let $f(x):=\sum_{i=1}^n \max\left\{0,1-\frac{c_i^R}{(r+B-c_i^W) \cdot x}\right\}$. Without loss of generality, sort the Miners in increasing order of $c_i^R/(r+B-c_i^W)$. The for all $x \leq c_1^R/(r+B-c_1^W)$, $f(x) = 0$. For all $x > c_1^R/(r+B-c_1^W)$, $f(x)$ is strictly increasing (because all terms in the sum are weakly increasing, and the term for $i=1$ is strictly increasing). Finally, $\lim_{x\rightarrow \infty} f(x) = n$. Therefore, there is a unique $x$ with $f(x) = 1$. 
\end{proof}

\begin{theorem} \label{thm:bitcoinasympure}
Let $(\vec{q},\vec{r})$ be an Equilibrium in the Distributed Ledger Model, and let the clearing price for End-Users be $r$. Then:
\begin{itemize}
    \item $\sum_j q_j = Q_A / c^*(\vec{c}^R, \vec{c}^W, B, r)$;
    \item For all Miners $i$, $\frac{q_i}{\sum_j q_j} = x_i^*(\vec{c}^R, \vec{c}^W, B, r)$.
\end{itemize}
Moreover, $r \geq D^{-1}_{\sup}(Q_A)$, and:
\begin{itemize}
    \item If $r = D^{-1}_{\sup}(Q_A)$, then $Q_A \cdot x^*_i(\vec{c}^R, \vec{c}^W, B, r) \leq \frac{(x - c_i^W) (Q_A-D(x))}{x - r}$ for all $i$ and all $x > r$. 
    \item If $r > D^{-1}_{\sup}(Q_A)$, then there is a is a single price-setter $i^*$, who sets a price equal to $r_{i^*}:= \arg\max_{x > D^{-1}_{\sup}(Q_A)}\{(x - c_i^W) \cdot (D(x)+Q_A\cdot x^*_{i^*}(\vec{c}^R, \vec{c}^W, B, r)-Q_A)\}$. 
\end{itemize}
\end{theorem}

We then provide components for proving \Cref{thm:bitcoinasympure}.

\begin{proposition}\label{prop:bitcoinasympricesetting} Let $(\vec{q},\vec{r})$ be an Equilibrium in the Distributed Ledger Model. Then one of the following holds:
\begin{itemize}
    \item Each Miner sells $Q_i := Q_A \cdot q_i/\sum_j q_j$ \writeOp{}s at a clearing price of $D^{-1}_{\sup}(Q_A)$, and $r_i \leq D^{-1}_{\sup}(Q_A)$ for all $i$. Further, $Q_i \leq \frac{(x - c_i^W) (Q_A-D(x))}{x - D^{-1}_{\sup}(Q_A)}$ for all $i$ and all $x > D^{-1}_{\sup}(Q_A)$. 
    \item There is a single price-setter $i$, who sets price $r_i:= \arg\max_{x > D^{-1}_{\sup}(Q)}\{(x-c_i^W) \cdot (D(x)+Q_i-Q)\}$.
\end{itemize}
\end{proposition}

\begin{proof}[Proof of Proposition~\ref{prop:bitcoinasympricesetting}] The proof follows immediately from Proposition~\ref{prop:maineq}, with $c_i = c_i^W$. Indeed, in order for $(q_i,r_i)$ to be a best response to $(\vec{q}_{-i},\vec{r}_{-i})$, it must be that $r_i$ optimizes Miner $i$'s payoff after fixing $q_i,\vec{q}_{-i},\vec{r}_{-i}$. Therefore, in any equilibrium it must hold simultaneously for all $i$ that $r_i$ optimizes Miner $i$'s payoff after fixing $q_i,\vec{q}_{-i},\vec{r}_{-i}$. 

Observe that this condition fixes $\vec{q}$ and therefore $\vec{Q}$, and asks that $r_i$ simultaneously optimize Miner $i$'s payoff in response to $\vec{r}_{-i}$. This is exactly asking for an equilibrium of the price-setting game parameterized by $D(\cdot)$ and $\vec{Q}$, and its equilibria are characterized in Proposition~\ref{prop:maineq}.
\end{proof}

Next, we establish that the market shares of each Miner must match those prescribed by Theorem~\ref{thm:bitcoinasympure}

\begin{proposition}\label{prop:bitcoinasymmustbeTullock}
Let $(\vec{q},\vec{r})$ be an Equilibrium in the Distributed Ledger Model, and let the clearing price for End-Users be $r$. Then:
\begin{itemize}
    \item $\sum_j q_j = Q_A /c^*(\vec{c}^R, \vec{c}^W, B, r)$.
    \item For all Miners $i$, $\frac{q_i}{\sum_j q_j} = x_i^*(\vec{c}^R, \vec{c}^W, B, r)$.
\end{itemize}
\end{proposition}

\begin{proof}[Proof of Proposition~\ref{prop:bitcoinasymmustbeTullock}] 

Recall that the goal of Proposition~\ref{prop:bitcoinasymmustbeTullock} is to propose \emph{necessary} conditions on any equilibrium $(\vec{q},\vec{r})$ of the Distributed Ledger Model. Therefore, it suffices to consider, for example, local optimality conditions (which are necessary, but not sufficient). Let $r$ denote the clearing price at the candidate equilibrium $(\vec{q},\vec{r})$.

To this end, we focus the optimization problem facing a particular Miner $i$, and consider deviations of the following form:
\begin{itemize}
    \item Miner $i$ will only consider deviations that \emph{do not change the clearing price $r$} and \emph{do not change the unsaturated price-setter} (if there is one).
    \item Therefore, if $r = D^{-1}_{\sup}(Q_A)$, we will consider deviations for Miner $i$ from $(q_i,r_i)$ to $(q'_i,0)$ (and all deviations of this form will be considered).\footnote{Observe that if $r = D^{-1}_{\sup}(Q_A)$, it must be the case that all Miners $j$ with $q_j > 0$ have $r_j \leq D^{-1}_{\sup}(Q_A)$, and therefore the clearing price will remain $D^{-1}_{\sup}(Q_A)$ so long as Miner $i$ deviates to some $r_i \leq D^{-1}(Q_A)$ as well.}
    \item If $r > D^{-1}_{\sup}(Q_A)$, then Proposition~\ref{prop:bitcoinasympricesetting} establishes that there is a single price setter $i^*$. In order for $i^*$ to be a price-setter, it must be that $Q_{i^*} > Q_A - D(r)$, and also that $r_i < r$ for all $i\neq i^*$ with $q_i > 0$. Therefore, for Miner $i^*$, we will consider deviations of the form $(q'_{i^*}, r)$ \emph{that still result in $Q'_{i^*} > Q_A - D(r)$}. Observe that there exists a sufficiently small $\delta>0$ such that deviations of the form $(q_{i^*}+x,r)$ satisfies this property for all $x \in (-\delta,\delta)$. For Miner $i\neq i^*$, we will consider deviations of the form $(q'_i, 0)$ that also \emph{still result in $Q'_{i^*} > Q_A - D(r)$}. Observe again that there exists a sufficiently small $\delta > 0$ such that deviations of the form $(q'_i,0)$ satisfy this property for all $q'_i \in [0,q_i+\delta)$. 
    \item In conclusion, we will only ever consider deviations that do not change the clearing price and do not change the unsaturated price-setter (if there is one). However, the above bullets note there is sufficient flexibility in choosing such deviations that local optimality conditions on the choice of $q_i$ must hold.
\end{itemize}

Now, we consider local optimality conditions for deviations of the prescribed type for a particular Miner $i$. The proof essentially breaks into two (interleaved) parts: (a) we repeat calculations identical to those in~\cite{ArnostiW22} for analyzing equilibria of Tullock Contests, and (b) we confirm that the same local optimality conditions must hold for any equilibrium in the Distributed Ledger Model.

So, consider the function $x_i(q_i;q_{-i}):=q_i/(q_i +\sum_{j \neq i} q_j)$, which determines the fraction of the $Q_A$ \append{}s won by Miner $i$ as a function of $q_i$ after fixing $q_{-i}$. We compute (identically to~\cite{ArnostiW22}):

$$\frac{\partial x_i(q_i;q_{-i})}{\partial q_i} =\frac{1}{q_i + \sum_{j \neq i} q_j} - \frac{q_i}{\left(q_i + \sum_{j\neq i} q_j\right)^2}= \frac{1-x_i(q_i;q_{-i})}{q_i + \sum_{j \neq i} q_j}.$$

Now, \emph{in the range where the clearing price remains $r$ and the unsaturated price-seller (if one exists) remains $i^*$}, Miner $i\neq i^*$'s payoff for investing $q_i$ is: $P_i(q_i;q_{-i}):=Q_A \cdot x_i(q_i;q_{-i})\cdot (r-c_i^W+B) - c_i^R \cdot q_i$. Therefore, its derivative \emph{in this range} is:

\begin{align*}
    \frac{\partial P_i(q_i;q_{-i})}{\partial q_i} &= Q_A \cdot (r-c_i^W+B)\cdot \frac{\partial x_i(q_i;q_{-i})}{\partial q_i} - c_i^R\\
    &= Q_A\cdot (r-c_i^W+B) \cdot \frac{1-x_i(q_i;q_{-i})}{q_i + \sum_{j\neq i} q_j} - c_i^R
\end{align*}

If there is a price-setter $i^*$, then Miner $i^*$'s payoff \emph{in the prescribed range} is: $P_{i^*}(q_{i^*};q_{-i^*}):=Q_A \cdot x_{i^*}(q_{i^*};q_{-i^*}) \cdot (r-c_{i^*}^W+B) - (Q_A-D(r))\cdot (r-c_{i^*}^W)- c^R_{i^*} \cdot q_{i^*}$. Therefore, its derivative \emph{in this range} is:

\begin{align*}
    \frac{\partial P_{i^*}(q_{i^*};q_{-i^*})}{\partial q_{i^*}} &= Q_A \cdot (r-c_{i^*}^W+B)\cdot \frac{\partial x_{i^*}(q_{i^*};q_{-i^*})}{\partial q_{i^*}} - c_{i^*}^R\\
    &= Q_A\cdot (r-c_{i^*}^W+B) \cdot \frac{1-x_{i^*}(q_{i^*};q_{-i^*})}{q_{i^*} + \sum_{j\neq i^*} q_j} - c_{i^*}^R
\end{align*}

In particular, \emph{in this range}, the partial derivatives are the same. Now, consider any candidate equilibrium $(\vec{q},\vec{r})$ with clearing price $r$. 

\begin{itemize}
    \item If $i^*$ is a price-setter, then we must have $Q_{i^*} > Q_A - D(r)$ and $r_{i^*} = r$. Moreover, as long as $Q'_{i^*} > Q_A - D(r)$ and $r'_{i^*} = r$, $i^*$ will remain a price-setter and have payoff $P_{i^*}(q_{i^*};q_{-i^*})$ as defined above. Therefore, unless $\frac{\partial P_{i^*}(q_{i^*};q_{-i^*})}{\partial q_{i^*}} = 0$, there exists a sufficiently small $\varepsilon$ such that $(q_i\pm \varepsilon,r)$ is a strictly better response than $(q_i, r)$. We conclude that for any price-setter, $\frac{\partial P_{i^*}(q_{i^*};q_{-i^*})}{\partial q_{i^*}} = 0$ is a necessary condition for $(\vec{q},\vec{r})$ to be an equilibrium.
    \item If $i$ is not a price-setter, then any deviation of the form $(q'_i,0)$ for $q'_i \leq q_i$ maintains both the clearing price and the identity of the price-setter (if one exists), along with any deviation of the form $(q_i+\varepsilon,0)$ for sufficiently small $\varepsilon$. Therefore, it must either hold that (a) $\frac{\partial P_{i}(q_{i};q_{-i})}{\partial q_{i}} = 0$ or (b) $q_i = 0$ and $\frac{\partial P_{i}(q_{i};q_{-i})}{\partial q_{i}} \leq 0$.
    \item Together, we conclude that for all $i$, a necessary condition for $(\vec{q},\vec{r})$ to be an equilibrium is that (a) $\frac{\partial P_{i}(q_{i};q_{-i})}{\partial q_{i}} = 0$ or (b) $q_i = 0$ and $\frac{\partial P_{i}(q_{i};q_{-i})}{\partial q_{i}} \leq 0$. Rewriting the partial derivatives computed above, (a) holds if and only if $x_i(q_i;\vec{q}_{-i}) = 1-\frac{c_i^R \cdot \sum_j q_j}{Q_A \cdot (r+B-c_i^W)}$. (b) holds if and only if $1-\frac{c_i^R \cdot \sum_j q_j}{Q_A \cdot (r+B-c_i^W)} \leq 0$. 
\end{itemize}

Because $\sum_{i=1}^n x_i(\vec{q}) = 1$, we must have: $\sum_{i=1}^n \max\{0,1-\frac{c_i^R \cdot \sum_j q_j}{Q_A \cdot (r+B-c_i^W)}\}=1$. In particular, this means that we must have $\sum_j q_j = Q_A/c^*(\vec{c}^R, \vec{c}^W, B, r)$ as desired, and therefore $x_i(\vec{q}) = x_i^*(\vec{c}^R,\vec{c}^W, B, r)$. 
\end{proof}

\begin{proof}[Proof of Theorem~\ref{thm:bitcoinasympure}]
The proof simply combines Propositions~\ref{prop:bitcoinasympricesetting} and~\ref{prop:bitcoinasymmustbeTullock}.
\end{proof}

\subsection{Sufficient Conditions with Asymmetric \writeOp{} Costs} \label{app:bitcoinasymsufficieint}
\begin{theorem}\label{thm:bitcoinasymsufficient} Let $B = 0$, and consider a potential equilibrium $(\vec{q}^*,\vec{r}^*)$ such that: 
\begin{itemize}
    \item The clearing price is $r = D^{-1}_{\sup}(Q_A)$.
    \item $\sum_j q_j = Q_A / c^*(\vec{c}^R, \vec{c}^W, 0, r)$;
    \item For all Miners $i$, $\frac{q_i}{\sum_j q_j} = x_i^*(\vec{c}^R, \vec{c}^W, 0, r)$.
\end{itemize}

Then:
\begin{itemize}
    \item If $D(\cdot)$ is Regular and $x_i^*(\vec{c}^R, \vec{c}^W, 0, r) \leq 1-\frac{1}{D(0)/Q_A -1}$, for all $i$, then $(\vec{q}^*,\vec{r}^*)$ is an Equilibrium.
    \item Define $k_i(z):=\frac{\lt(D^{-1}_{\sup}(z \cdot Q_A) - c_i^W \rt) \cdot z\cdot Q_A}{\lt(D^{-1}_{\sup}(Q_A) - c_i^W \rt) \cdot Q_A}$. Then $(\vec{q}^*,\vec{r}^*)$ is an Equilibrium if and only if $x_i^*(\vec{c}^R, \vec{c}^W, B, r) \leq 1-\sup_{z \in [0,1]}\left\{\frac{k_i(z)-1}{2\cdot\left(\sqrt{k_i(z)/z}-1\right)}\right\}$ for all Miners $i$.
\end{itemize}
\end{theorem}

We then proceed to prove \Cref{thm:bitcoinasymsufficient}. Throughout this section, we will find a sufficient condition for an equilibrium that clears $Q_A$ \writeOp{}s with a block reward of $0$. 

\begin{definition} For $k(\cdot)$, we say that a quantity $Q$ \emph{$k(\cdot)$-covers} a Demand Curve $D(\cdot)$ and \writeOp{} cost $c_i^W$ if $k(x) \cdot Q\cdot (D^{-1}_{\sup}(Q)-c_i^W) \geq (x\cdot Q) \cdot (D^{-1}_{\sup}(x \cdot Q)-c_i^W)$ for all $x \in (0,1)$, and $k(1) = 1$. 

We say that a quantity $Q$ \emph{exactly $k(\cdot)$-covers} $D(\cdot),c_i^W$ if $k(x) \cdot Q\cdot (D^{-1}_{\sup}(Q)-c_i^W) = (x\cdot Q) \cdot (D^{-1}_{\sup}(x \cdot Q)-c_i^W)$ for all $x \in (0,1)$, and $k(1) = 1$. 
\end{definition}

Intuitively, $Q$ $k(\cdot)$-covers $D(\cdot), c_i^W$ if the total revenue earned selling quantity $Q$ guarantees some fraction of the total revenue that could be earned selling quantity $x\cdot Q$ instead, with the precise coverage required parameterized by $x$. $Q$ exactly $k(\cdot)$-covers $D(\cdot),c_i^W$ if $k(\cdot)$ is the tightest possible coverage. Note that $k(x) \geq x$ for all $x\in (0,1)$, as $D^{-1}_{\sup}(x\cdot Q) \geq D^{-1}_{\sup}(Q)$ for all $x \in (0,1)$.

\begin{lemma}\label{lem:bitcoinregularasym}
    Let $D(\cdot)$ be Regular. Then for all $Q \leq D(0)$, $Q$ $\frac{D(0)/Q-x}{D(0)/Q-1}$-covers $D(\cdot), c_i^W$.
\end{lemma}
\begin{proof}
    Consider the function $R_{c_i^W}(Q):=Q \cdot (D^{-1}(Q)-c_i^W)$. Then $R'_{c_i^W}(Q) = \varphi_D(D^{-1}(Q)) - c_i^W$. Because $D(\cdot)$ is Regular, $R'_{c_i^W}(\cdot)$ is decreasing, and therefore $R_{c_i^W}(\cdot)$ is concave. 

    Observe that $Q = \frac{D(0)/Q-1}{D(0)/Q-x}\cdot x \cdot Q + \frac{1-x}{D(0)/Q-x}\cdot D(0)$. Because $R_{c_i^W}(\cdot)$ is concave:
    \begin{align*}
R_{c_i^W}(Q) &\geq \frac{D(0)/Q -1}{D(0)/Q-x} \cdot R_{c_i^W}(x\cdot Q) + \frac{1-x}{D(0)/Q-x}\cdot R_{c_i^W}(D(0))\\
&\geq \frac{D(0)/Q -1}{D(0)/Q-x} \cdot R_{c_i^W}(x\cdot Q),
    \end{align*}
    confirming that $Q$ $\frac{D(0)/Q-x}{D(0)/Q-1}$-covers $D(\cdot), c_i^W$.
\end{proof}

Lemma~\ref{lem:bitcoinregularasym} allows us to conclude $k(\cdot)$-coverage immediately from the fact that $D(\cdot)$ is Regular, although for most Regular $D(\cdot)$ a tighter bound is possible.

Now, we argue that when $Q_A$ sufficiently-covers $D(\cdot)$, it is an Equilibrium for each miner to set $r_i = D^{-1}_{\sup}(Q_A)$, $\sum_j q_j = Q_A /c^*(\vec{c}^R, \vec{c}^W, B, r)$, and $q_i/\sum_j q_j = x_i^*(\vec{c}^R, \vec{c}^W, B, r)$ for all $i$ (i.e.~the potential equilibrium described by Bullet One in the second half of Theorem~\ref{thm:bitcoinasympure}). For simplicity of notation in the rest of this section, we refer to equilibrium investments as $\vec{q}^*$, the equilibrium quantity of \append{}s won by each miner as $\vec{Q}^*$, the equilibrium fraction of \append{}s won as $\vec{x}^*$ (where $x_i^* = Q_i^*/Q_A$), and $c^*:=c^*(\vec{c}^R, \vec{c}^W, B, r)$. We further use the notation $\rew_i(Q_A):=Q_A\cdot (r- c_i^W)$ to denote the profit when all of \append{}s are sold by Miner $i$.

First, we analyze the investment cost Miner $i$ must pay in order to win a $(1-y)$ fraction of \append{}s against $\vec{q}_{-i}^*$.

For the rest of this section, we will leverage the following observation. For all Miners $i$ with $x_i^* = 0$, we are hoping to show that their best response is to maintain $q_i = 0$. Certainly this holds if their best response is to maintain $q_i = 0$ \emph{even if $\frac{c_i^R}{r-c_i^W}$ were lowered to $c^*$}.\footnote{Recall that $x_i^* = \max\{0,1-\frac{c_i^R}{c^*}\}$, therefore, the cost of all Miners with $x_i^* = 0$ is at least $c^*$.} Therefore, if we can show that $(\vec{q}^*,\vec{r}^*)$ is an Equilibrium \emph{even when all Miners with $x_i^*=0$ have $\frac{c_i^R}{r-c_i^W} = c^*$}, then we will have established that $(\vec{q}^*,\vec{r}^*)$ is an Equilibrium even when non-participating Miners have higher costs. 

\begin{lemma}\label{lem:bitcoincostasym} Let $x^*_i > 0$, or $x_i^* = 0$ and $\frac{c_i^R}{r-c_i^W} = c^*$. Then in order to win $(1-y) \cdot Q_A$ \append{}s, against strategy profile $\vec{q}^*_{-i}$, Miner $i$ must invest $(1/y-1)\cdot (1-x^*_i)^2\cdot \rew_i(Q_A)$.   
If $x_i^* = 0$, then in order to win $(1-y) \cdot Q_A$ \append{}s, against strategy profile $\vec{q}^*_{-i}$, Miner $i$ must invest at least $(1/y-1)\cdot (1-x^*_i)^2\cdot \rew_i(Q_A)$.
\end{lemma}
\begin{proof}[Proof of Lemma~\ref{lem:bitcoincostasym}]
    By definition, the total Resources purchased by Miners $\neq i$ is $Q_A \cdot (1-x^*_i) /c^*$. Therefore, in order to win a $(1-y)$ fraction of the market against these Resources, Miner $i$ must invest such that Miners $\neq i$ Resources become a $y$ fraction of the total Resources. Therefore, Miner $i$ must purchase $(1/y-1)\cdot Q_A \cdot (1-x^*_i)/c^*$ Resources.\footnote{To quickly see that the calculation is correct, observe that this results in a total Resources of $(1/y) \cdot \rew(Q_A)\cdot (1-x^*_i)/c^*$, of which $\rew(Q_A)\cdot (1-x^*_i)/c^*$ is a $y$ fraction.} 

    Moreover, a Miner $i$ pays a cost of $c_i^R$ per Resource, meaning that Miner $i$ must invest $(1/y-1)\cdot Q_A \cdot (1-x^*_i)\cdot \frac{c^R_i}{c^*}$. 
    
    Finally, if $x_i^*>0$, $x_i^* = 1- \frac{c^R_i}{(r-c_i^W)c^*}$ and therefore $\frac{c^R_i}{c^*} = (1-x_i^*)(r-c_i^W)$ (and if $x_i^* = 0$, this holds by hypothesis). So we conclude a total investment of $(1/y-1)\cdot \rew_i(Q_A) \cdot (1-x^*_i)^2$, as desired.
\end{proof}

Now, we observe that Miner $i$'s strategy space consists of the following two decisions (made jointly): (a) pick a price $D^{-1}_{\sup}(x \cdot Q_A)$ to set, (b) pick a fraction $y$ to win $(1-y)\cdot Q_A$ \append{}s. After both choices are made, Miner $i$ earns revenue $(x-y) \cdot Q_A \cdot D^{-1}_{\sup}(x \cdot Q_A)$. Therefore, we get the following lemma:

\begin{lemma}\label{lem:bitcoinpureprofitasym} Let $x_i^*>0$, or $x_i^* = 0$ and $\frac{c_i^R}{r - c_i^W} = c^*$, for all $i$. Then for every strategy $(q_i,r_i)$ that Miner $i$ can use against $(\vec{q}_{-i}^*,\vec{r}_{-i}^*)$, there exists a $z \leq 1$ and $y \leq z$ such that:
$$P_i((q_i,r_i);(\vec{q}_{-i}^*,\vec{r}_{-i}^*)) = \left(1-\frac{y}{z}\right) \cdot (z \cdot Q_A) \cdot (D^{-1}_{\sup}(z \cdot Q_A)-c_i^W) - (1/y-1)\cdot(1-x_i^*)^2 \cdot \rew_i(Q_A).$$

Moreover, for every $z \leq 1$ and $y \leq z$, there exists a strategy guaranteeing Miner $i$ payoff exactly $\left(1-\frac{y}{z}\right) \cdot (z \cdot Q_A) \cdot (D^{-1}_{\sup}(z \cdot Q_A)-c_i^W) - (1/y-1)\cdot(1-x_i^*)^2 \cdot \rew_i(Q_A).$

Therefore, $(\vec{q}^*,\vec{r}^*)$ is an Equilibrium if and only if the function $\left(1-\frac{y}{z}\right) \cdot (z \cdot Q_A) \cdot (D^{-1}_{\sup}(z \cdot Q_A)-c_i^W)  - (1/y-1)\cdot(1-x_i^*)^2 \cdot \rew_i(Q_A)$ is optimized at $z = 1$ and $y = 1-x_i^*$.
    
\end{lemma}

\begin{proof}
    Ultimately, Miner $i$ makes some investment that wins $(1-y)\cdot Q_A$ \append{}s for some $y \in [0,1]$. The total cost of doing so, by Lemma~\ref{lem:bitcoincostasym} is $(1/y-1)\cdot \rew_i(Q_A) \cdot (1-x^*_i)^2$.

    Ultimately, Miner $i$ also sets some price of the form $D^{-1}_{\sup}(z\cdot Q_A)$, for $z \in [0,1]$. This implies that a total quantity of $z \cdot Q_A$ \writeOp{}s are sold, of which $\max\{0,z-y\}\cdot Q_A$ are sold by Miner $i$. Therefore, Miner $i$'s total payoff is $\max\{0,z-y\}\cdot Q_A \cdot (D^{-1}_{\sup}(z \cdot Q_A)-c_i^W) - (1/y-1)\cdot(1-x_i^*)^2 \cdot \rew_i(Q_A)$. 

    If Miner $i$ happens to choose $y \leq z$, this matches the desired form. If not, observe that $\max\{0,z-y\}\cdot Q_A \cdot (D^{-1}_{\sup}(z \cdot Q_A)-c_i^W) = 0 = (y-y) \cdot Q_A \cdot (D^{-1}_{\sup}(y \cdot Q_A) - c_i^W)$. Therefore Miner $i$'s payoff matches the desired form after updating $z:=y$. This completes the proof (after observing that $(z-y) = (1-y/z) \cdot z$).

    To see the ``Moreover' portion of the lemma, simply observe that Miner $i$ can indeed pick any $D^{-1}_{\sup}(z \cdot Q_A)$ as a price to set, and any $y \leq z$ as a fraction of \append{}s to leave for other Miners, inducing the prescribed payoff. 

    To see the `Therefore' portion of the lemma, simply observe that $z = 1$ and $y = 1-x_i^*$ corresponds to $(q_i^*,r_i^*)$.
\end{proof}

From now on, we will use the function $P_i(y,z)$ to denote the payoff $P_i( (q_i,r_i);(\vec{q}_{-i}^*,\vec{r}_{-i}^*))$ of the strategy $(q_i,r_i)$ that wins a $(1-y)$ fraction of \append{}s and sets price $D^{-1}_{\sup}(z\cdot Q_A)$.

\begin{corollary}\label{cor:optimizeasym} Let $Q_A$ $k(\cdot)$-cover $D(\cdot), c_i^W$. Then:
\begin{align*}
    P_i(y,z) &\leq \left(1-\frac{y}{z}\right) \cdot k(z) \cdot Q_A \cdot (D^{-1}(Q_A)- c_i^W) - (1/y-1)\cdot(1-x_i^*)^2 \cdot \rew_i(Q_A)\\
    &= \rew_i(Q_A) \cdot\left( \left(1-\frac{y}{z}\right) \cdot k(z) - (1/y-1)\cdot (1-x_i^*)^2\right) 
\end{align*}

with equality at $z = 1, y = 1-x^*_i$. Therefore, $(\vec{q}^*,\vec{r}^*)$ is an Equilibrium if, for all $i$, the function $\rew_i(Q_A) \cdot\left( \left(1-\frac{y}{z}\right) \cdot k(z) - (1/y-1)\cdot (1-x_i^*)^2\right)$ is optimized at $z = 1$ and $y = 1-x_i^*$.

Moreover, if $Q_A$ exactly $k(\cdot)$-covers $D(\cdot),c_i^W$, then:
\begin{align*}
    P_i(y,z) &= \rew_i(Q_A) \cdot\left( \left(1-\frac{y}{z}\right) \cdot k(z) - (1/y-1)\cdot (1-x_i^*)^2\right) 
\end{align*}
Therefore, $(\vec{q}^*,\vec{r}^*)$ is an Equilibrium if and only if, for all $i$, the function $P_i(y,z)$ is optimized at $z = 1$ and $y = 1-x_i^*$.
\end{corollary}

\begin{proof}[Proof of Corollary~\ref{cor:optimizeasym}]
The proof follows immediately from Lemma~\ref{lem:bitcoinpureprofitasym} after substituting the definition of $k(\cdot)$-cover and exactly $k(\cdot)$-cover, and that $k(z) = 1$.  
\end{proof}

From here, we simply optimize the function provided in Corollary~\ref{cor:optimizeasym}. We begin by optimizing $y$ as a function of $z$.

\begin{lemma}\label{lem:bitcoinoptimizeyasym} For any $z$ and $y \leq z$:
\begin{align*}
&\rew_i(Q_A) \cdot\left( \left(1-\frac{y}{z}\right) \cdot k(z) - (1/y-1)\cdot (1-x_i^*)^2\right) \\
& \leq \rew_i(Q_A) \cdot \left(k(z) - 2(1-x_i^*)\cdot \sqrt{k(z)/z}  +(1-x_i^*)^2\right)
\end{align*}

with equality at $y = (1-x_i^*)\cdot \sqrt{\frac{z}{k(z)}}$ (which in particular implies equality at $y=1$ and $z = 1-x_i^*$).
\end{lemma}

\begin{proof}[Proof of Lemma~\ref{lem:bitcoinoptimizeyasym}]
Simply take the derivative with respect to $y$. We get:
\begin{align*}
    &\frac{\partial \left( \rew_i(Q_A) \cdot\left( \left(1-\frac{y}{z}\right) \cdot k(z) - (1/y-1)\cdot (1-x_i^*)^2\right)\right) }{\partial y}\\
    &= -\frac{k(z)}{z} \cdot \rew_i(Q_A) + (1-x^*_i)^2\cdot \rew_i(Q_A)/y^2
\end{align*}

Observe that the derivative is decreasing in $y$. Therefore, the maximum is achieved when the derivative is $0$. This occurs when:
\begin{align*}
    &-\frac{k(z)}{z} \cdot \rew_i(Q_A) + (1-x^*_i)^2\cdot \rew_i(Q_A)/y^2 = 0\\
    \Rightarrow &y = (1-x_i^*)\cdot \sqrt{\frac{z}{k(z)}}.
\end{align*}

This completes the proof of the core lemma. To see that equality holds at $z=1, y=1-x^*_i$, simply observe that $k(1) = 1$, and therefore the RHS above simplifies when substituting $z=1$ to $(1-x_i^*)$. Therefore, the LHS and RHS in the lemma statement are identical after substituting $z=1$ and $y=(1-x_i^*)$ to both sides.

To see the simplification in the statement, simply substitute $y = (1-x^*_i)\cdot \sqrt{\frac{z}{k(z)}}$ as below:
\begin{align*}
   & \rew_i(Q_A) \cdot\left( \left(1-\frac{y}{z}\right) \cdot k(z) - (1/y-1)\cdot (1-x_i^*)^2\right) \\
    &\qquad \leq \rew_i(Q_A) \cdot\left( \left(1-\frac{(1-x_i^*)\cdot \sqrt{z/k(z)}}{z}\right) \cdot k(z) - \left(\frac{1}{(1-x_i^*)\cdot \sqrt{z/k(z)}}-1\right)\cdot (1-x_i^*)^2\right)\\
    &\qquad = \rew_i(Q_A) \cdot \left(k(z) - (1-x_i^*)\cdot \sqrt{k(z)/z} -(1-x_i^*)\cdot \sqrt{k(z)/z} +(1-x_i^*)^2\right)\\
    &\qquad = \rew_i(Q_A) \cdot \left(k(z) - 2(1-x_i^*)\cdot \sqrt{k(z)/z}  +(1-x_i^*)^2\right)
\end{align*}

\end{proof}

\begin{definition}
    From now on, we define
    \begin{align*}
L_i(z)&= P_i\left((1-x_i^*)\cdot \sqrt{\frac{z}{k(z)}}, z\right) = \rew_i(Q_A) \cdot \left(k(z) - 2(1-x_i^*)\cdot \sqrt{k(z)/z}  +(1-x_i^*)^2\right)
\end{align*}
\end{definition}

\begin{corollary}
    Let $Q_A$ $k(\cdot)$ cover $D(\cdot), c_i^W$. Then $(\vec{q}^*,\vec{r}^*)$ is an equilibrium if $L_i(z)$ is optimized at $z=1$. If $Q_A$ exactly $k(\cdot)$ covers $D(\cdot), c_i^W$, then $(\vec{q}^*,\vec{r}^*)$ is an equilibrium if and only if $L_i(z)$ is optimized at $z=1$. 
\end{corollary}
\begin{proof}
    If $L_i(z)$ is optimized at $z=1$, then we can conclude the following:
\begin{align*}
    P_i(y,z) &\leq \left(1-\frac{y}{z}\right) \cdot k(z) \cdot Q_A \cdot (D^{-1}(Q_A)- c_i^W) - (1/y-1)\cdot(1-x_i^*)^2 \cdot \rew_i(Q_A)\\
    &\leq L_i(z)\\
    &\leq L_i(1)\\
    &= P_i(\vec{q}^*,\vec{r}^*).
\end{align*}
      Above, the first line follows from Corollary~\ref{cor:optimizeasym}, for some $y \leq z \leq 1$. The second line follows from Lemma~\ref{lem:bitcoinoptimizeyasym}. The third line follows by assumption. The fourth line follows by the `with equality' portions of Corollary~\ref{cor:optimizeasym} and Lemma~\ref{lem:bitcoinoptimizeyasym}.

      If we further have that $Q_A$ exactly $k(\cdot)$ covers $D(\cdot),c_i^W$, and $L_i(z)$ is not optimized at $z=1$, we conclude that for whatever $z$ $L_i(z) > L_i(1)$ it holds:
\begin{align*}
    P_i\left((1-x_i^*)\cdot \sqrt{\frac{z}{k(z)}},z\right) &=  L_i(z)\\
    &> L_i(1)\\
    &= P_i(\vec{q}^*,\vec{r}^*).
\end{align*}

Above, the first equality holds by Lemma~\ref{lem:bitcoinoptimizeyasym} and Corollary~\ref{cor:optimizeasym}. The second line follows by assumption that $L_i(z)$ is not optimized at $z=1$. The third line follows by the `with equality' portions of Corollary~\ref{cor:optimizeasym} and Lemma~\ref{lem:bitcoinoptimizeyasym}.
      
\end{proof}

\begin{lemma}\label{lem:bitcoinmain2asym}
    $L_i(z) \leq L_i(1)$ for all $z \in [0,1]$ if and only if $1-x_i^* \geq \max_{z \in [0,1]}\lt\{\frac{k(z)-1}{2\cdot \left(\sqrt{k(z)/z}-1\right)}\rt\}$. Therefore, if for all $i$, $1-x_i^* \geq \max_{z \in [0,1]}\left\{\frac{k(z)-1}{2\cdot \left(\sqrt{k(z)/z}-1\right)}\right\}$, and $Q_A$ $k(\cdot)$-covers $D(\cdot),c_i^W$, $(\vec{q}^*,\vec{r}^*)$ is an Equilibrium.

    If $Q_A$ exactly $k(\cdot)$-covers $D(\cdot),c_i^W$, then $(\vec{q}^*,\vec{r}^*)$ is an Equilibrium if and only if for all $i$, it holds that: $1-x_i^* \geq \max_{z \in [0,1]}\left\{\frac{k(z)-1}{2\cdot \left(\sqrt{k(z)/z}-1\right)}\right\}$.
\end{lemma}
\begin{proof}[Proof of Lemma~\ref{lem:bitcoinmain2asym}]
\begin{align*}
    L_i(1) - L_i(z) &=\rew_i(Q_A) \cdot \left(k(1) - 2(1-x_i^*)\cdot \sqrt{k(1)/1)}+(1-x_i^*)^2\right)\\
    &\qquad + \rew_i(Q_A)\cdot \left(- k(z) +2(1-x_i^*)\cdot \sqrt{k(z)/z} - (1-x_i^*)^2 \right)\\
    &= \rew_i(Q_A) \cdot \left(1-k(z) +2(1-x_i^*) \cdot \left(\sqrt{k(z)/z}-1\right)\right)
\end{align*}

In particular, $L_i(1) - L_i(z) \geq 0$ if and only if:
\begin{align*}
    &1-k(z) +2(1-x_i^*) \cdot \left(\sqrt{k(z)/z}-1\right) \geq 0\\
   \Leftrightarrow & 1-x_i^* \geq \frac{k(z)-1}{2\cdot \left(\sqrt{k(z)/z}-1\right)}
\end{align*}
\end{proof}

\begin{corollary} \label{cor:bitcoin-sufficient-largest-miner-asym}
    Let $D(0) = k \cdot Q_A$. Then as long as $x^*_i \leq 1-\frac{1}{k-1}$, $(\vec{q}^*,\vec{r}^*)$ is an Equilibrium.
\end{corollary}
\begin{proof}
    We simply plug into Lemma~\ref{lem:bitcoinregularasym} and Lemma~\ref{lem:bitcoinmain2asym}. Lemma~\ref{lem:bitcoinregularasym} asserts that $Q_A$ $\frac{k-z}{k-1}$-covers $D(\cdot), c^W_i$. Therefore, Lemma~\ref{lem:bitcoinmain2asym} concludes the desired equilibrium so long as $1-x_i^* \geq \frac{\frac{k-z}{k-1}-1}{2\cdot \left(\sqrt{\frac{k-z}{z\cdot (k-1)}}-1\right)}$ for all $z$. We observe that:
    \begin{align*}
        \frac{\frac{k-z}{k-1}-1}{2\cdot \left(\sqrt{\frac{k-z}{z\cdot (k-1)}}-1\right)} & \leq \frac{\frac{k-z}{k-1}-1}{2\cdot \left(\sqrt{\frac{k-1}{z\cdot (k-1)}}-1\right)}\\
        &=\frac{1}{2(k-1)}\cdot \frac{1-z}{1/\sqrt{z}-1}\\
        &=\frac{1}{2(k-1)}\cdot \frac{(1-\sqrt{z})\cdot (1+\sqrt{z})}{(1-\sqrt{z})/\sqrt{z}}\\
        &=\frac{\sqrt{z}+z}{2(k-1)}\\
        &\leq \frac{1}{k-1}
    \end{align*}
    Above, the first inequality follows as $z \in [0,1]$. The three equalities are basic algebra. The final inequality follows again as $z \in [0,1]$. 
\end{proof}

\begin{proof}[Proof of Theorem~\ref{thm:bitcoinasymsufficient}]
The proof follows immediately from the definition of exactly $k(\cdot)$-covers and Lemma~\ref{lem:bitcoinmain2asym}, and Corollary~\ref{cor:bitcoin-sufficient-largest-miner-asym}
\end{proof}

\subsection{Impact of Block Rewards with Asymmetric \writeOp{} Costs}

We have already noted that Proposition~\ref{prop:blockrewards} does not naturally extend with asymmetric costs, as the block reward impacts the resulting market shares. Still, a variant of Proposition~\ref{prop:blockrewards2} holds, stated and proved below.

\begin{proposition}\label{prop:asymblockrewards2}
Let $\vec{c}^R, D(\cdot), Q_A, c^W$ be such that $Q_A \cdot x_i^*(\vec{c}^R) < \inf_{x > D^{-1}_{\sup}(Q_A)}\left\{\frac{(x-c^W_i) \cdot (Q_A - D(x))}{x-D^{-1}_{\sup}(Q_A)}\right\}$ for all $i$. Then, there exists a sufficiently large $B < \infty$ such that a market-clearing equilibrium exists in the market defined by $\vec{c}^R, D(\cdot), B, Q_A, \vec{c}^W$.
\end{proposition}

In particular, note that the \emph{statement} of Proposition~\ref{prop:asymblockrewards2} is nearly identical to that of Proposition~\ref{prop:blockrewards2} -- the only change is a substitution of $c_i^W$ for $c^W$, and the statement even keeps $x_i^*(\vec{c}^R)$ (which does not depend on $\vec{c}^W$ or $B$). The proof requires one extra step to argue that this is the appropriate condition, even though block rewards can impact market shares. Essentially, the key idea is that as block rewards grow sufficiently large, $x_i^*(\vec{c}^R, \vec{c}^W, B, r)$ approaches $x_i^*(\vec{c}^R)$. 

We first need similar technical lemmas to the proof of Proposition~\ref{prop:blockrewards2}. Let us define $c^* := c^*(\vec{c}^R, \vec{c}^W, B, r), x_i^* := x_i^*(\vec{c}^R, \vec{c}^W, B, r)$, and $\rew_i(Q_A):= (r+B-c_i^W) \cdot Q_A$ for simplicity of notation. Below, 

\begin{lemma}\label{lem:asymTullockDeviation} Consider a game where each player $j$ chooses a $q_j \geq 0$, and Player $i$ receives payoff $\rew_i(Q_A) \cdot q_i / \sum_j q_j - c_i^R \cdot q_i$ where $\rew_i(Q_A) = Q_A \cdot (r + B - c_i^W)$, and the current strategy profile has $\sum_j q_j = Q_A/c^*$ and $q_i/\sum_j q_j = \max\{1-\frac{c_i^R}{(r+B-c_i^W) \cdot c^*},0\}$. Then if Player $i$ invests an additional $z \cdot Q_A /c^*$, their payoff decreases by at least $\frac{z^2}{1+z} \cdot \rew_i(Q_A) \cdot \min\{1, \frac{c_i^R}{c^* \cdot (r+B-c_i^W)}\}$.
\end{lemma}

\begin{proof}
    To more easily match notation with previous proofs, define $x_i^*:=\max\{1-\frac{c_i^R}{(r+B-c_i^W) \cdot c^*},0\}$.

    First, assume that $x^*_i > 0$. Then, by investing $q_i + z \cdot Q_A / c^*$ instead of $q_i = x_i^*\cdot Q_{A} /c^*$, Miner $i$'s change in payoff is:

    \begin{align*}
        \rew_i(Q_A) \cdot \frac{x_i^*+z}{1+z} - &\rew_i(Q_A) \cdot x_i^*- z \cdot Q_A \cdot c_i^R /c^* \\
        =& \rew_i(Q_A) \cdot \left(\frac{z-z\cdot x_i^*}{1+z} - z\cdot \frac{1}{r+B-c_i^W} \cdot \frac{c_i^R}{c^*}\right)\\
        =& \rew_i(Q_A) \cdot \left(z\cdot \left(\frac{1-x_i^*}{1+z} - \frac{c_i^R}{(r+B-c_i^W) \cdot c^*} \right)\right)\\
        =& \rew_i(Q_A) \cdot z \cdot \frac{c_i^R}{c^* \cdot (r+B-c_i^W)} \cdot \left(\frac{1}{1+z}-1\right)\\
        =& - \rew_i(Q_A) \cdot z^2 \cdot \frac{c_i^R}{c^* \cdot (r+B-c_i^W)} / (1+z).
    \end{align*}

     Above, the LHS of the first line simply states the increase in Miner $i$'s prize, and subtracts the additional cost. All equalities are basic algebra.

     If instead $x_i^* = 0$, then Miner $i$'s change in payoff is instead:
    \begin{align*}
        \rew_i(Q_A) \cdot \frac{z}{1+z}  - z \cdot Q_A \cdot c_i^R /c^*
        &= \rew_i(Q_A) \cdot z \cdot \left(\frac{1}{1+z}-\frac{c_i^R}{c^*} \cdot \frac{1}{r+B-c_i^W} \right)\\
        &\leq \rew_i(Q_A) \cdot z \cdot \left(\frac{1}{1+z}-1\right)\\
        &= -\rew_i(Q_A) \cdot z^2 /(1+z)
    \end{align*}
\end{proof}

\begin{corollary}\label{cor:asymTullockIncrease} Consider a game where each player $j$ chooses a $q_j \geq 0$, and Player $i$ receives payoff $\rew_i(Q_A) \cdot q_i / \sum_j q_j - c_i^R \cdot q_i$, and the current strategy profile has $\sum_j q_j = Q_A/c^*$ and $q_i/\sum_j q_j = \max\{0,1-\frac{c_i^R}{(r+B-c_i^W) \cdot c^*}\}$. Then, Player $i$ increasing their market share from $x_i^*$ to $x_i^*+w$ decreases payoff by at least $w^2\cdot \rew_i(Q_A)/2$.
\end{corollary}
\begin{proof}
Again, to match earlier notation more easily, we define $x_i^*:=\max\{0,1-\frac{c_i^R}{(r+B-c_i^W) \cdot c^*}\}$. 

The total current investment is $Q_A/c^*$, which means that in order to increase $x_i^*$ by $w$, Player $i$ must invest at least an additional $a\cdot Q_A/c^*$, where $\frac{x_i^*+a}{1+a} \geq x_i^*+w$. Solving for $a$, we conclude $a \geq \frac{w}{1-x_i^*-w}$.

Recall that $x_i^* = \max\{0,1-\frac{c_i^R}{(r+B-c_i^W) \cdot c^*} \}$. Let us denote $\frac{c_i^R}{(r+B-c_i^W) \cdot c^*}$ as $\tau_i$. Next, observe that if $x_i^* = 0$, then $\frac{w}{1-x_i^*-w} = \frac{w}{1-w} \geq w/\min\{1,\tau_i\}$.

Moreover, observe that if $x_i^* > 0$, then $\frac{w}{1-x_i^*-w} = \frac{w}{\tau_i-w} \geq w/\min\{1,\tau_i\}$.

Now that we have that Player $i$ must invest at least an additional $\frac{w \cdot Q_A}{c^*\cdot \min\{1,\tau_i\}}$ in order to increase to $x_i^*+w$, we can simply plug into Lemma~\ref{lem:asymTullockDeviation} with $z = \frac{w}{\min\{1,\tau_i\}}$, and observe that $z^2/(1+z) \geq z^2/2$ whenever $z \leq 1$, and $z^2/(1+z) \geq z/2$ whenever $z \geq 1$.

So, if $w \leq \min\{1,\tau_i\}$, we can now conclude a payoff loss of at least $\frac{w^2}{2(\min\{1,\tau_i\})^2} \rew_i(Q_A) \cdot \min\{1,c_i^R/c^*\} \geq z^2 \cdot \rew_i(Q_A)/2$ (the final inequality follows as $\min\{1,\tau_i\} \leq 1$. Similarly, if $w \geq \min\{1,\tau_i\}$, we can now conclude a payoff loss of at least $\frac{w}{\min\{1,\tau_i\}}\cdot \rew_i(Q_A) \cdot \min\{1,\tau_i\}/2 = w \cdot \rew_i(Q_A) /2 \geq w^2 \cdot \rew_i(Q_A)/2$ (the final inequality follows as $w \leq 1$ is the amount by which Miner $i$ aims to increase their market share).
\end{proof}

\begin{proof}[Proof of Proposition~\ref{prop:asymblockrewards2}]

Consider now the candidate equilibrium proposed by Theorem~\ref{thm:bitcoinasympure}:
\begin{itemize}
    \item $\sum_j q_j = Q_A / c^*(\vec{c}^R, \vec{c}^W, B, D^{-1}_{\sup}(Q_A))$.
    \item For all Miners $i$, $\frac{q_i}{\sum_j q_j} = x_i^*(\vec{c}^R, \vec{c}^W, B, D^{-1}_{\sup}(Q_A))$.
    \item $r_i = D^{-1}_{\sup}(Q_A)$ for all $i$. 
\end{itemize}

First, we wish to take a sufficiently large $B$ so that each $x_i^*(\vec{c}^R, \vec{c}^W, B, D^{-1}_{\sup}(Q_A))$ is close to $x_i^*(\vec{c}^R)$. Consider the function $f(x):= \sum_{i=1}^n \max\{0, 1 - \frac{c_i^R}{x}\}$ and $g(x):=\sum_{i=1}^n \max\{0, 1-\frac{c_i^R}{(D^{-1}_{\sup}(Q_A)+B-c_i^W)\cdot x}\}$. 

Let $A:=\max_{i\in [n]}\{D^{-1}_{\sup}(Q_A) - c_i^W\}$. Then observe that $\frac{c_i^R}{Bx} \geq \frac{c_i^R}{(D^{-1}_{\sup}(Q_A)+B-c_i^W)\cdot x} \geq \frac{c_i^R}{(A + B) \cdot x}$, for all $i$. Therefore, $f(Bx)\leq g(x) \leq f( ( A+B)x)$ for all $B$. In particular, this means that because $g(c^*(\vec{c}^R, \vec{c}^W, B, D^{-1}_{\sup}(Q_A)) = 1$, then $f(B \cdot c^*(\vec{c}^R, \vec{c}^W, B, D^{-1}_{\sup}(Q_A) ) \leq 1$, and $f((A+B)\cdot c^*(\vec{c}^R, \vec{c}^W, B, D^{-1}_{\sup}(Q_A) )) \geq 1$, and therefore $c^*(\vec{c}^R) \in [B \cdot c^*(\vec{c}^R, \vec{c}^W, B, D^{-1}_{\sup}(Q_A)), (A+B) \cdot c^*(\vec{c}^R, \vec{c}^W, B, D^{-1}_{\sup}(Q_A))]$. 

This implies $x_i^*(\vec{c}^R) \in [\max\{0,1-\frac{c_i^R}{B \cdot c^*(\vec{c}^R, \vec{c}^W, B, D^{-1}_{\sup}(Q_A))}\}, \max\{0, 1-\frac{c_i^R}{(A+B) \cdot c^*(\vec{c}^R, \vec{c}^W, B, D^{-1}_{\sup}(Q_A))}\}]$, whereas $x_i^*(\vec{c}^R, \vec{c}^W, B, D^{-1}_{\sup}(Q_A)) = \max\{0,1-\frac{c_i^R}{(D^{-1}_{\sup}(Q_A) + B - c_i^W) \cdot c^*(\vec{c}^R, \vec{c}^W, B, D^{-1}_{\sup}(Q_A))}\} \in [\max\{0,1-\frac{c_i^R}{B \cdot c^*(\vec{c}^R, \vec{c}^W, B, D^{-1}_{\sup}(Q_A))}\}, \max\{0, 1-\frac{c_i^R}{(A+B) \cdot c^*(\vec{c}^R, \vec{c}^W, B, D^{-1}_{\sup}(Q_A))}\}]$ as well. So our goal is just to show that the interval $[\max\{0,1-\frac{c_i^R}{B \cdot c^*(\vec{c}^R, \vec{c}^W, B, D^{-1}_{\sup}(Q_A))}\}, \max\{0, 1-\frac{c_i^R}{(A+B) \cdot c^*(\vec{c}^R, \vec{c}^W, B, D^{-1}_{\sup}(Q_A))}\}]$ approaches width $0$ as $B \rightarrow \infty$.

To see this, observe simply that if $\frac{c_i^R}{(A+B) \cdot c^*(\vec{c}^R, \vec{c}^W, B, D^{-1}_{\sup}(Q_A))} \geq 1$, then the entire interval is $[0,0]$. If instead, $\frac{c_i^R}{(A+B) \cdot c^*(\vec{c}^R, \vec{c}^W, B, D^{-1}_{\sup}(Q_A))} < 1$, then the right end of the interval is $1-\frac{c_i^R}{(A+B) \cdot c^*(\vec{c}^R, \vec{c}^W, B, D^{-1}_{\sup}(Q_A))}$, and the left end of the interval is at least $1-\frac{c_i^R}{B \cdot c^*(\vec{c}^R, \vec{c}^W, B, D^{-1}_{\sup}(Q_A))} = 1-\frac{A+B}{B} \cdot \frac{c_i^R}{(A+B) \cdot c^*(\vec{c}^R, \vec{c}^W, B, D^{-1}_{\sup}(Q_A))} \geq 1-\frac{c_i^R}{(A+B)\cdot c^*(\vec{c}^R, \vec{c}^W, B, D^{-1}_{\sup}(Q_A))} - \frac{A}{B}$ (because $\frac{c_i^R}{(A+B) \cdot c^*(\vec{c}^R, \vec{c}^W, B, D^{-1}_{\sup}(Q_A))} < 1$). Therefore, the width of the interval is at most $A/B$, which approaches $0$ as $B \rightarrow \infty$. Therefore, because $x^*_i(\vec{c}^R) < \inf_{x > D^{-1}_{\sup}(Q_A)}\left\{\frac{(x-c^W_i) \cdot (Q_A - D(x))}{x-D^{-1}_{\sup}(Q_A)}\right\}$ for all $i$, there exists a sufficiently large $B'$ such that for all $B \geq B'$, $x^*_i(\vec{c}^R, \vec{c}^W, B, D^{-1}_{\sup}(Q_A)) < \inf_{x > D^{-1}_{\sup}(Q_A)}\left\{\frac{(x-c^W_i) \cdot (Q_A - D(x))}{x-D^{-1}_{\sup}(Q_A)}\right\}$ for all $i$ as well. So from now on we will only consider $B > B'$, and the remaining proof will be similar to that of Proposition~\ref{prop:blockrewards2}.

Now, consider a Miner $i$ deviating from this strategy profile to some $q'_i, r'_i$., and let $\varepsilon:= Q_A \cdot x_i^*(\vec{c}^R) - \inf_{x > D^{-1}_{\sup}(Q_A)}\{\frac{(x-c_i^W) \cdot (Q_A - D(x))}{x-D^{-1}_{\sup}(Q_A)}\} > 0$ by the work above. Recall that the change in reward comes in two parts:
\begin{itemize}
    \item First, Miner $i$ changes $q_i$ to $q'_i$, but keeps $r_i = D^{-1}_{\sup}(Q_A)$. This maintains a clearing price of $D^{-1}_{\sup}(Q_A)$, and therefore the rewards change exactly as a \emph{Tullock-ish contest that treats Miner $i$'s reward as $Q_A \cdot (D^{-1}_{\sup}(Q_A)+B-c_i^W)$}. If this deviation causes Miner $i$'s fraction of \append{}s to increase to $x_i^*(\vec{c}^R)+w$, then Corollary~\ref{cor:asymTullockIncrease} lower bounds this loss as at least $Q_A \cdot (D^{-1}_{\sup}(Q_A)+B-c_i^W)\cdot w^2/2$. 
    \item Next, Miner $i$ changes $r_i$ to $r'_i$. As long as $w \leq \varepsilon$, then $x_i^*(\vec{c}^R, \vec{c}^W, B, D^{-1}_{\sup}(Q_A)) + w \leq \inf_{x > D^{-1}_{\sup}(Q_A)}\{\frac{(x-c^W_i) \cdot (Q_A - D(x))}{x-D^{-1}_{\sup}(Q_A)}\}$, and Lemma~\ref{lem:norevenue} guarantees that this change cannot possibly improve Miner $i$'s payoff. Therefore, if $w \leq \varepsilon$, Miner $i$ cannot be strictly better responding (because they lose payoff when considering the Tullock contest, and do not gain when considering the price-setting game).
    \item So, the only possible joint better responses are for Miner $i$ to increase their investment in the Upstream game to increase the resulting $Q_i$ by at least $\varepsilon$. 
    \begin{itemize}
        \item Observe now that our starting point is a game where Miner $i$ receives payoff $Q_A\cdot (B+D^{-1}_{\sup}(Q_A) - c_i^W) \cdot q_i/\sum_j q_j - c_i^R \cdot q_i$, and the equilibrium from which we started satisfies $\sum_j q_j = Q_A / c^*(\vec{c}^R, \vec{c}^W, B, D^{-1}_{\sup}(Q_A))$ and $q_i = \max\{0,1-\frac{c_i^R}{(B+D^{-1}_{\sup}(Q_A) - c_i^W) \cdot c^*(\vec{c}^R, \vec{c}^W, B, D^{-1}_{\sup}(Q_A))}\}$. Therefore, we can apply Corollary~\ref{cor:asymTullockIncrease}. Then indeed Miner $i$ receives payoff $Q_A \cdot (B+D^{-1}_{\sup}(Q_A)-c_i^W) \cdot q_i /\sum_j q_j - c_i^R\cdot q_i$, and the equilibrium we start with satisfies $\sum_j q_j = Q_A / c^*(\vec{c}^R, \vec{c}^W, B, D^{-1}_{\sup}(Q_A))$, and $q_i/\sum_j q_j = x_i^*(\vec{c}^R, \vec{c}^W, B, D^{-1}_{\sup}(Q_A)) = \max\{0,1-\frac{c_i^R}{(B+D^{-1}_{\sup}(Q_A)-c_i^W)\cdot c^*(\vec{c}^R, \vec{c}^W, B, D^{-1}_{\sup}(Q_A))}\}$, which are the necessary conditions to apply Corollary~\ref{cor:asymTullockIncrease}. Applying Corollary~\ref{cor:asymTullockIncrease}, we see that Miner $i$ must lose at least $\varepsilon^2 \cdot Q_A \cdot (B + D^{-1}_{\sup}(Q_A) - c_i^W)/2$ by this deviation.
    \end{itemize}
    \item At the same time, even when increasing $Q_i$ all the way to $Q_A$, the best that can possibly result from the price-setting game is that Miner $i$ earns the full revenue of a monopolist, which is some finite number $X_i:=\sup_{x \leq Q_A} \{x \cdot (D^{-1}_{\sup}(x)-c_i^W)\}$.\footnote{$X$ is finite by hypothesis, stated in Section~\ref{sec:prelim}.}
    \item Therefore, as long as $(B + D^{-1}_{\sup}(Q_A) - c_i^W)\cdot Q_A \cdot \varepsilon^2/2 \geq X_i$, the payoff loss from increasing investment in the Tullock contest outweighs any possible gain in the price-setting game, and therefore a block reward of $B \geq \frac{2X_i}{\varepsilon^2\cdot Q_A}+c_i^W - D^{-1}_{\sup}(Q_A)$ suffices to guarantee that the candidate equilibrium is in fact an equilibrium (after additionally taking $B \geq B'$ so that the first half of the argument works).
\end{itemize}

\end{proof}

\end{document}